\documentclass[11pt,reqno]{amsart}
\usepackage{amsaddr}
\usepackage[left=2cm,right=2cm,top=2cm,bottom=2cm]{geometry}

\usepackage[usenames,dvipsnames]{xcolor}
\usepackage{colortbl}
\usepackage[shortlabels]{enumitem}
\usepackage{amsmath,amsthm,amssymb,bm,bbm,dsfont}
\usepackage{thmtools}
\usepackage{thm-restate}
\usepackage{mathtools}\mathtoolsset{centercolon}
\mathtoolsset{showonlyrefs=true}

\usepackage[mathscr]{eucal}
\usepackage{upgreek}
\usepackage{setspace}
\usepackage{graphicx}
\usepackage{verbatim}
\usepackage{float}
\usepackage{placeins}
\usepackage{array}
\usepackage{booktabs}

\usepackage{threeparttable}
\usepackage[update,prepend]{epstopdf}
\usepackage{multirow}
\usepackage{amsfonts,amssymb,dsfont}
\usepackage[abs]{overpic}

\usepackage{tikz}
\usepackage{tikzscale}
\usepackage{pgfplots}
\usetikzlibrary{calc}
\usetikzlibrary{patterns}
\usetikzlibrary{decorations.pathreplacing}
\usepackage{blox}

\bXLineStyle{-}

\definecolor{sccolor}{RGB}{62, 150, 81}
\definecolor{spcolor}{RGB}{57, 106, 177}

\definecolor{dullmagenta}{rgb}{0.4,0,0.4}   \definecolor{darkblue}{rgb}{0,0,0.4}
\usepackage{hyperref}
\hypersetup{
		unicode=false,          	pdftoolbar=true,        	pdfmenubar=true,        	pdffitwindow=false,     	pdfstartview={FitH},  
		pdfnewwindow=true,      	linktocpage=true,
	colorlinks=true,        	linkcolor=Red,          	citecolor=ForestGreen,  	filecolor=Magenta,      	urlcolor=BlueViolet,    				}

\usepackage{doi}
\usepackage{caption, subcaption}
\usepackage{enumitem}
\usepackage{shadethm}

\makeatletter
\newcommand{\opnorm}{\@ifstar\@opnorms\@opnorm}
\newcommand{\@opnorms}[1]{	\left|\mkern-1.5mu\left|\mkern-1.5mu\left|
	#1
	\right|\mkern-1.5mu\right|\mkern-1.5mu\right|
}
\newcommand{\@opnorm}[2][]{	\mathopen{#1|\mkern-1.5mu#1|\mkern-1.5mu#1|}
	#2
	\mathclose{#1|\mkern-1.5mu#1|\mkern-1.5mu#1|}
}
\makeatother

\makeatletter
\ifx\@NODS\undefined

\else\fi
\makeatother

\makeatletter
\ifx\@NODS\undefined
\let\mathbb=\mathds
\else\fi
\makeatother

\usepackage{xargs}
\usepackage[prependcaption]{todonotes}
\newcommandx{\eric}[2][1=]{\todo[inline, author={Eric}, linecolor=yellow,backgroundcolor=yellow!25,bordercolor=yellow,#1]{#2}}

\newcommandx{\ericside}[2][1=]{\todo[author={Eric}, linecolor=yellow,backgroundcolor=yellow!25,bordercolor=yellow,#1]{#2}}

\DeclarePairedDelimiter{\floor}{\lfloor}{\rfloor}
\DeclareMathOperator*{\argmax}{\arg\max}
\DeclareMathOperator*{\argmin}{\arg\min}

\DeclareMathOperator{\Tr}{Tr}
\DeclareMathOperator{\tr}{Tr}
\DeclareMathOperator{\e}{\mathrm{e}}

\DeclareMathOperator{\B}{\mathbf{B}}

\DeclareMathOperator{\Sym}{Sym}

\DeclareMathOperator{\N}{\mathbb{N}}
\DeclareMathOperator{\one}{\mathds{1}}

\DeclareMathOperator{\Pe}{P_\text{e}}
\newcommand{\Pestar}{P^*_\text{e}}
\DeclareMathOperator{\Ps}{P_\text{s}}

\DeclareMathOperator{\rank}{rank}

\newcommand{\bra}[1]{\langle #1 |}
\newcommand{\ket}[1]{| #1 \rangle}

\newcommand{\be}{{\mathbf e}}

\newcommand{\bx}{{\mathbf x}}
\newcommand{\by}{{\mathbf y}}

\renewcommand{\vec}{\mathbf}

\newcommand{\cH}{{\mathcal{H}}}
\newcommand{\cP}{{\mathcal{P}}}
\newcommand{\cE}{{\mathcal{E}}}
\newcommand{\sC}{{\mathcal{C}}}
\newcommand{\cD}{{\mathcal{D}}}
\newcommand{\cC}{{\mathcal{C}}}

\newcommand{\cX}{{\mathcal{X}}}
\newcommand{\cS}{{\mathcal{S}}}
\newcommand{\cW}{{\mathcal{W}}}
\newcommand{\cR}{{\mathcal{R}}}
\newcommand{\cs}{{\mathcal{W}}}

\def\0{{\mathbf{0}}}
\def\1{{\mathbf{1}}}
\def\2{{\mathbf{2}}}
\def\3{{\mathbf{3}}}
\def\4{{\mathbf{4}}}
\def\5{{\mathbf{5}}}
\def\6{{\mathbf{6}}}

\def\7{{\mathbf{7}}}
\def\8{{\mathbf{8}}}
\def\9{{\mathbf{9}}}

\def\bbR{\mathbb{R}}

\def\be{\begin{equation}}
\def\ee{\end{equation}}
\def\bea{\begin{eqnarray}}
\def\eea{\end{eqnarray}}

\def\eps{\varepsilon}

\def\B{{\mathcal B}}

\theoremstyle{plain}
\newtheorem{theo}{Theorem}
\newtheorem{prop}{Proposition}[section]
\newtheorem{lemm}[prop]{Lemma}
\newtheorem{coro}[theo]{Corollary}
\newshadetheorem{conj}{Conjecture}
\newtheorem*{prop2}{Proposition~\ref{prop:H}}

\newtheorem*{prop5}{Proposition~\ref{prop:E_SW}}
\newtheorem*{prop6}{Proposition~\ref{prop:spCh}}

\theoremstyle{definition}
\newtheorem{defn}[prop]{Definition} 

\theoremstyle{remark}
\newtheorem{remark}{Remark}[section]

\begin{document}
	
\let\origmaketitle\maketitle
\def\maketitle{
	\begingroup
	\def\uppercasenonmath##1{} 	\let\MakeUppercase\relax 	\origmaketitle
	\endgroup
}

\title{\bfseries \Large{  Non-Asymptotic Classical Data Compression with Quantum Side Information  }}

\author{ {Hao-Chung Cheng$^{1,2}$, Eric P. Hanson$^{3,4}$, Nilanjana Datta$^{3}$, Min-Hsiu Hsieh$^1$ }}
\address{\small  	
			$^{1}$University of Technology Sydney, Australia\\
	$^{2}$National Taiwan University, Taiwan (R.O.C.)\\
	$^{3}$Cambridge University, United Kingdom\\
	$^{4}$University of Paris-Sud, France}
 
\date{\today}

\begin{abstract}
In this paper, we analyze classical data compression with quantum side information (also known as the classical-quantum Slepian-Wolf protocol) in the so-called  large and moderate deviation regimes. In the non-asymptotic setting, the protocol involves compressing classical sequences of finite length $n$ and decoding them with the assistance of quantum side information. In the large deviation regime, the compression rate is fixed, and we obtain bounds on the error exponent function, which characterizes the minimal probability of error as a function of the rate.  Devetak and Winter showed that the asymptotic data compression limit for this protocol is given by a conditional entropy. For any protocol with a rate below this 
quantity, the probability of error converges to one asymptotically and its speed of convergence is given by the strong converse exponent function. We obtain finite blocklength bounds on this function, and determine exactly its asymptotic value. In the moderate deviation regime for the compression rate, the latter is no longer considered to be fixed. It is allowed to depend on the blocklength $n$, but assumed to decay slowly to the asymptotic data compression limit. Starting from a rate above this limit, we determine the speed of convergence of the error probability to zero and show that it is given in terms of the conditional information variance.
Our results complement earlier results obtained by Tomamichel and Hayashi, in which they analyzed the so-called small deviation regime of this protocol.
\end{abstract}
	
\maketitle

\section{Introduction} \label{sec:introduction}

Source coding (or data compression) is the task of compressing information emitted by a source in a manner such that it can later be decompressed to yield the original information with high probability. The information source is said to be memoryless if there is no correlation between the successive messages emitted by it. In this case, $n$ successive uses of the source is modeled by a sequence of $n$ independent and identically distributed (i.i.d.) random variables $X_1, X_2, \ldots, X_n$, each taking values $x$ in a finite alphabet ${\mathcal{X}}$ with probability $p(x)$. Such a source is equivalently modeled by a single random variable $X$ with probability mass function $p(x)$, with $x \in {\mathcal{X}}$, and is called a \emph{discrete memoryless source} (DMS). Let $H(X)$ denote the Shannon entropy of $X$. Shannon's Source Coding Theorem \cite{Sha48} tells us that if the messages emitted by $n$ copies of the source are compressed into at least $nH(X)$ bits, then they can be recovered with arbitrary accuracy upon decompression, in the asymptotic limit ($n \to \infty$).

One variant of the above task is that of data compression with classical side information (at the decoder), which is also called {\em{Slepian-Wolf coding}}, first studied by Slepian and Wolf \cite{SW73}. In this scenario, one considers a memoryless source emitting two messages, $x$ and $y$, which can be considered to be the values taken by a pair of correlated random variables $(X, Y)$. The task is once again to optimally compress sequences $\bx:=(x_1,\dotsc,x_n)$ emitted on $n$ copies of the source so that they can be recovered with vanishing probability of error in the asymptotic limit.
However, at the recovery (or decompression) step, the decoder also has access to the corresponding sequence $\by:=(y_1,\dotsc,y_n)$. Since $X$ and $Y$ are correlated, the knowledge of $\by$ gives information about the sequence $\bx$, and thus assists in the decoding. Slepian and Wolf showed that as long as the sequences $\bx$ are compressed into $nR$ bits with $R \geq {nH(X|Y)}$, where $H(X|Y)$ is the conditional entropy of $X$ given $Y$, this task can be  accomplished with vanishing probability of error in the asymptotic limit \cite{SW73}. Here, $R$ is called the \emph{rate} of the protocol and $n$ is called the coding blocklength. In fact, Slepian and Wolf also considered the case in which $Y$ is compressed and sent to the decoder at a rate $R'$; the decoder attempts to faithfully decode both $X$ and $Y$. This gives rise to an achievable rate \emph{region} of pairs $(R,R')$ for which this task is possible. In this work, we do not bound the rate $R'$; that is, we consider the decoder to receive an uncompressed version of $Y$.

In this setting, Slepian and Wolf showed that the data compression limit, that is, the minimal rate of asymptotically lossless compression, is given by $H(X|Y)$. Moreover, Oohama and Han \cite{OH94} established that the data compression limit for Slepian-Wolf coding satisfies the so-called strong converse property. That is, for any attempted compression to ${nR}$ bits with $R< H(X|Y)$, the probability of error converges to $1$ in the  asymptotic limit. This protocol has been extended to countably infinite alphabets and a class of information sources with memory by Cover \cite{cover_proof_1975}, and to various other settings \cite{Wyn75, AK75}.

The above characterization of the data compression limit in terms of the conditional entropy is only valid in the asymptotic limit. It tells us that there exist data compression protocols for which infinite sequences of outputs of a DMS can be compressed and later recovered with vanishing probability of error, but it gives no control on the probability of error incurred for any finite sequence. However, in practical implementations of the protocol one is obliged to consider finite sequences. Hence, it is important to determine the behavior of the optimal error probability in the  non-asymptotic setting (i.e.~finite $n$). To do so, we consider the so-called \emph{reliability function} or {\em{error exponent function}} (see ~\cite{Gal76, CK11} and references therein), which gives the exponential rate of decay of the minimal probability of error achievable by a Slepian-Wolf  protocol, at a fixed rate of compression. On the other hand, one can evaluate the minimum compression rate as a function of the coding blocklength, under the constraint that the error probability is below a certain threshold \cite{NH14, TK14, Tan14}.
 
A quantum generalization of the Slepian-Wolf protocol, which was first introduced by Devetak and Winter \cite{DW03} is the task of classical data compression with quantum side information. They referred to this task as the Classical-Quantum Slepian-Wolf (CQSW) problem. In this protocol, the correlated pair of random variables $(X,Y)$ is replaced by a classical-quantum (c-q) state $\rho_{XB}$. Here $B$ denotes a quantum system which is in the possession of the decoder (say, Bob) and constitutes the quantum side information (QSI), while $X$ is a classical system in the possession of the encoder (say, Alice) and corresponds to a random variable $X$ with probability mass function~$p(x)$, with $x \in {\mathcal{X}}$, as in the classical setting. Such a c-q state is described by an ensemble $\{ p(x), \rho_B^x\}_{x\in\mathcal{X}}$: with probability $p(x)$ the random  variable $X$ takes the value $x$ and Bob's system $B$ is in the state $\rho_B^x$\footnote{Such a state arises when the messages emitted from a quantum DMS are subject to a quantum instrument, and the classical outputs are sent to Alice, and the quantum ones to Bob.
}.  In the so-called \emph{asymptotic, memoryless setting} of CQSW, one considers Alice and Bob to share a large number, $n$, of identical copies of the c-q state $\rho_{XB}$.  Consequently, Alice knows the sequence $\bx := (x_1,x_2, \ldots, x_n)$, whereas the quantum state (i.e.~the QSI)  $\rho_B^{\bx} := \rho_B^{x_1} \otimes \rho_B^{x_2} \ldots \otimes \rho_B^{x_n}$ is accessible only to Bob. However, Bob has no knowledge of the sequence $\bx$. The aim is for Alice to convey the sequence $\bx$ to Bob using as few bits as possible. Bob can make use of the QSI in order to help him decode the compressed message sent by Alice. Devetak and Winter proved that the data compression limit of CQSW, evaluated in the asymptotic limit ($n \to \infty$), is given by the conditional entropy $H(X|B)_\rho$ of the c-q state $\rho_{XB}$.

In this paper we primarily study the CQSW protocol in the {\em{non-asymptotic setting}} in which one no longer takes the limit $n\to\infty$. This corresponds to the more realistic scenario in which only a finite number of copies of the c-q state $\rho_{XB}$ are available. First, we focus on the so-called \emph{large deviation regime}\footnote{That is, the regime in which the compression rate deviates from the data compression limit by a constant amount. We refer the readers to \cite{CH17,CTT2017,RD18} for details of different deviation regimes. }, in which the compression rate $R$ is fixed, and we analyze the optimal probability of error as a function of blocklength $n$. Specifically, in the range $R>H(X|B)_\rho$, we obtain upper and lower bounds on the error exponent function (see Theorems~\ref{theo:large_ach} and \ref{theo:sp_SW}). 
The lower bound shows that for any $R> H(X|B)_\rho$ the CQSW task can be accomplished with a probability of error which decays to zero exponentially in $n$. The upper bound puts a limit on how quickly the probability of error can decay. 
We term this upper bound the ``sphere-packing bound'' for CQSW, since it is analogous to the so-called sphere-packing bound obtained in c-q channel coding~\cite{Dal13, DW14, CHT17}.

For any protocol with a rate $R<H(X|B)_\rho$, the probability of error converges to one asymptotically and its speed of convergence is given by the strong converse exponent function. We obtain finite blocklength bounds on this function (see Theorems~\ref{thm:SC-converse-bound} and \ref{thm:SC-achiev-bound}), and determine exactly its asymptotic value (see Corollary~\ref{cor:exact_sc}), in terms of the \emph{sandwiched conditional R\'enyi entropy} \cite{MDS+13, WWY14, Tom16}. A non-asymptotic study of CQSW in the strong converse domain was also carried out by Tomamichel~\cite{tom-thesis}, and by Leditzky, Wilde, and Datta~\cite{leditzky_strong_2016}. In these works, one-sided bounds were obtained, and hence the asymptotic value of the strong converse exponent was not determined.

The bounds we obtain are expressed in terms of certain entropic exponent functions involving conditional R\'enyi entropies. To derive these results, we prove and employ properties of these functions. In obtaining the strong converse bounds, we employ  variational representations for certain auxiliary exponent functions by making use of those for the so-called log-Euclidean R\'enyi relative entropies developed in \cite{MO17}. Our variational representations are analogous to those obtained by Csisz{\'a}r and K{\"o}rner in the classical setting \cite{CK80, CK81, Csi82, CK11}.

We also study the trade-offs between the rate of compression, the minimal probability of error, and the blocklength $n$. Specifically, we characterize the  behaviors of the error probability and the compression rate in the  \emph{moderate deviation regime}. In contrast to the previously discussed results for which the rate $R$ was considered to be fixed, here we allow the rate to change with $n$, approaching $H(X|B)_\rho$ slowly (slower than $\frac{1}{\sqrt{n}}$), from above.
In this case, we show that the probability of error vanishes asymptotically.
In addition, we obtain an asymptotic formula describing the minimum compression rate which  converges to $H(X|B)_\rho$ when the probability of error decays sub-exponentially in $n$.
We summarize the error behaviors of different regimes in Table~\ref{table:finite}.

\begin{table}[th]
	\centering
	\resizebox{0.85\columnwidth}{!}{
				\begin{tabular}{|c|c|c|} 			\toprule
						Different Regimes & Concentration Phenomena& Slepian-Wolf Coding \\ 			\midrule
			\midrule
			\multirow{2}{*}{Small Deviation} & \multirow{2}{*}{$\Pr\left( S_n \geq \sqrt{n} x \right) \rightarrow 1 - \Phi\left(\frac{x}{\sqrt{v}}\right)$} & \multirow{2}{*}{$ \Pestar\left(n,H(X|B)_\rho + \frac{A}{\sqrt{n}}\right) \rightarrow \Phi\left(\frac{A}{\sqrt{V }}\right)$} \\ 			&& \\
			
			\hline

			\multirow{2}{*}{Moderate Deviation} &  \multirow{2}{*}{ $\Pr\left( S_n \geq n a_n x \right) = \mathrm{e}^{ -\frac{ na_n^2 }{2v} x + o(n a_n^2)}$} 
			&\multirow{2}{*}{$\Pestar(n,H(X|B)_\rho+a_n ) = \mathrm{e}^{- \frac{n a_n^2}{2V} + o(n a_n^2) }$}
			\\ 			&& \\
			
			\hline
			
			Large Deviation & \multirow{2}{*}{ $ \Pr\left( S_n \geq n x \right)= \mathrm{e}^{ -n \Lambda^*(x) + o(n) }$} 
			& \multirow{2}{*}{$\Pestar(n,R) = \mathrm{e}^{-n E (R) + o(n) }$}
			\\ 			($R>H(X|B)_\rho$) && \\
			
			\hline
			
			Large Deviation & \multirow{2}{*}{ $ \Pr\left( S_n \geq n x \right)= \mathrm{e}^{ -n \Lambda^*(x) + o(n) }$} 
			& \multirow{2}{*}{$1 - \Pestar(n,R) = \mathrm{e}^{-n E_\text{sc}^* (R) + o(n) }$}
			\\ 			($R<H(X|B)_\rho$) && \\
			
			\bottomrule
			
					\end{tabular}
	}
	\vspace{1em}
	\caption{The comparison of different regimes of Slepian-Wolf coding with quantum side information. Here, $\Pestar(n,R)$ denotes the optimal error probability with rate $R$ and blocklength $n$;
	$V:= V(X|B)_\rho$ is the conditional information variance defined in Eq.~\eqref{eq:V_cond};$(a_n)_{n\geq 0}$ is an arbitrary sequence such that $\lim_{n\to+\infty} a_n = 0$ and $\lim_{n\to+\infty} \sqrt{n} a_n = +\infty$. 
		The quantity $\Lambda^*$ is the Legendre-Fenchel transform of the cumulant generating function of the random variable, and $\Phi$ is  the  cumulative distribution function of a  standard normal  distribution. The strong converse exponent $E_\text{sc}^*(R)$ is defined in Eq.~\eqref{eq:def_Esc}.
		Note that determining the error exponent $E(R)$ in Slepian-Wolf coding with quantum side information is still an open problem.	}	\label{table:finite}

\end{table}		

\subsection{Prior Works} \label{sec:prior}

Renes and Renner \cite{RR12} analyzed the protocol in the so-called one-shot setting (which corresponds to the case $n=1$)
for a given threshold ($\varepsilon$, say) on the probability of error. They proved that in this case the classical
random variable $X$ can be compressed to a number of bits
given by a different entropic quantity, the so-called {\em{smoothed conditional max-entropy}}, the smoothing parameter being dependent on $\varepsilon$. They also established that this entropic quantity gives
the minimal number of bits, up to small additive quantities involving $\varepsilon$. More precisely, the authors established upper and lower bounds on the minimal number of bits in terms of the {{smoothed conditional max-entropy}}. The asymptotic result of Devetak and Winter could be recovered from their results by replacing the c-q state $\rho_{XB}$ by its $n$-fold tensor power $\rho_{XB}^{\otimes n}$ in these one-shot bounds, dividing by $n$, and taking the limit $n \to \infty$.
In \cite{TH13}, the authors improved these bounds and established a second order expansion of the so-called minimum code size\footnote{That is, the minimum number of bits required to accomplish the task with blocklength $n$, subject to the constraint that the probability error is at most $\eps$.} given an error $\eps$:
\begin{equation*} 
m^*(n,\eps) = nH(X|B)_\rho - \sqrt{nV(X|B)_\rho} \Phi^{-1}(\eps) + O \left(\log n\right),
\end{equation*}
where $V(X|B)_\rho$ is the quantum conditional information variance, and $\Phi$ is  the  cumulative distribution function of a  standard normal  distribution.

As regards the strong converse regime, Tomamichel analyzed the maximal success probability (as opposed to the average success probability considered in this work), and derived a lower bound on the strong converse exponent in terms of a min-conditional entropy in Section 8.1.3 of \cite{tom-thesis}.  Later, Leditzky, Wilde, and Datta determined two lower bounds on the strong converse exponent, which are given in Theorem 6.2 of \cite{leditzky_strong_2016}. The first one bounds the strong converse exponent in terms of a difference of (Petz-) R\'enyi relative entropies, whereas the second one gives a bound in terms of the conditional sandwiched R\'enyi relative entropy. In fact, the second bound can be written as
\begin{equation}
sc(n,R) \geq \sup_{\alpha > 1}\frac{1-\alpha}{2\alpha} (R - H_\alpha^{*,\uparrow} (X|B)_\rho)
\end{equation}
where $sc(n,R)$ is the strong converse exponent for blocklength $n$ and rate $R$ (defined in~\eqref{eq:def_sc}), and $H_\alpha^{*,\uparrow} (X|B)_\rho$ is the conditional sandwiched $\alpha$-R\'enyi relative entropy of the source state $\rho_{XB}$ (defined in \eqref{eq:cond_ent}). This bound is weaker by a factor of $\frac{1}{2}$ to that given in Theorem~\ref{thm:SC-converse-bound}.
We note that the first bound is also weaker than Theorem~\ref{thm:SC-converse-bound}. We refer  the readers to Remark~\ref{remark_sc} for the discussion.

This paper is organized as follows. We introduce the CQSW protocol in Sec.~\ref{sec:protocol}, and state our main results in Sec.~\ref{ssec:main}. The notation and definitions for the entropic quantities and exponent functions are described in Sec.~\ref{sec:pre}.
Sec.~\ref{sec:Large} presents the error exponent analysis for CQSW as $R>H(X|B)_\rho$ (large deviation regime), and we study the optimal success exponent as $R< H(X|B)_\rho$ (strong converse regime)  in Sec.~\ref{sec:SC}.
In Sec.~\ref{sec:mod_dev_reg} we discuss the moderate deviation regime.
We conclude this paper in Sec.~\ref{sec:conclusions} with a discussion.

\section{Classical Data Compression with Quantum Side Information (Slepian-Wolf Coding)}\label{sec:protocol}

Suppose Alice and Bob share multiple (say $n$) identical copies of a classical-quantum (c-q) state 
\begin{equation} \label{eq:rho_XB-cq}
\rho_{XB} = \sum_{x\in \cX} p(x)\ket{x}\bra{x}\otimes \rho_B^x
\end{equation}
where $\cX$ is a finite alphabet and $\rho_B^x$ is a quantum state, of a system $B$ with Hilbert space $\cH_B$, in Bob's possession. The letters $x\in \cX$ can be considered to be the values taken by a random variable $X$ with probability mass function $p(x)$. One can associate with $X$ a quantum system (which we also refer to as $X$) whose Hilbert space has an orthonormal basis labeled by $x\in \cX$, i.e. $\{\ket{x}\}_{x\in \cX}$.

The aim of classical-quantum Slepian-Wolf (CQSW) coding is for Alice to convey sequences $\bx = (x_1,\dotsc,x_n)\in \cX^n$ to Bob using as few bits as possible; Bob can employ the corresponding quantum state $\rho^{\bx}_{B^n} = \rho^{x_1}_B\otimes \dotsm \otimes \rho^{x_n}_B$ which is in his possession, and plays the role of \emph{quantum side information} (QSI), to help decode Alice's compressed message.

Alice's encoding (compression) map is given by $\cE: \cX^n\to \cW$, where the alphabet $\cW$ is such that  $|\cW| < |\cX|^n$.
If Alice's message was $\bx$, the compressed message that Bob receives is $\cE(\bx)\in \cW$. He applies a decoding map $\cD$ on the pair $(\cE(\bx), \rho^{\bx}_{B^n})$ in order to infer Alice's original message. Thus, Bob's decoding is given by a map $\cD: \cs \times \mathcal{S}(B^n) \to  \cX^n$
where $\cS(B^n)$ denotes the set of states on $\cH_B^{\otimes n}$.

If we fix the first argument as $w\in \cs$, we have that the decoding $\cD(w,\cdot)$ is a map from $\mathcal{S}(B^n) \to \cX^n$ which is given by a positive operator-valued measurement (POVM). Thus, we can represent the decoding by a collection of POVMs $\{\cD_w\}_{w\in \cs}$, where $\cD_w = \{\Pi_{\bx}^{(w)}\}_{\bx\in \cX^n}$ with $\Pi_{ \bx}^{(w)}\geq 0$ and $\sum_{ \bx\in \cX^n} \Pi_{ \bx}^{(w)} = \one$,  for each $w\in \cs$.  That is, if Alice sends the message $\bx$, Bob receives $\cE(\bx)$, and measures the state $\rho_{B^n}^\bx$ with the POVM $\{\Pi_{\bx}^{(\cE(\bx))}\}_{\bx\in \cX^n}$. We depict the protocol in Figure~\ref{fig:SW-qsi}.

\begin{figure}[th]
	\centering
	\includegraphics{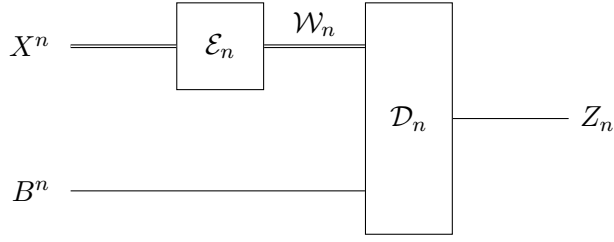}
	\caption{We are given $n$ copies of a classical source $X$ which is correlated with a quantum system $B$. We  compress the source into $\cW_n$  via the encoding $\cE_n$, and then perform a decompression via $\cD_n$ which has access to the side information $B^n$. This yields the output $Z_n$ with associated alphabets $\mathcal{Z}^n$. } \label{fig:SW-qsi}
	
\end{figure}

Given $n\in \N$ and $R>0$, an encoding-decoding pair $(\cE,\cD)$ of the form described above is said to form an $(n,R)$-code if $|\cW| = 2^{nR}$ (or, more precisely, $|\cW| = \lceil 2^{nR} \rceil$). Here, $R$ is called the \emph{rate} of the code $\sC = (\cE,\cD)$. For such a code, the probability of error is given by
\begin{equation}
\Pe(\sC)\equiv \Pe(\rho_{XB},\sC) = 1 - \sum_{\bx \in \cX^n} p(\bx) \tr[ \Pi_{\bx}^{(\cE(\bx))} \rho_{B^n}^{\bx} ],
\end{equation}
where $p(\bx) = p(x_1)\dotsm p(x_n)$ for $\bx = (x_1,\dotsc, x_n)$. We can also consider a random encoding which maps $\bx$ to $w$ with some probability $P(w|x)$. In this case, the probability of error is given by
\begin{equation}
\Pe(\sC) = 1 - \sum_{\bx \in \cX^n, w \in \cW} P(w|\bx) p(\bx) \tr[ \Pi_{\bx}^{(w)} \rho_{B^n}^{\bx} ].
\end{equation}
Alternatively, we can see the random encoding $\cE$ as applying a deterministic encoding $\cE_j$ with some probability $Q_j$. Then for a code $\sC = (\cE, \cD)$,
\begin{equation} \label{eq:Pe_rand_as_det}
\Pe((\cE, \cD)) = 1 - \sum_j Q_j \sum_{\bx \in \cX^n} P(w|\bx) p(\bx) \tr[ \Pi_{\bx}^{(\cE_j)} \rho_{B^n}^{\bx} ] = \sum_j Q_j \Pe((\cE_j,\cD)).
\end{equation}
Thus, the error probability for a random encoding is an average of error probabilities of deterministic encodings. In particular, $\min_j \Pe((\cE_j, \cD)) \leq \Pe((\cE, \cD))$, so the optimal error probability is achieved for a deterministic code.

 The optimal (minimal) rate of data compression evaluated in the asymptotic limit ($n\to \infty$), under the condition that the probability of error vanishes in this limit is called the \emph{data compression limit}. Devetak and Winter~\cite{DW03} proved that it is given by the conditional entropy of $\rho_{XB}$:
\begin{equation}
H(X|B)_\rho = H(\rho_{XB}) - H(\rho_B)
\end{equation}
where $H(\omega) := - \tr (\omega \log \omega )$ denotes the von Neumann entropy of a state $\omega$.

In this paper, we analyze the Slepian-Wolf protocol primarily in the non-asymptotic scenario (finite $n$). The two key quantities that we focus on are the following. The \emph{optimal error probability} for a rate $R$ and blocklength $n$ is defined as
\begin{equation} \label{def:optPe}
\Pestar(n,R)\equiv \Pestar(\rho_{XB},n,R) := \inf \{ \Pe(\sC) : \sC \text{ is an }(n,R)\text{-code for }\rho_{XB}\}.
\end{equation}
Similarly, for any $\eps\in(0,1)$, we define the \emph{optimal rate} of compression at an error threshold $\eps$ and blocklength $n$ by
\begin{equation}
R^*(n,\eps) := \inf \{ R : \exists\, \text{ an }(n,R)\text{-code } \sC \text{ with } \Pe(\sC) \leq \eps\}.
\end{equation}
In particular, we obtain bounds on the \emph{finite blocklength error exponent}
\begin{equation} \label{eq:def_e}
e(n,R) := - \frac{1}{n}\log \Pestar (n,R)
\end{equation}
and the \emph{finite blocklength strong converse exponent}
\begin{equation} \label{eq:def_sc}
sc(n,R) := - \frac{1}{n}\log(1- \Pestar (n,R)).
\end{equation}
In terms of $\Pestar(n,R)$, Devetak and Winter's result can be reformulated as
\begin{equation}  \label{eq:pr-error-DW}
\begin{split}
& \forall R > H(X|B)_\rho: \quad \limsup_{n\to\infty} \Pestar(n,R) = 0,\\
& \forall R < H(X|B)_\rho: \quad  \liminf_{n\to\infty}\Pestar(n,R) > 0.
\end{split}
\end{equation}

 Hence, $H(X|B)_\rho$ is called the \emph{Slepian-Wolf limit}. We may illustrate this result by Figure~\ref{fig:pr-succ-diag} below.
\begin{figure}[ht]
\centering
\includegraphics{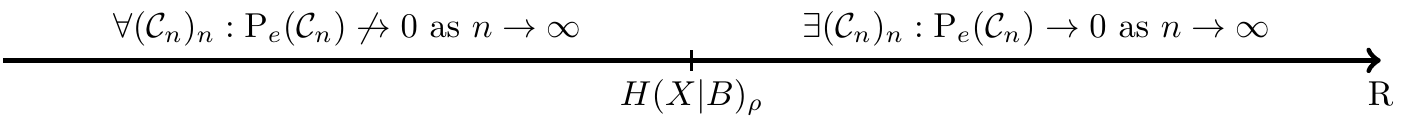}
\caption{Diagram depicting the relationship between the optimal error probability of $(n,R)$-Slepian-Wolf codes, and the rate of compression $R$, as established in \cite{DW03}. 
For rates $R> H(X|B)_\rho$, the task can be accomplished with probability of error tending to $0$ as the number of copies $n$ of the state $\rho_{XB}$ tend to infinity. In contrast, for any rate $R<H(X|B)_\rho$, the probability of error does not converge to $0$ as $n\to\infty$. \label{fig:pr-succ-diag}}
\end{figure}

\section{Main Results} \label{ssec:main}

The main contributions of this work consist of a refinement of \eqref{eq:pr-error-DW}.
 We derive bounds on the speed of convergence of  $\Pestar(n,R)$ to zero for any $R> H(X|B)_\rho$. Further, for $R<H(X|B)_\rho$ we obtain bounds on the strong converse $sc(n,R)$, and determine its exact value in the asymptotic limit. In addition, we analyze the asymptotic behavior of $\Pestar(n,R)$ and $R^*(n,\eps)$ in the so-called
 moderate deviations regime. These results are given by the following theorems, in each of which $\rho_{XB}$ denotes a c-q state (eq.~\ref{eq:rho_XB-cq}), with $H(X|B)_\rho > 0$.

Given a rate $R > H(X|B)_\rho$, there exists a sequence of codes $\sC_n$ such that the probability of error tends to zero as $n\to \infty$, as shown by \eqref{eq:pr-error-DW}. In fact, this convergence occurs exponentially quickly with $n$, and the exponent can be bounded from below and above, as we show in Theorems~\ref{theo:large_ach} and \ref{theo:sp_SW} respectively.

\begin{restatable}{theo}{largeach}
\label{theo:large_ach}
For any rate $R \geq H(X|B)_\rho$, and any blocklength $n\in\mathbb{N}$, the finite blocklength error exponent defined in \eqref{eq:def_e} satisfies
\begin{equation}
e (n,R) \geq {E}_\textnormal{r}^{\downarrow}(R) - \frac{2}{n},
\end{equation}
where 
\begin{align}
   {E}_\textnormal{r}^{\downarrow}(R) \equiv {E}_\textnormal{r}^{\downarrow}(\rho_{XB},R) := \sup_{ \frac12 \leq \alpha\leq 1} \frac{1-\alpha}{\alpha} \left( R - H_{2-\frac{1}{\alpha}}^\downarrow(X|B)_\rho \right),
\end{align}
$H_\alpha^\downarrow(X|B)_\rho := -D_\alpha(\rho_{XB}\|\mathds{1}_X\otimes \rho_B)$, and $D_\alpha$ is the $\alpha$-R\'enyi divergence: $D_\alpha(\rho\|\sigma) = \frac{1}{\alpha-1} \log  \Tr \left[ \rho^\alpha \sigma^{1-\alpha} \right]$.
\end{restatable}

\noindent The proof of Theorem~\ref{theo:large_ach} is in  Section~\ref{subsec:large_achiev}.

\begin{restatable}[Sphere-Packing Bound for Slepian-Wolf Coding]{theo}{spSW} \label{theo:sp_SW}
 Let $R \in ( H(X|B)_\rho, H_0^\uparrow(X|B)_\rho )$. Then, there exist $N_0, K\in\mathbb{N}$, such that for all $n\geq N_0$, the finite blocklength error exponent defined in \eqref{eq:def_e} satisfies
\begin{align*} e (n,R) \leq E_\textnormal{sp} (R) + \frac12\left(1+ \left|\frac{\partial E_\textnormal{sp} (r)}{\partial r}\Big|_{r=R}\right| \right) \frac{\log n}{n} + \frac{K}{n}, 
\end{align*}
where
\begin{align}
{E}_\textnormal{sp}(R) \equiv{E}_\textnormal{sp}(\rho_{XB},R) := \sup_{0 \leq \alpha\leq 1} \frac{1-\alpha}{\alpha} \left( R - H_{\alpha}^\uparrow(X|B)_\rho \right), \end{align}
and $H_{\alpha}^\uparrow(X|B)_\rho := \max_{\sigma_B \in \mathcal{S}(B)} - D_\alpha(\rho_{XB}\|\mathds{1}_X\otimes \sigma_B)$.
\end{restatable}
\noindent The proof of Theorem~\ref{theo:sp_SW} is in  Section~\ref{subsec:large_converse}.

On the other hand, for $R < H(X|B)_\rho$, no sequence of codes $\sC_n$ can achieve vanishing error asymptotically. For this range, we in fact show that the probability of error converges exponentially quickly to one, as shown by the bounds on $sc(n,R)$ given in the following theorems.

\begin{restatable}{theo}{SCconversebound}\label{thm:SC-converse-bound}
For all $R<H(X|B)_{\rho},$ the finite blocklength strong converse exponent defined in \eqref{eq:def_sc} satisfies
\begin{equation}
sc(n,R)\geq E_\textnormal{sc}^*(R) >0,
\end{equation}
where 
\begin{equation} \label{eq:def_Esc}
  {E}_\textnormal{sc}^*(R) \equiv {E}_\textnormal{sc}^*(\rho_{XB},R) := \sup_{\alpha> 1} \frac{1-\alpha}{\alpha} \left( R - H_{\alpha}^{*,\uparrow}(X|B)_\rho \right),
\end{equation}
and $H_{\alpha}^{*,\uparrow}(X|B)_\rho:= \max_{\sigma_B\in\mathcal{S}(B)} - D_\alpha^*(\rho_{XB}\|\mathds{1}_X \otimes \sigma_B)$, with $D^*_\alpha(\rho\|\sigma) := \frac{1}{\alpha-1} \log 
\Tr \left[ \left( \rho^{\frac12} \sigma^{\frac{1-\alpha}{\alpha}} \rho^{\frac12} \right)^\alpha \right]$ being the sandwiched R\'enyi divergence~\cite{MDS+13,WWY14}.
\end{restatable}
\noindent The proof of Theorem~\ref{thm:SC-converse-bound} is in  Section~\ref{subsec:SC_converse}.

We also obtain an upper bound on $sc(n,R)$, which, together with Theorem~\ref{thm:SC-converse-bound} shows that $E^*_\text{sc}(R)$ is the strong converse exponent in the asymptotic limit.
\begin{restatable}{theo}{SCachievbound}\label{thm:SC-achiev-bound}
For all $R<H(X|B)_{\rho}$,  the finite blocklength strong converse exponent defined in~\eqref{eq:def_sc} satisfies
	\begin{equation}\label{eq:sc-ub}
	 sc(n,R) \leq  E_\textnormal{sc}^*(R) + \frac{c}{m}\log(m+1) + O_m\left(\tfrac{1}{\sqrt{n}}\right)
	\end{equation}
for $c:=\frac{3(|\cH_B|+2)(|\cH_B|-1)}{2}$ and any $m\in \N$, where we denote by $O_m(\frac{1}{\sqrt{n}})$ any term which is bounded by $C_m \frac{1}{\sqrt{n}}$ for all $n$ large enough, for some constant $C_m$ depending only on $m$, $|\cX|$, and $\rho_{XB}$. In particular, taking $n\to\infty$ then $m\to\infty$ yields
\begin{equation}
\limsup_{n\to\infty} sc(n,R) \leq  E_\textnormal{sc}^*(R).
\end{equation}
\end{restatable}
\noindent  The proof of Theorem~\ref{thm:SC-achiev-bound} is in  Section~\ref{subsec:SC_achiev}, along with Proposition~\ref{prop:full-SC-achiev-bound}, a more detailed version of the result with the constants written explicitly. Note that, together, Theorems~\ref{thm:SC-converse-bound} and \ref{thm:SC-achiev-bound} imply the following result.
\begin{coro} \label{cor:exact_sc}
For all $R < H(X|B)_\rho$, the strong converse exponent defined in~\eqref{eq:def_sc} satisfies
\begin{equation}
\lim_{n\to\infty} sc(n,R) = E^*_\text{sc}(R).
\end{equation}
\end{coro}

Lastly, we consider the case where the rate depends on $n$ as $R_n := H(X|B)_\rho + a_n$, where $a_n$ is a \emph{moderate sequence}, that is, a sequence of real numbers  satisfying
\begin{align} \label{eq:a_n}
\begin{split}
&\textnormal{(i)} \;a_n \to 0, \quad\text{as} \quad n\to \infty,\\
&\textnormal{(ii)} \; a_n \sqrt{n} \to \infty, \quad\text{as}\quad n\to \infty.
\end{split}
\end{align}
In this case, we have the following asymptotic result.

\begin{restatable}{theo}{modlarge}\label{theo:mod_large}
Assume that the c-q state $\rho_{XB}$ has strictly positive \emph{conditional information variance} $V(X|B)_\rho$, where  
\begin{equation}
V(X|B)_\rho := V(\rho_{XB}\|\mathds{1}_X\otimes \rho_B )
\end{equation}
 with $V(\rho\|\sigma) :=  \Tr  [ \rho \left( \log \rho - \log \sigma \right)^2  ] - D(\rho\|\sigma)^2$. Then for any sequence $(a_n)_{n\in \mathbb{N}}$ satisfying Eq.~\eqref{eq:a_n}, 
\begin{equation}\label{eq:mod_large0}
\lim_{n\to \infty} \frac{1}{ na_n^2} \log \Pestar( n, R_n ) = -\frac{1}{2 V(X|B)_\rho }
\end{equation}
for $R_n:=H(X|B)_\rho + a_n$.

\end{restatable}

\begin{restatable}{theo}{modrate} \label{theo:moderate_rate}
Assume that the c-q state $\rho_{XB}$ has $V(X|B)_\rho>0$.
Then for any sequence $(a_n)_{n\in \mathbb{N}}$ satisfying Eq.~\eqref{eq:a_n}, and $\eps_n := \e^{-n a_n^2}$, we have the asymptotic expansion
\begin{equation} \label{eq:moderate_exp-rate}
  R^* (n, \eps_n) = H(X|B)_\rho +\sqrt{2V(X|B)_\rho}  a_n + o(a_n).
\end{equation} 
\end{restatable}
\noindent The proof of Theorem~\ref{theo:moderate_rate} is in  Section~\ref{sec:mod_rate}.

\section{Preliminaries and Notation} \label{sec:pre}
Throughout this paper, we consider a finite-dimensional Hilbert space $\mathcal{H}$.  
The set of density operators (i.e.~positive semi-definite operators with unit trace) on $\mathcal{H}$ is defined as $\mathcal{S(H)}$.
The {quantum systems}, denoted by capital letter (e.g.~$A$, $B$), are modeled by finite-dimensional Hilbert spaces (e.g.~$\mathcal{H}_A, \mathcal{H}_B$);  $n$ copies of a system $A$ is denoted by $A^{n}$, and is modeled by the $n$-fold tensor product of the Hilbert spaces, $\mathcal{H}_{A^n} = \mathcal{H}_A^{\otimes n}$.
For $\rho,\sigma\in\mathcal{S(H)}$, we denote by $\rho\ll \sigma$ if the support of $\rho$ is contained in the support of $\sigma$. The identity operator on $\mathcal{H}$ is denoted by $\mathds{1}_\mathcal{H}$. The subscript will be removed if no confusion is possible.
We use $\Tr\left[\cdot \right]$ as the standard trace function.
For a bipartite state $\rho_{AB} \in \mathcal{S}(AB)$, $\Tr_B\left[ \rho_{AB} \right]$ denotes the partial trace with respect to system $B$.
We denote by $|t|^+ := \max\{0,t\}$.
The indicator function $\mathbf{1}_{\{A\}}$ is defined as follows: $\mathbf{1}_{\{A\}} = 1$ if the event $A$ is true; otherwise $\mathbf{1}_{\{A\}} = 0$.

For a positive semi-definite operator $\bm{X}$ whose spectral decomposition is $\bm{X} = \sum_{i} a_i P_i$, where $(a_i)_i$ and $(P_i)_i$ are the eigenvalues and eigenprojections of $\bm{X}$, its power is defined as: $\bm{X}^p := \sum_{i:a_i\neq 0} a_i^p P_i$.
In particular, $\bm{X}^0$ denotes the projection onto $\texttt{supp}(\bm{X})$, where we use $\texttt{supp}(A)$ to denote the support of the operator $A$.
Further, $A\perp B$ means $\texttt{supp}(A) \cap \texttt{supp}(B) = \emptyset$.
Additionally, we define the pinching map with respect to $\bm{X}$ by $\cP_{\bm{X}}(A) = \sum_i P_i A P_i$.
The $\exp$ and $\log$ are performed on base $2$ throughout this paper.

\subsection{Entropic Quantities} 
For any pair of density operators $\rho$ and $\sigma$, we define the {quantum relative entropy}, Petz's quantum R\'enyi divergence \cite{Pet86}, sandwiched R\'enyi divergence \cite{MDS+13, WWY14}, and the {log-Euclidean R\'enyi divergence} \cite{ON00,MO17}, respectively, as follows: 
\begin{align}
D(\rho\|\sigma) &:=  \Tr \left[ \rho \left( \log \rho - \log \sigma \right) \right], \label{eq:relative}\\
D_\alpha(\rho\|\sigma) &:= \frac{1}{\alpha-1} \log Q_\alpha(\rho\|\sigma) , \quad
Q_\alpha(\rho\|\sigma) := \Tr \left[ \rho^\alpha \sigma^{1-\alpha} \right]; \label{eq:Petz}\\
D^*_\alpha(\rho\|\sigma) &:= \frac{1}{\alpha-1} \log Q^*_\alpha(\rho\|\sigma), \quad
Q^*_\alpha(\rho\|\sigma) := \Tr \left[ \left( \rho^{\frac12} \sigma^{\frac{1-\alpha}{\alpha}} \rho^{\frac12} \right)^\alpha \right]; \label{eq:sandwich} \\
D^\flat_\alpha(\rho\|\sigma) &:= \frac{1}{\alpha-1} \log Q^\flat_\alpha(\rho\|\sigma), \quad
Q^\flat_\alpha(\rho\|\sigma) := \Tr \left[\mathrm{e}^{\alpha \log \rho + (1-\alpha) \log \sigma}\right]. \label{eq:chaotic}
\end{align}
We define the {quantum relative entropy variance}~\cite{TH13, Li14} by
\begin{align}
&V(\rho\|\sigma) :=  \Tr \left[ \rho \left( \log \rho - \log \sigma \right)^2 \right] - D(\rho\|\sigma)^2. \label{eq:V}\end{align}
The above quantity is non-negative. Further, it follows that
\begin{align} \label{eq:positive_D}
V(\rho\|\sigma) > 0 \quad \text{implies} \quad D(\rho\|\sigma )> 0.
\end{align}
For $\rho_{AB} \in \mathcal{S}(AB)$, $\alpha\geq 0$ and $t= \{\,\}, \{\flat\}$, or $\{*\}$, the \emph{quantum conditional R\'enyi entropies} are given by
\begin{align}
\begin{split}\label{eq:cond_ent}
H_{\alpha}^{t,\uparrow} (A|B)_{\rho} &:= \max_{\sigma_B \in \mathcal{S}(B)} - D_\alpha^{t}\left( \rho_{AB} \| \mathds{1}_A \otimes \sigma_{B} \right),   \\
H_{\alpha}^{t,\downarrow} (A|B)_{\rho} &:= - D_\alpha^{t}\left( \rho_{AB} \| \mathds{1}_A \otimes \rho_{B} \right).
\end{split}
\end{align}
When $\alpha=1$ and $t= \{\,\}, \{\flat\}$, or $\{*\}$ in Eq.~\eqref{eq:cond_ent}, both quantities coincide with the usual \emph{quantum conditional entropy}:
\begin{align}
H_{1}^{t,\uparrow} (A|B)_{\rho} = H_{1}^{t,\downarrow} (A|B)_{\rho} = H(A|B)_\rho := H(AB)_\rho - H(B)_\rho,
\end{align}
where $H(A)_\rho := -\Tr[\rho_A\log \rho_A]$ denotes the \emph{von Neumann entropy}.
Similarly, for $\rho_{AB} \in \mathcal{S}(AB)$, we define the \emph{conditional information variance}:
\begin{align}
	V(A|B)_\rho := V(\rho_{AB}\|\mathds{1}_A\otimes \rho_B ). \label{eq:V_cond}
\end{align}
It is not hard to verify from Eq.~\eqref{eq:positive_D} that
\begin{align} \label{eq:positive_C}
	V(A|B)_\rho > 0 \quad \text{implies} \quad H(A|B)_\rho >0.
\end{align}

\begin{lemm}[{\cite{Lie73}, \cite[Lemma III.3, Lemma III.11, Theorem III.14, Corollary III.25]{MO17}}, {\cite[Corollary 2.2]{Hia16}}] \label{lemma:chaotic}
	Let $\rho,\sigma \in \mathcal{S(H)}$. Then, 
	\begin{align}
		&\alpha \mapsto \log Q_\alpha(\rho\|\sigma)
		\text{ and } \alpha \mapsto \log Q^\flat_\alpha(\rho\|\sigma) \; \text{ are convex on} \; (0,1); \label{eq:chaotic1}\\
	&\alpha \mapsto D_\alpha\left( \rho \|\sigma\right) \text{ is continuous and monotone increasing on } [0,1]. \label{eq:chaotic5}
	\end{align}
	Moreover\footnote{It was shown in \cite[Lemma III.22]{MO17} that the map $\sigma \mapsto D_\alpha(\rho\|\sigma)$ is lower semi-continuous on $\mathcal{S(H)}$ for all $\alpha\in(0,1)$. The argument can be extended to the range $\alpha\in[0,1]$ by the same method in \cite[Lemma III.22]{MO17}.},
	\begin{align}
	&\forall \alpha\in(0,1), \quad (\rho,\sigma) \mapsto  Q^\flat_\alpha(\rho\|\sigma) \; \text{ is jointly concave on } \mathcal{S(H)}\times \mathcal{S(H)}; \label{eq:chaotic3}\\
	&\forall \alpha\in[0,1], \quad \sigma \mapsto D_\alpha(\rho\|\sigma) \; \text{ is convex and lower semi-continuous on } \mathcal{S(H)}. \label{eq:chaotic4}
	\end{align}
\end{lemm}

\begin{prop}[Properties of $\alpha$-R\'enyi Conditional Entropy] \label{prop:H}
	Given any classical-quantum state $\rho_{XB} \in \mathcal{S}(XB)$, the following holds:
	\begin{enumerate}[(a)]
		\item\label{H-a} The map $\alpha \mapsto H_\alpha^\uparrow (X|B)_\rho $ is continuous and monotonically decreasing on $[0,1]$.
		
		\item\label{H-b} The map $\alpha \mapsto \frac{1-\alpha}{\alpha} H_\alpha^\uparrow(X|B)_\rho$ is concave on $(0,1)$.
	\end{enumerate}
\end{prop}
\noindent The proof is provided in Appendix~\ref{app:H}.

Given two states $\rho$ and $\sigma$, one can define an associated binary hypothesis testing problem of determining which of the two states was given via a binary POVM. Such a POVM is described by an operator $Q$ (associated, say, with the outcome $\rho$) such that $0\leq Q\leq \one$, called the \emph{test}. Two types of errors are possible;  the probability of measuring $\rho$ and reporting the outcome $\sigma$ is given by $\tr[(\one-Q)\rho]$ and called the \emph{type-I} error, while the probability of measuring $\sigma$ and reporting the outcome $\rho$ is given by $\tr[Q\sigma]$ and is called the \emph{type-II} error.
The hypothesis testing relative entropy (e.g.~as defined in~\cite{WR13}) is defined by
\begin{equation} \label{eq:def-DH}
D^\eps_\text{H}(\rho\|\sigma) = -\log \inf_{\substack{Q: 0\leq Q \leq \one \\ \tr[(\one - Q)\rho]\leq \eps}} \tr[Q\sigma]
\end{equation}
and characterizes the minimum type-II error incurred via a test which has type-I error at most $\eps$. The hypothesis testing relative entropy satisfies the data-processing inequality
	\begin{equation} \label{eq:DH_DPI}
	D^\eps_\text{H}(\rho\|\sigma) \geq D_\text{H}^\eps (\Phi(\rho)\|\Phi(\sigma))
	\end{equation}
for any completely positive map $\Phi$~\cite{WR13}.
This quantity has an interpretation as a relative entropy as it satisfies the following asymptotic equipartition property:
\begin{equation}
\lim_{n\to \infty} \frac{1}{n}D^\eps_\text{H}(\rho^{\otimes n}\|\sigma^{\otimes n})  = D(\rho\|\sigma)
\end{equation}
which was proven in two steps, by~\cite{HP91} and \cite{NO00}.

We can consider a related quantity, $\widehat{\alpha}_\mu(\cdot\|\cdot)$ which denotes the minimum type-I error such that the type-II error does not exceed $\mu$. That is,
	\begin{equation} \label{eq:def_hat-alpha}
	\widehat \alpha_\mu(\rho\|\sigma) =  \inf_{\substack{T: 0\leq T \leq \one \\ \tr[T\sigma]\leq \mu}} \tr[(\one - T)\rho]  \equiv \exp(-D^\mu_\text{H}(\sigma\|\rho)).
	\end{equation}
	By~\eqref{eq:DH_DPI}, for any completely positive map $\Phi$ we have
	\begin{equation} \label{eq:hat_alpha_DPI}
	\widehat \alpha_\mu(\rho\|\sigma) = \exp(-D^\mu_H(\sigma\|\rho)) \leq \exp(-D_H^\mu (\Phi(\sigma)\|\Phi(\rho)))=\widehat \alpha_\mu(\Phi(\rho)\|\Phi(\sigma)).
	\end{equation}

We also consider the so-called max-relative entropy, given by
\begin{equation} \label{eq:def_Dmax}
D_\text{max}(\rho\|\sigma) := \inf\{ \gamma: \rho\leq 2^\gamma \sigma \}.
\end{equation} 
In establishing the exact strong converse exponent, we will employ a smoothed variant of this quantity,
\begin{equation} \label{eq:Dmax_smoothed}
D^\delta_{\text{max}}(\rho\|\sigma) = \min_{\bar \rho \in B_\delta(\rho)} D_\text{max}(\bar \rho \|\sigma),
\end{equation}
where 
\begin{equation}
 \B_\delta(\rho) = \{ \bar \rho \in \cD(\cH): d_\text{op}(\bar \rho, \rho) \leq \delta\}
 \end{equation} is the $\eps$-ball in the \emph{distance of optimal purifications}, $d_\text{op}$, defined by 
\begin{equation} \label{eq:def_dist_purifications}
d_\text{op}(\rho,\sigma) = \min_{\psi_\rho, \psi_\sigma} \frac{1}{2}\| \ket{\psi_\rho}\bra{\psi_\rho} - \ket{\psi_\sigma}\bra{\psi_\sigma}\|_1
\end{equation}
where the minimum is over purifications $\psi_\rho$ of $\rho$ and $\psi_\sigma$ of $\sigma$. By equation (2) of \cite{DMHB11}, the distance $d_\text{op}$ satisfies
\begin{equation}
\frac{1}{2}d_\text{op}(\rho,\sigma)^2 \leq \frac{1}{2}\|\rho-\sigma\|_1\leq d_\text{op}(\rho,\sigma). \label{eq:d_op_vs_td}
\end{equation}
It was shown in~\cite{DMHB11} that $D_{\text{max}}^\delta(\rho\|\sigma)$ satisfies an asymptotic equipartition property. In fact, Theorem 14 of~\cite{DMHB11} gives finite $n$ upper and lower bounds on $\frac{1}{n}D^\delta_{\text{max}}(\rho^{\otimes n}\|\sigma^{\otimes n})$ which converge to $D(\rho\|\sigma)$. We will only need the upper bound, namely
\begin{equation} \label{eq:Ddelta-AEP-UB}
\frac{1}{n}D^\delta_{\text{max}}(\rho^{\otimes n}\|\sigma^{\otimes n}) \leq D(\rho\|\sigma) + \frac{1}{\sqrt{n}}4 \sqrt{2}(\log \eta) \log \frac{1}{1 - \sqrt{1-\eps^2}},
\end{equation}
where $\eta := 1 + \tr [\rho^{3/2}\sigma^{-1/2} + \rho^{1/2} \sigma^{1/2}]$.

The smoothed max-relative entropy satisfies the following simple but useful relation (which is a one-shot analog of Lemma V.7 of \cite{MO17}).
\begin{lemm}	\label{lem:$n$-shot-stein-exp}
	For $\delta \geq 0$, if $a\geq D_\textnormal{max}^\delta(\sigma\|\rho)$ then
	\begin{equation}
	\tr[(\sigma - \e^a\rho)_+] \leq 2\delta.
	\end{equation}
\end{lemm}
\begin{proof}[Proof of Lemma~\ref{lem:$n$-shot-stein-exp}]	
	By the definition \eqref{eq:def_Dmax},
	for any $a\geq D_\text{max}(\sigma\|\rho)$, we have $\sigma - \e^a\rho \leq 0$, and therefore $\tr[(\sigma-\e^a\rho)_+]=0$, which proves the result for $\delta=0$. For $\delta >0$ and $a\geq D_\text{max}^\delta(\sigma\|\rho)$ there exists a density matrix $\tilde \sigma$ with $\tilde \sigma - \e^a\rho \leq 0$ and $d_\text{op}(\tilde \sigma,\sigma)\leq \delta$. Setting $P_+ = \{ \sigma \leq \e^a \rho\}$, we have
	\begin{equation}
	\tr[P_+(\sigma - \e^a \rho)] = \tr[ P_+ (\sigma-\tilde \sigma)] + \tr[ P_+ (\tilde\sigma-\e^a \rho)].
	\end{equation}
	Since $0\leq P_+ \leq \one$, we have $\tr[ P_+ (\tilde\sigma-\e^a \rho)] \leq \tr[(\tilde\sigma-\e^a \rho)_+]\leq 0$. Thus,
	\begin{equation}
	\tr[P_+(\sigma - \e^a \rho)] \leq \tr[ P_+ (\sigma-\tilde \sigma)] \leq \|\sigma-\tilde \sigma\|_1 \leq 2d_\text{op}(\sigma,\tilde\sigma)\leq 2\delta
	\end{equation}
	using \eqref{eq:d_op_vs_td} in the second to last inequality.
\end{proof}

\subsection{Error Exponent Function} \label{subsec:exponent}

For $t= \{\,\}, \{*\}$, or $\{\flat\}$, we define
\begin{align} 
E_\text{r}^t (R)\equiv E_\text{r}^t (\rho_{XB},R) &:= \max_{0\leq s\leq 1}\left\{ E_0^t (s ) + sR\right\}; \label{eq:gallager_r2}\\
E_\text{sp}^t (R) \equiv E_\text{sp}^t (\rho_{XB},R)  &:= \sup_{s\geq 0}\left\{ E_0^t(s ) + sR\right\}; \label{eq:gallager_sp2}\\
E_\text{sc}^t (R)\equiv E_\text{sc}^t (\rho_{XB},R) &:= \sup_{-1 < s< 0}\left\{ {E}^t_0 (s ) + sR\right\}; \label{eq:gallager_sc2}  \\
E_0^t(s ) \equiv E_0^t(\rho_{XB}, s ) &:= - s H_{\frac{1}{1+s}}^{t,\uparrow} (X|B)_{\rho}, \label{eq:E0SW1}
\end{align}
omitting the dependence on $\rho_{XB}$ except where necessary. For $t = \{ \, \}$, i.e.~the Petz's R\'enyi conditional entropy, one has 
\begin{align}
E_0(s) = -\log \Tr \left[ \left( \Tr_{X} \rho_{XB}^{\frac{1}{1+s}} \right)^{1+s} \right]
\label{eq:E0SW2}
\end{align}
by quantum Sibson's identity \cite{SW12}. 
We also define another version of exponent function via $H_\alpha^{\downarrow}$:
\begin{align}
{E}_\text{r}^\downarrow(R) &:= \max_{0\leq s\leq 1} \left\{  {E}_0^\downarrow (s ) + s R \right\}, \label{eq:E0_bar} \\
{E}_0^\downarrow (s ) &:= -s H_{1-s}^\downarrow (X|B)_\rho. \label{eq:E0SW2_bar}
\end{align}
Note that \cite[Proposition III.18]{MO17}
\begin{align}
&D^*_\alpha(\cdot\|\cdot) \leq D_\alpha(\cdot\|\cdot) \leq D^\flat_\alpha(\cdot\|\cdot), \quad \alpha\in[0,1) \\
&D^\flat_\alpha(\cdot\|\cdot) \leq D^*_\alpha(\cdot\|\cdot) \leq D_\alpha(\cdot\|\cdot), \quad \alpha\in(1,+\infty],
\end{align}
which implies that
\begin{align}
&{E}_\text{sp} (R)  \leq {E}_\text{sp}^{\flat} (R) \\
&{E}_\text{r} (R)  \leq  {E}_\text{r}^{\flat} (R), \\
&{E}_\text{sc} (R)  \leq  {E}_\text{sc}^{\flat} (R).
\end{align}
Further, $H_\alpha^\uparrow (X|B)_\rho \leq H_{2-\frac{1}{\alpha}}^\downarrow (X|B)_\rho$ for $\alpha\in[1/2,+\infty]$ \cite[Corollary 4]{TBH14}, \cite[Corollary 5.3]{Tom16}.
For $R\in[H^\uparrow_1(X|B)_\rho, H^\uparrow_{1/2}(X|B)_\rho]$, together with Proposition~\ref{prop:E_SW}-\ref{E_SW-a} below, we have
\begin{align}
{E}^{\downarrow}_\text{r} (R) \leq E_\text{r} (R)
= E_\text{sp} (R) \leq {E}_\text{sp}^\flat (R) = {E}^\flat_\text{r} (R).
\end{align}
In Section~\ref{subsec:large_achiev}, we obtain an achievability bound of the optimal error in terms of $E_\text{r}^{\downarrow}$. We conjecture that it can be further improved to $E_\text{r}$.

In the following, we collect some useful properties of the auxiliary functions $E_0 (s)$ and $ {E}^\downarrow_0 (s )$. 

\begin{restatable}[Properties of $E_0$ and $E_0^\downarrow$]{prop}{propaux}\label{prop:E0_SW+E0down}
	Let $\rho_{XB}$ be a classical-quantum state with $H(X|B)_{\rho}>0$, the auxiliary functions $E_0(s )$ and $E_0^\downarrow(s)$ admit the following properties. 
	\begin{enumerate}[(a)]
		\item\label{E0_SW-a} (Continuity) The function $s\mapsto E_0 (s )$ is continuous for all $s\in (-1,+\infty)$, and $s\mapsto E_0^\downarrow(s)$ is continuous for all $s\in [0,+\infty)$.		
		\item\label{E0_SW-b} (Negativity) For $s\geq 0$,
		\begin{align}
		E_0 (s ) \leq 0, \quad\text{and}\quad E_0^\downarrow (s ) \leq 0
		\end{align}
		with $E_0 (0 ) = E_0^\downarrow(0) = 0$.
		
		\item\label{E0_SW-c} (Concavity)  Both $s\mapsto E_0(s )$ and $s\mapsto E_0^\downarrow(s )$ are  concave in $s$ for all $s\in(-1,+\infty)$.
		
		\item\label{E0_SW-d} (First-order Derivative)
		\begin{align} \label{eq:first}
		\left.\frac{ \partial E_0^{\downarrow} (s )}{\partial s}\right|_{s = 0} =	\left.\frac{ \partial E_0 (s )}{\partial s}\right|_{s = 0} = -H(X|B)_{\rho}.
		\end{align}
		
		\item\label{E0_SW-e} (Second-order Derivative) 
		\begin{align} \label{eq:second}
		\left.\frac{ \partial^2 E_0^\downarrow (s )}{\partial s^2}\right|_{s = 0}=	\left.\frac{ \partial^2 E_0 (s )}{\partial s^2}\right|_{s = 0} = -V(X|B)_{\rho}.
		\end{align}
	\end{enumerate}
\end{restatable}
\noindent The proof is provided in Appendix~\ref{app:prop}.

Proposition~\ref{prop:E_SW} below discusses the properties of the exponent functions. See Figure~\ref{fig:Esp_SW} for the illustration. 

\begin{prop}[Properties of the Exponent Function] \label{prop:E_SW}
	Let $\rho_{XB}$ be a classical-quantum state with $H(X|B)_{\rho}>0$, the following holds.
	\begin{enumerate}[(a)]
		\item\label{E_SW-a}	$E_\textnormal{sp}(\cdot)$ is convex, differentiable, and monotonically increasing on $[0,+\infty]$. Further,
		\begin{align}
		E_\textnormal{sp}(R) = \begin{cases}
		0, & R\leq H_1^{\uparrow}(X|B)_{\rho} \\
		E_\textnormal{r}(R) , & H_{1}^{\uparrow}(X|B)_{\rho} \leq R \leq  H_{1/2}^{\uparrow}(X|B)_{\rho} \\
		+\infty, & R >  H_{0}^{\uparrow}(X|Y)_{\rho}
		\end{cases}.
		\end{align}
		
		\item\label{E_SW-b}	Define
		\begin{align}
		F_R(\alpha,\sigma_B) := \begin{dcases}\frac{1-\alpha}{\alpha} \left( R + D_\alpha\left(\rho_{XB}\| \mathds{1}_X \otimes \sigma_B  \right) \right), &\alpha\in(0,1), \\
		0, &\alpha = 1,
		\end{dcases}
		\end{align}
		on $(0,1]\times \mathcal{S}(B)$.
		For $R\in(H_1^{\uparrow}(X|B)_\rho, H_0^{\uparrow}(X|B)_\rho)$, there exists a unique saddle-point $(\alpha^\star,\sigma^\star) \in (0,1)\times\mathcal{S}(B)$ of $F_R(\cdot,\cdot)$ such that
		\begin{align}
		F_R(\alpha^\star,\sigma^\star) = \sup_{\alpha\in [0,1] } \inf_{\sigma_B \in \mathcal{S}(B)} F_R(\alpha,\sigma_B) = \inf_{\sigma_B \in \mathcal{S}(B)}\sup_{\alpha\in [0,1] }   F_R(\alpha,\sigma_B) = E_\textnormal{sp}(R).
		\end{align}
		
		\item\label{E_SW-c} Any saddle-point $(\alpha^\star, \sigma^\star)$ of $F_{R}(\cdot,\cdot)$ satisfies 		\begin{align}
		\mathds{1}_X \otimes \sigma^\star 
		\gg \rho_{XB}.
		\end{align}

	\end{enumerate}	
\end{prop}
\noindent The proof is provided in Appendix~\ref{app:prop}.

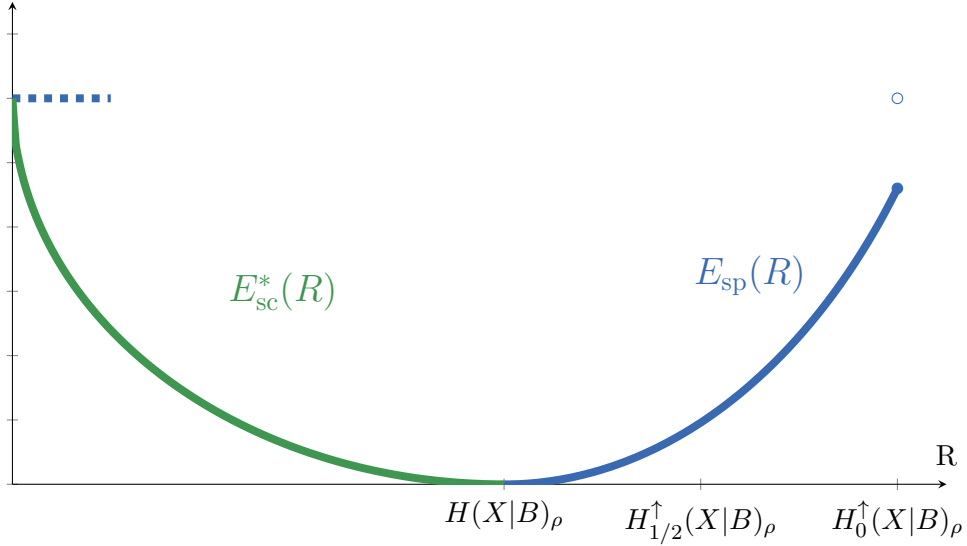
\begin{figure}
	\centering
	
	  \begin{tikzpicture}
    \begin{axis}[axis x line=bottom,
    axis y line = left,
    width=140mm, height=80mm,
      xmin=0, xmax = 1.9, ymin = 0, ymax = 1.5,
      xtick = {1,1.4,1.8},
      xticklabels={$H(X|B)_\rho$, $H_{1/2}^\uparrow(X|B)_\rho$, $H_0^\uparrow(X|B)_\rho$},
      yticklabels={},
      x label style={at={(axis description cs:1,0.1)},anchor=south},
      ylabel = {},
      xlabel={R},
      legend style={
        cells={anchor=west},
        legend pos=north west,}
      ]
	\draw[spcolor, line width=3pt,dashed] (axis cs: 0,1.2) --(axis cs: .2,1.2);

      \addplot[smooth,sccolor, line width=3pt,samples=155, domain=-0:1]  {
        1.2*(-(1 - (1-x)^2)^(.5) + 1)
      };
      
      \addplot[smooth,spcolor,line width=3pt, samples=155,domain=1:1.8]  {
       2.3*( -(1 - (1-x)^2)^(.5) + 1)
      };

      \node[text=spcolor] at (axis cs:1.5,.65){{\Large $E_\text{sp}(R)$}};
      \node[text=sccolor] at (axis cs:.55,.6){{\Large $E^*_\text{sc}(R)$}};

      \addplot [only marks,mark=*, spcolor] coordinates {
		(1.8,2.3/2.5)
	};
	\addplot [only marks,mark=o, spcolor] coordinates {
		(1.8,1.2)
	};
    \end{axis}
  \end{tikzpicture}

	\caption{Illustration of the sphere-packing exponent $E_\text{sp}(R)$ (right, blue curve) and the correct decoding exponent $E_\text{sc}^*(R)$ (left, green curve) in Slepian-Wolf coding over $R\geq 0$. 
		The conditional R\'enyi entropy is denoted by $H^{\uparrow}_\alpha(X|B)_\rho$. We note that $E_\text{sp}(R) = E_\text{r}(R)$ for $R\in[H^{\uparrow}_1(X|B)_\rho, H^{\uparrow}_{1/2}(X|B)_\rho]$.
	}
	\label{fig:Esp_SW}
\end{figure}

In Proposition~\ref{prop:representation} below, we show that the exponent functions defined in terms of $D^\flat$ admit the variational representations, analogous to those introduced in the classical case by Csisz{\'a}r and J.~K{\"o}rner's  \cite{CK80, Csi82, CK11}.
\begin{prop}\label{prop:representation}
Given a c-q state
\begin{equation}
\rho_{XB} = \sum_{x\in \cX} p(x) \ket{x}\bra{x}\otimes \rho_B^{x}
\end{equation}
we have the variational expressions
\begin{align}
{E}_\textnormal{r}^{\flat}(R) &= \min_{\sigma_{XB} \in \mathcal{S}_\rho(XB)}  \left\{ D\left( \sigma_{XB} \| {\rho}_{XB}  \right) + \left| R - H(X|B)_{\sigma} \right|^+ \right\}, \label{eq:representation1} \\
{E}_\textnormal{sp}^{\flat}(R) &=\min_{\sigma_{XB} \in \mathcal{S}_\rho(XB)} \left\{ D\left( \sigma_{XB} \| {\rho}_{XB}  \right) : R \leq H(X|B)_{\sigma} \right\}, \label{eq:representation2} \\
{E}_\textnormal{sc}^{\flat}(R) &= \min_{\sigma_{XB} \in \mathcal{S}_\rho(XB)} \left\{ D\left( \sigma_{XB} \| {\rho}_{XB}  \right)+ \left|  H(X|B)_{\sigma} - R \right|^+ \right\},\label{eq:representation3}
	\end{align}
where we denote $\mathcal{S}_\rho(XB)$ as the set of states $\sigma_{XB}$ with  $\sigma_{XB} \ll \rho_{XB}$ and which can be written as
\begin{equation}
\sigma_{XB} = \sum_{x\in \cX}q(x) \ket{x}\bra{x}\otimes \sigma_{B}^{x}
\end{equation}
for some probability distribution $q$ on $\cX$ and states $\sigma_B^{x}$ in $\mathcal{S}(B)$. 
\end{prop}
\noindent The proof is provided in Appendix~\ref{app:prop}.

\section{Error exponent at a fixed rate above the Slepian-Wolf limit (Large Deviation Regime)} \label{sec:Large}
Given a fixed compression rate $R$ above the Slepian-Wolf limit $H(X|B)_{\rho}$, Devetak and Winter showed that the optimal error probability vanishes asymptotically \cite{DW03}. In this section, we establish the finite blocklength achievability and converse bounds and show that the optimal error will exponential decay as a function of blocklength $n$. Specifically, we obtain (Theorems~\ref{theo:large_ach} and \ref{theo:sp_SW}):
\begin{equation}
n E_\text{r}^{\downarrow} (R) - 2 \leq-\log \Pestar(n,R) \leq n E_\text{sp}(R) + O(\log n).
\end{equation}
The exponent functions and their properties are introduced in Section~\ref{subsec:exponent}. The achievability and converse bounds are proved in Section \ref{subsec:large_achiev} and \ref{subsec:large_converse}.

\subsection{Achievability} \label{subsec:large_achiev}

Let us recall Theorem~\ref{theo:large_ach}.
\largeach*

\begin{proof}
	Our technique is to use a random coding argument to prove Theorem~\ref{theo:large_ach}. The idea originates from Gallager~\cite{Gal76} and later studied by Renes and Renner~\cite{RR12}.
	
 	We first present an one-shot lower bound on $e(1, R)$ before extending to the $n$-shot case. Let $\cs$ be a set. Consider a CQSW code $\sC$ consisting of a random encoder $f:\mathcal{X} \to  \cs$ which encodes every source $x\in\mathcal{X}$ into some index $w\in \cs$ with equal probability $1/| \cs|$, and a decoder
	given by a so-called pretty good measurement:
	\begin{align}
	\Pi_{x}^{(w)} := \left( \sum_{\bar{x}: f(\bar{x}) = w} \Lambda_{ \bar{x}}  \right)^{-1/2} 
	\Lambda_{ x }
	\left( \sum_{\bar{x}: f(\bar{x}) = w} \Lambda_{ \bar{x}}  \right)^{-1/2},
	\end{align}
	where $0\leq \Lambda_x \leq \mathds{1}_B$ for each $x\in\mathcal{X}$ will be specified later.
	The optimal probability of error can be upper bounded by
	\begin{equation}
	\Pestar (1,\log |\cs|) \leq \Pe(\sC) = \mathbb{E}_{ {x} } \mathbb{E}_w [\eps(x,w)], \label{eq:large_ach1}
	\end{equation}  
	where  $\eps(x,w) := \Tr\left[ \rho_B^{x} \left( \mathds{1}_B - \Pi_{ {x} }^{(w)} \right) \right]$ is the  probability of error conditioned on the source emitting the symbol $x$ and the encoder encoding $x$ as $w$.

	Applying the Hayashi-Nagaoka inequality \cite[Lemma 2]{HN03} to obtain
	\begin{equation}
	\mathds{1}_X - \Pi_{ {x} }^{(w)} \leq 2 \left( \mathds{1}_B - \Lambda_x \right) + 4 \sum_{\bar{x}\neq x } \mathbf{1}_{ \{ f(\bar{x}) = w  \}} \,\Lambda_{\bar{x}}, \label{eq:large_ach2}
	\end{equation}
	where $\mathbf{1}_{ \{ f(\bar{x}) = w  \}}$ denotes the indicator function for the event $f(\bar{x}) = w$.
	Combining Eqs.~\eqref{eq:large_ach1} and \eqref{eq:large_ach2} gives
	\begin{equation}
	\eps(x,w) \leq 2 \Tr\left[ \rho_B^{x} \left( \mathds{1}_B - \Lambda_x \right)\right] + 4 \Tr\Big[ \rho_B^{x} \, \sum_{\bar{x}\neq x } \mathbf{1}_{ \{ f(\bar{x}) = w  \}}\, \Lambda_{\bar{x}} \Big].
	\end{equation}
	Taking average over $w$, we find
	\begin{equation}
	\mathbb{E}_w\left[ \eps( {x},w) \right] \leq  2 \Tr\left[ \rho_B^{x} \left( \mathds{1}_B - \Lambda_x \right)\right] + 4 \Pr\left\{ { f(\bar{x}) = w  } \right\} \Tr\Big[ \rho_B^{x} \, \sum_{\bar{x}\neq x }  \Lambda_{\bar{x}} \Big].
	\end{equation}
	 By using the assumption $\Pr\left\{  f(\bar{x}) = w  \right\} = 1/|\cs|$,
	\begin{align}
	\mathbb{E}_w\left[ \eps( {x},w) \right] &= 2 \Tr\left[ \rho_B^{x} \left( \mathds{1}_B - \Lambda_x \right)\right] + \frac{4}{|\cs|}  \Tr\Big[ \rho_B^{x}\, \sum_{\bar{x}\neq x }  \Lambda_{\bar{x}} \Big], \\
	&\leq 2 \Tr\left[ \rho_B^{x} \left( \mathds{1}_B - \Lambda_x \right)\right] + \frac{4}{|\cs|}  \Tr\Big[ \rho_B^{x} \, \sum_{\bar{x}\in\mathcal{X} }  \Lambda_{\bar{x}} \Big].
		\end{align}
	By taking average over $x$ we obtain
	\begin{align}
	\Pestar (1,\log |\cs|) &\leq 2 \sum_{x\in\mathcal{X}} p(x) \Tr\left[ \rho_B^{x} \left( \mathds{1}_B - \Lambda_x \right)\right] + \frac{4}{|\cs|}  \Tr\Big[ \rho_B \, \sum_{\bar{x}\in\mathcal{X} }  \Lambda_{\bar{x}} \Big], \\
	&= 2 \Tr\left[ \rho_{XB} \left( \mathds{1}_{XB} - \Lambda_{XB} \right) \right] + \frac{4}{|\mathcal{W}|} \Tr\left[  \mathds{1}_X\otimes \rho_B \,\Lambda_{XB} \right],
	\end{align}
	where $\Lambda_{XB} := \sum_{x\in\mathcal{X}} |x\rangle\langle x| \otimes \Lambda_x$.
	Now, we choose, for $x\in\mathcal{X}$,
	\begin{align}
	\Lambda_x := \left\{ p(x) \rho_B^{x} - \frac{1}{|\mathcal{W}|} \rho_B \geq 0   \right\}.
	\end{align}
	We invoke Audenaert \textit{et al.}'s inequality \cite{ACM+07, ANS+08}: for every $A,B\geq 0$ and $s\in[0,1]$,
	\begin{align}
	\Tr\left[ \left\{ A - B \geq 0  \right\} B + \left\{ B - A \leq 0\right\} A \right] \leq \Tr\left[ A^{1-s} B^s \right].
	\end{align}
	Letting $A = \rho_{XB}$, $B = \frac{1}{|\mathcal{W}|}\mathds{1}_X\otimes \rho_B$, and noting that $\Lambda_{XB} = \left\{ \rho_{XB} - \frac{1}{|\mathcal{W}|}\mathds{1}_X\otimes \rho_B \geq 0   \right\}$, we have one-shot achievability:
		\begin{align}
	\Pestar (1,
	\log |\mathcal{W}|) &\leq 4 \min_{s\in[0,1]} |\mathcal{W}|^{-s} \Tr\left[ \rho_{XB}^{1-s} \left( \mathds{1}_X\otimes \rho_B\right)^s  \right]. \label{eq:large_ach3}
	\end{align}
	
	Finally, we consider the $n$-tuple case. Note that $\rho_{X^n B^n} = \rho_{XB}^{\otimes n}$, and let $|\mathcal{W}| = \exp\{nR\}$. Eqs.~\eqref{eq:large_ach3} and \eqref{eq:E0_bar} lead to 
	\begin{align}
	\Pestar (n,
	R) &\leq 4 \exp\left\{ -n {E}_\text{r}^{\downarrow}(R)  \right\},
	\end{align}	
	which completes the proof.
\end{proof}

\subsection{Optimality} \label{subsec:large_converse}
The main result of this section is the finite blocklength converse bound for the optimal error probability---Theorem~\ref{theo:sp_SW}. We call this the sphere-packing bound for Slepian-Wolf coding with quantum side information, as a counterpart of the sphere-packing bound in classical-quantum channel coding \cite{Dal13,CHT17}.
The proof technique relies on an one-shot converse bound in Proposition~\ref{prop:one-shot_converse} below (adapting the technique of \cite{WR13} to the case with side information), and a sharp $n$-shot expansion from \cite{CHT17,CH17} of a hypothesis testing quantity.

\begin{prop}[One-Shot Converse Bound for Error] \label{prop:one-shot_converse}
	Consider a Slepian-Wolf coding with a joint classical-quantum state  $\rho_{XB}\in\mathcal{S}(XB)$ and the index size $ |\cs|<|\mathcal{X}|$.
		Then,
	\begin{align} \label{eq:one-shot_converse}
	-\log \Pestar(1,\log |\mathcal{W}| )
	\leq \min_{\sigma_B \in \mathcal{S}(B)} - \log \widehat{\alpha}_{   \frac{|\mathcal{W}|}{|\mathcal{X}|} }\left(\rho_{XB}\| \tau_{{X}}\otimes \sigma_B \right),
	\end{align}
	where $\tau_{X}$ denotes the uniform distribution on the input alphabet $\mathcal{X}$; and $\widehat{\alpha}_\mu(\cdot\|\cdot)$ is defined by \eqref{eq:def_hat-alpha}.
\end{prop}
\begin{proof}[Proof of Proposition~\ref{prop:one-shot_converse}] Let $\cs$ be a set.
As discussed below \eqref{eq:Pe_rand_as_det}, we may reduce to a deterministic encoding $\cE$.
	Consider the map $\Lambda: \mathcal{S}(XB)\to \mathcal{S}(XX)$ 
	\begin{equation}
	\Lambda  =\sum_{x\in \cX} L_{\ket{x}\bra{x}}R_{\ket{x}\bra{x}} \otimes \Lambda^x, \qquad \Lambda^x : \sigma_B \mapsto \sum_{\hat x}\tr[\Pi_{\hat x}^{(\cE(x))}\sigma_B]\ket{\hat x}\bra{\hat x}
	\end{equation}
	which, for each $x$, projects into the classical state $\ket{x}\bra{x}$ and applies the measure-and-prepare map $\Lambda^x$ which measures according to the POVM $\{ \Pi_{\hat x}^{(\cE(x))}\}_{\hat x\in \cX}$ and records the outcome in a classical register. Here, $L_{\ket{x}\bra{x}}$ is the operator that acts by left-multiplication by projector $\ket{x}\bra{x}$ and similarly $R_{\ket{x}\bra{x}}$ acts by right-multiplication by $\ket{x}\bra{x}$. 
	We can quickly see $\Lambda$ is completely positive (CP):
	the map $X\mapsto L_AR_{A^*}(X) = A X A^*$ is CP (since $A$ is its only Kraus operator); as a measure-and-prepare map, $\Lambda^x$ is CP (see e.g.~\cite{Wilde2}) and the sum and tensor product of CP maps is CP.
					We set $\tau_X = \frac{1}{|\cX|} \one_X = \frac{1}{|\cX|} \sum_{x\in\cX} \ket{x}\bra{x}$ as the completely mixed state on $X$ and $\sigma_B\in \mathcal{S}(B)$ arbitrary.
				Note that
	\begin{equation}
	\Lambda( \tau_X \otimes \sigma_B) =\frac{1}{|\cX|} \sum_{x\in \cX} \Lambda(\ket{x}\bra{x}\otimes \sigma_B) = \frac{1}{|\cX|} \sum_{x \in \cX} \ket{x}\bra{x} \otimes \sum_{\hat x} \tr[\Pi^{(\cE(x))}_{\hat x} \sigma_B] \ket{\hat x}\bra{\hat x}, \label{eq:Lambda-on-tausig}
	\end{equation}
	and
	\begin{equation}
	\Lambda(\rho_{XB}) = \sum_{x\in \cX} p(x) \Lambda(\ket{x}\bra{x}\otimes\rho_B^x) =\sum_{x\in \cX} p(x) \ket{x}\bra{x} \otimes  \sum_{\hat x\in \cX}\tr[\Pi^{(\cE(x))}_{\hat x} \rho_B^x]\ket{\hat x}\bra{\hat x}. \label{eq:Lambda-on-rho}
	\end{equation}
	
	Now, take a two element POVM (the test) as $T = \sum_y \ket{y}\bra{y} \otimes \ket{y} \bra{y}$. Then, by \eqref{eq:Lambda-on-rho}
	\begin{equation}
	\tr[ (\one - T) \Lambda(\rho_{XB})] = 1 - \sum_y p(y)\tr[\Pi^{(\cE(y))}_{y} \rho_B^y] = \Pe(\cC),
	\end{equation}
	so this test has type I error of $\Pe(\cC)$.
	On the other hand, by \eqref{eq:Lambda-on-tausig} the quantity $\tr[ T\Lambda(\tau_X\otimes \sigma_B)]$ can be expanded as
	\begin{equation}
	 \sum_y \frac{1}{|\cX|}  \tr[\Pi^{(\cE(y))}_{ y} \sigma_B] = \frac{1}{|\cX|} \sum_{w\in \cs} \sum_{y\in \cX : \cE(y) = w} \tr[\Pi^{(w)}_{ y} \sigma_B] \leq \frac{1}{|\cX|} \sum_{w\in \cs}  \tr[\one_B \sigma_B]  \leq \frac{|\cs|}{|\cX|}.
	\end{equation}
	That is, this test achieves type II error of $\frac{|\mathcal{W}|}{|\cX|}$.
	As the infimum over all such tests, we have that
	\begin{equation}
	\Pe(\cC) \geq  \widehat \alpha_{\frac{|\cs|}{|\cX|}}(\Lambda(\rho_{XB})\| \Lambda(\tau_X\otimes \sigma_B))  \geq \widehat \alpha_{\frac{|\cs|}{|\cX|}}(\rho_{XB} \| \tau_X\otimes \sigma_B),
	\end{equation}
	where the second inequality follows from the data-processing inequality~\eqref{eq:hat_alpha_DPI}.
	Then taking the infimum over $\cE$ and $\cD$,
	\begin{equation}
	\widehat \alpha_{\frac{|\mathcal{W}|}{|\cX|}} (\rho_{XB}\|\tau_X\otimes \sigma_B)\leq \Pestar(1,\log|\cs|)
	\end{equation}
	yielding
	\begin{equation}
	-\log \Pestar(1,\log |\mathcal{W}|)  \leq - \log \widehat \alpha_{\frac{|\cs|}{|\cX|}} (\rho_{XB}\|\tau_X\otimes \sigma_B).
	\end{equation}
	Since this holds independently of $\sigma_B\in \mathcal{S}(B)$, we may minimize over $\sigma_B$ to obtain the result.
\end{proof}

\begin{prop}[Sharp Converse Hoeffding Bound {\cite[Proposition~14]{CHT17}, \cite[Proposition~7]{CH17}}] \label{prop:sharp_Hoeffding}
	Consider a binary hypothesis testing: $\mathsf{H}_0: \rho^n = \bigotimes_{i=1}^n \rho_i $ and $\mathsf{H}_1: \sigma^n = \bigotimes_{i=1}^n \sigma_i$ with $\rho^n \ll \sigma^n $, and a sequence of positive numbers $(r_n)_{n\in\mathbb{N}}$. 
	Denote by 
	\begin{align}
	\phi_n ( r_n | \rho^n \| \sigma^n ) &:=  \sup_{\alpha\in(0,1]}  \frac{1-\alpha}{\alpha} \left( \frac1n D_\alpha(\rho^n\|\sigma^n) - r_n \right); \\
	s_n^\star &:= \argmax_{s\geq 0} \left\{    \frac{s}{n} D_{\frac{1}{1+s}} (\rho^n\|\sigma^n) - s r_n
	\right\}.
	\end{align}
	If $ \frac1n V(\rho^n\|\sigma^n) \geq \nu$ for some $\nu > 0$, then there exist $N_1\in\mathbb{N}$, $K_1, C >0$ such that for all $n\geq N_1$, we have
	\begin{align}
	-\log \widehat{\alpha}_{  \exp\{-n r_n \}} \left( \rho^{ n}\|\sigma^{ n} \right) \leq n \phi_n ( r_n - \gamma_n  |\rho^n\|\sigma^n) + \frac12 \log \left(  n s_n^\star \right)  + K_1, \label{eq:sharp1}
	\end{align}
	where $\gamma_n := \frac{\log n}{2n} + \frac{C}{n}$.
	Moreover, if there exists $\eps>0$ such that for all $\bar{r}\in (r_n-\eps,r_n]$, 
	\begin{align}
		\phi_n (\bar{r}| \rho^n \| \sigma^n )  \in [\nu,+\infty), \label{eq:cond_sharp}
	\end{align}
	for some $\nu>0$, then there exist $N_2\in\mathbb{N}$, $K_2>0$ such that  for all $n\geq N_2$, we have
	\begin{align}
	-\log \widehat{\alpha}_{  \exp\{-nr_n\}} \left( \rho^{ n}\|\sigma^{ n} \right)
	&\leq n \phi_n ( r_n |\rho^n\|\sigma^n) + \frac12\left( 1 + s_n^\star \right) \log n  + K_2. \label{eq:sharp2}
	\end{align}	
\end{prop}

With Propositions~\ref{prop:one-shot_converse} and \ref{prop:sharp_Hoeffding}, we are able to show the $n$-shot bound given by Theorem~\ref{theo:sp_SW}, which we recall here:

\spSW*

\begin{proof}[Proof of Theorem \ref{theo:sp_SW}]
	The proof is split into two parts. We first invoke an one-shot converse bound in Proposition \ref{prop:one-shot_converse} to relate the optimal error of Slepian-Wolf coding to a binary hypothesis testing problem. Second, we employ a sharp converse Hoeffding bound in Proposition \ref{prop:sharp_Hoeffding} to asymptotically expand the optimal type-I error, which yields the desired result.

	Applying Proposition~\ref{prop:one-shot_converse} with an $n$-shot extension $\rho_{X^nB^n}$ of the c-q state $\rho_{XB}$, $|\cs|= \exp\{nR\}$, and $\tau_{X^n} = \frac{1}{|X^n|} \one_{X^n}$ gives
	\begin{align}
	\log \left( \frac{1}{\Pestar (n, R )} \right)
	&\leq \min_{\sigma_B^n \in \mathcal{S}(B^n)} - \log \widehat{\alpha}_{   \frac{|\cs|}{|\mathcal{X}^n|} }\left(\rho_{X^n B^n}\|\tau_{{X^n}}\otimes \sigma_B^n \right) \\
	&\leq - \log \widehat{\alpha}_{   \frac{|\cs|}{|\mathcal{X}^n|} }\left(\rho_{X^n B^n}\|\tau_{{X^n}}\otimes \left(\sigma_B^\star\right)^{\otimes n} \right), \\
	&= - \log \widehat{\alpha}_{   \frac{|\cs|}{|\mathcal{X}|^n} }\left(\rho_{X B}^{\otimes n} \| \left( \tau_{{X}}\otimes \sigma_B^\star\right)^{\otimes n} \right),
	\label{eq:sp_SW1}
	\end{align}
	where we invoke the saddle-point property in Proposition~\ref{prop:E_SW}-\ref{E_SW-b} to denote by 
	\begin{align}
	\sigma_R^\star := \argmin_{\sigma_B \in \mathcal{S}(B)} \sup_{\alpha\in [0,1] } 
	\frac{1-\alpha}{\alpha} \left( R + D_\alpha\left(\rho_{XB}\| \mathds{1}_X \otimes \sigma_B  \right) \right).
	\end{align}
	
	Next, we show that Eq.~\eqref{eq:cond_sharp} is satisfied, and thus we can exploit Proposition~\ref{prop:sharp_Hoeffding} to expand the right-hand side of Eq.~\eqref{eq:sp_SW1}.
	Let $r = \log |\mathcal{X}| - R$, and note that item \ref{E_SW-c} in Proposition \ref{prop:E_SW} implies
	\begin{align}
	\rho_{XB} \ll \tau_X \otimes \sigma_R^\star.
	\end{align}
	One can verify that
	\begin{align}
	\phi_n\left( r|\rho_{XB}^{\otimes n} \| \left( \tau_X\otimes \sigma_R^\star \right)^{\otimes n}  \right) 
	&=  \sup_{\alpha\in(0,1]} \frac{1-\alpha}{\alpha} \left( D_\alpha\left( \rho_{XB} \|\tau_X\otimes \sigma_R^\star \right) - r \right) \label{eq:sp_SW5} \\
	&=  \sup_{\alpha\in(0,1]} \frac{1-\alpha}{\alpha} \left( D_\alpha\left( \rho_{XB} \|\mathds{1}_X\otimes \sigma_R^\star \right) - \log|\mathcal{X}|- r \right) \\	
	&= E_\text{sp}(R)	\label{eq:sp_SW2} \\
	&> 0, \label{eq:sp_SW4}
	\end{align}
	where $\phi_n$ is defined in Eq.~\eqref{eq:cond_sharp}; equality \eqref{eq:sp_SW2} follows from the saddle-point property, item \ref{E_SW-b} in Proposition \ref{prop:E_SW}, and the definition of $E_\text{sp}(R)$ in Eq.~\eqref{eq:gallager_sp2}; the last inequality \eqref{eq:sp_SW4} is due to item \ref{E_SW-a} in Proposition \ref{prop:E_SW} and the given range of $R$.
	Further, the positivity of $	\phi_n\left( r|\rho_{XB}^{\otimes n} \| \left( \tau_X\otimes \sigma_R^\star \right)^{\otimes n}  \right)$ implies that $r> D_0(\rho_{XB}\| \tau_X\otimes \sigma_R^\star)$.
	By choosing $\eps =  r - D_0(\rho_{XB}\| \tau_X \otimes \sigma_R^\star ) > 0$, $\rho = \rho_{XB}$ and $\sigma = \tau_X\otimes \sigma_R^\star$, Eq.~\eqref{eq:sp_SW5} guarantees that Eq.~\eqref{eq:cond_sharp} is satisfied.
	Hence, we apply Eq.~\eqref{eq:sharp2} in Proposition~\ref{prop:sharp_Hoeffding} on Eq.~\eqref{eq:sp_SW1} to obtain	
	\begin{align}
	\log \left( \frac{1}{\Pestar (n, R )} \right)
	&\leq n \phi_n \left(r|\rho_{XB}^{\otimes n} \| \left( \tau_X\otimes \sigma_R^\star\right)^{\otimes n} \right) + \frac12\left( 1 + \left| \left.\frac{\partial \phi_n \left( \tilde{r}| \rho_{XB}^{\otimes n} \| \left( \tau_X\otimes \sigma_R^\star\right)^{\otimes n} \right)}{\partial \tilde{r}}\right|_{\tilde{r}=r} \right|  \right) \log n  + K, \label{eq:sp_SW3}
	\end{align}
where $K>0$ is some finite constant independent of $n$.
Finally, combining Eqs.~\eqref{eq:sp_SW2} and \eqref{eq:sp_SW3} completes the proof.
\end{proof}

\section{Optimal success exponent at a fixed rate below the Slepian-Wolf limit (Strong converse regime)} \label{sec:SC}
In this section, we investigate the case of rate below the Slepian-Wolf limit, i.e.~$R<H(X|B)_{\rho}$, which is analogous to the strong converse in channel coding. We establish both the finite blocklength converse and achievability bound in Section~\ref{subsec:SC_converse} and \ref{subsec:SC_achiev}. As a result of Theorem~\ref{thm:SC-converse-bound} and Theorem~\ref{thm:SC-achiev-bound} below, we are able to show that in the strong converse regime, $E_\text{sc}^*(R)$ characterizes the exponential decay of the probability of success:
\begin{equation}
\lim_{n\to\infty} sc(n,R)=\lim_{n\to\infty} -\frac1n \log \left( 1 - \Pestar(n,R) \right) = E_\text{sc}^*(R)
\end{equation}
where $E_\text{sc}^*(R)$ is defined by \eqref{eq:def_Esc} below, and $\Pestar(n,R)$ is defined by \eqref{def:optPe}.

\begin{remark} \label{remark_sc}
	Theorem~\ref{thm:SC-converse-bound} together with Corollary~\ref{cor:exact_sc} imply that the established $n$-shot strong converse converse bound is stronger than the results obtained in Ref.~\cite{leditzky_strong_2016} whenever $E_\text{sc}^*(R)$ is finite.
	To see this, let us call the exponent of the bounds in \cite[Theorem 6.2]{leditzky_strong_2016} by $E_\text{LWD}(R)$. Namely,  $E_\text{LWD}(R) \leq sc(n,R)$ for all $n$.
	Now, we assume for some rate $R > 0$ that $E_\text{sc}^*(R) < E_\text{LWD}(R)$.
	Then, it holds that
	\[
	E_\text{sc}^*(R) < E_\text{LWD}(R) \leq sc(n,R)
	\]
	for all $n$. 
	Taking $n$ to infinity and invoking Corollary~\ref{cor:exact_sc}, we have 
	\[
	E_\text{sc}^*(R) < E_\text{LWD}(R) \leq \lim_{n \to \infty} sc(n,R) = E^*_\text{sc}(R),
	\]
	which is a contradiction as long as $E_\text{sc}^*(R)$ is finite. Therefore, we conclude that
	\[E_\text{LWD}(R) \leq E_\text{sc}^*(R)
	\] 
	for any $R$ such that $E_\text{sc}^*(R)$ is finite.
\end{remark}

\subsection{Establishing a strong converse exponent} \label{subsec:SC_converse}
Let us recall Theorem~\ref{thm:SC-converse-bound}.

\SCconversebound*

\begin{proof}[Proof of Theorem~\ref{thm:SC-converse-bound}]
	Let $R\geq 0$. We first claim that any one-shot code $\sC$ with $R < H(X|B)_\rho$ satisfies
	\begin{equation}
	-\log \left( 1-\Pestar(1, R) \right) \leq  E_\text{sc}^*(R), \label{eq:strong1}
	\end{equation} 
	where $\rho_{XB}$ is a c-q state defined by \eqref{eq:rho_XB-cq}.

	Let $\sC$ be a code with encoder $\cE: \cX\to \cW$ and $\log|\cW| = R$, and decoder $\cD$. To bound the optimal error probability, we may reduce to the case of deterministic $\cE$, as discussed below \eqref{eq:Pe_rand_as_det}. The decoder is a family of POVMs $\mathcal{D} = \{\cD_w\}_{w\in \cW}$, where $\cD_w = \{\Pi_x^{(w)}\}_{x\in \cX}$.

	Let $\sigma_B\in\mathcal{S}(B)$. We will consider a two-outcome hypothesis test between $\rho_{XB}$ and $\tau_X\otimes \sigma_B$, where $\tau_X = \frac{\one}{|\cX|}$. Let us define the test
	\begin{equation}
		\Pi_{XB} := \sum_{x\in\mathcal{X}}  |x\rangle\langle x| \otimes \Pi_x^{(\cE(x))}.
	\end{equation}
						Then  $0\leq \Pi_{XB}\leq \one$ and moreover, 
	\begin{align}
	 \Tr\left[ \rho_{XB} \Pi_{XB} \right]&= 
	 \sum_{x\in\mathcal{X}} p(x)\Tr\left[ \rho_B^{x} \Pi_x^{(\cE(x))} \right] =  1 - \Pe(\sC)\\
	 \Tr\left[ \tau_X\otimes \sigma_B \Pi_{XB} \right] &= \frac{1}{|\cX|}\sum_{w\in \cs} \sum_{\substack{x\in \cX \\ \cE(x) = w}} \tr[\Pi_{x}^{(w)} \sigma] \leq  \frac{1}{|\cX|}\sum_{w\in \cs} \tr[\sigma_B]=   \frac{|\cs|}{|\mathcal{X}|} \leq 1.
	\end{align}
		Then, for all $\alpha>1$,
	\begin{align}
	\left( 1 - \Pe(\sC) \right)^\alpha \left(\frac{|\cs|}{|\mathcal{X}|}\right)^{1-\alpha} &\leq \left( \Tr\left[ \rho_{XB} \Pi_{XB} \right] \right)^{\alpha} \left( \Tr\left[ \tau_X\otimes \sigma_B\Pi_{XB} \right] \right)^{1-\alpha}\\
	&\leq \left( \Tr\left[ \rho_{XB} \Pi_{XB} \right] \right)^{\alpha} \left( \Tr\left[ \tau_X\otimes \sigma_B\Pi_{XB} \right] \right)^{1-\alpha} \nonumber \\ &\qquad +
	\left( \Tr\left[ \rho_{XB} (\mathds{1} - \Pi_{XB} ) \right] \right)^{\alpha} \left( \Tr\left[ \tau_X\otimes \sigma_B(\mathds{1} - \Pi_{XB}) \right] \right)^{1-\alpha}. \label{eq:Q_rho_gamma}
	\end{align}
Consider the measure-and-prepare map $\Phi: \mathcal{S}(\cH_X\otimes \cH_B) \to \mathcal{S}(\mathbb{C}^2)$ given by
\begin{equation}
 \Phi: \quad \eta_{XB} \mapsto \tr[\eta_{XB}\Pi_{XB}] \ket{0}\bra{0} + \tr[\eta_{XB} (\one - \Pi_{XB})] \ket{1}\bra{1}.
 \end{equation}
 Then we can recognize the right-hand side of \eqref{eq:Q_rho_gamma} as $Q_\alpha^*(\Phi(\rho_{XB}) \| \Phi(\tau_X\otimes \sigma_B))$. By the monotonicity of $Q_\alpha^*$ under CPTP maps,
 \begin{equation}
 Q_\alpha^*(\Phi(\rho_{XB}) \| \Phi(\tau_X\otimes \sigma_B)) \leq Q_\alpha^*( \rho_{XB} \| \tau_X\otimes \sigma_B).
 \end{equation}
	Note that this holds for every $\sigma_B\in\mathcal{S}(B)$. Thus, it follows for all $\alpha> 1$.
	\begin{align}
	\frac{\alpha}{\alpha - 1} \log \left( 1 - \Pe(\sC)\right) - \log \frac{|\cs|}{|\mathcal{X}|}
	\leq \inf_{\sigma_B\in\mathcal{S}(B)} D_\alpha^*\left(\rho_{XB} \left\|  \tau_X \otimes  \sigma_B\right. \right),
	\end{align}
	or equivalently
	\begin{align}
	\frac{\alpha}{\alpha - 1} \log \left( 1 - \Pe(\sC) \right) - R \leq - H_\alpha^{*,\uparrow}(X|B).
	\end{align}
	Since $\sC$ is arbitrary of rate $R$, we obtain Eq.~\eqref{eq:strong1}.
	By the additivity of $D_\alpha^*$ under tensor products, we find that any $n$-blocklength code with $R < H(X|B)_\rho$ satisfies
	\begin{align}
	sc(n,R) \geq E_\textnormal{sc}^*(R).
	\end{align}
\end{proof}

\subsection{Matching bound} \label{subsec:SC_achiev}
The main result in this section is the $n$-shot upper bound on the decay exponent of the probability of success in terms of the $E_\text{sc}^*$ and additional residual terms, using a proof based on Mosonyi and Ogawa's proof of an analogous result in study of the transmission of classical information over quantum channels  \cite{MO17}.
Let us recall Theorem~\ref{thm:SC-achiev-bound}:

\SCachievbound*

This is a consequence of the following, more detailed, result.
\begin{prop}[$n$-shot Strong Converse Matching Bound] \label{prop:full-SC-achiev-bound}
	Let $R<H(X|B)_{\rho}$, $m\in \N$ and $\delta \in (0, \frac{1}{4})$. For all $n> m$, we have the bound
								\begin{equation}
	 sc(n,R) \leq  E_\textnormal{sc}^*(R) + \frac{c}{m}\log(m+1) + \frac{1}{\sqrt{n-1}} f_1(m,\delta) + \frac{1}{n-m}f_2(n,m)
	\end{equation}
	for 
	\begin{align}
	{E}_\textnormal{sc}^*(R) &:= \sup_{\alpha> 1} \frac{1-\alpha}{\alpha} \left( R - H_{\alpha}^{*,\uparrow}(X|B)_\rho \right) \\
	c &:=  \frac{3(|\cH_B|+2)(|\cH_B|-1)}{2} \\
	\sqrt{m}f_1(m,\delta) &:= \sqrt{3\log 2} \left( 4 \cosh (1+2\sqrt{|\mathcal{X}|^m}) [\log(1+2\sqrt{|\mathcal{X}|^m})]^2  + 1\right) \nonumber\\
	&\qquad +  4\sqrt{2} \left(\log \left(1 + \tr[P_{\texttt{supp}(\rho_m)}\rho_{m}^{-1/2} + \rho_{m}^{1/2}]\right)\right)\log \frac{1}{1 - \sqrt{1-\delta^2}} \\
	f_2(n,m) &:= 1- \log( 1 - 2^{- (\frac{n}{m}-1)R})
	\end{align}
	with $\rho_m := \cP_m(\rho_{XB}^{\otimes m})$, where $\cP_m$ is the pinching map associated to $\one_{X^m}\otimes \sigma_{\text{u},m}$,  $\sigma_{\text{u},m}$ is the universal symmetric state on $\cH_B^{\otimes m}$ given in Lemma~\ref{lem:univ-symm-state} below, and $H_{\alpha}^{*,\uparrow}(X|B)_\rho:= \max_{\sigma_B\in\mathcal{S}(B)} - D_\alpha^*(\rho_{XB}\|\mathds{1}_X \otimes \sigma_B)$ for $D_\alpha^*$ being sandwiched R\'enyi divergence, see Eq.~\eqref{eq:sandwich}.
	In particular,
	\begin{equation}
	\limsup_{n\to\infty} sc(n,R) \leq  E_\textnormal{sc}^*( R) .
	\end{equation}
\end{prop}

To establish Proposition~\ref{prop:full-SC-achiev-bound}, we first obtain a bound with $E_\text{sc}^{\flat}$ (Proposition~\ref{prop:SC-achiev-flat-bound}), and then exploit a pinching argument to further relate $E_\text{sc}^{\flat}$ to the desired $E_\text{sc}^*$. The proof of Proposition~\ref{prop:SC-achiev-flat-bound} and Proposition~\ref{prop:full-SC-achiev-bound} are delayed to Section~\ref{subsubsec:proof_SC_flat} and \ref{subsubsec:proof_SC_pinching}, respectively

\begin{prop} \label{prop:SC-achiev-flat-bound}
	
	For any $n\in \N$, $R>0$, and $\delta \in (0,\frac{1}{4})$, we have
	\begin{equation}
	sc(n,R) \leq  E_\textnormal{sc}^\flat(R) +\frac{1}{\sqrt{n}} [ \sqrt{4\log 2} (e_2 + 1) + e_1(\delta)]  + \frac{1}{n} \left[1- \log( 1 - 2^{-nR}) \right]
	\end{equation}
	where
	\begin{gather}
	e_1(\delta) :=   4\sqrt{2} \left(\log \left(1 + \tr[P_{\texttt{\textnormal{supp}}(\rho_{XB})}\rho_{XB}^{-1/2} + \rho_{XB}^{1/2}]\right)\right)\log \frac{1}{1 - \sqrt{1-\delta^2}} \label{eq:def_e1}\\
	e_2 := 4 \cosh (1+2\sqrt{|\mathcal{X}|}) [\log(1+2\sqrt{|\mathcal{X}|})]^2. \label{eq:def_e2}
	\end{gather}
	In particular,
	\begin{equation}
	\limsup_{n\to\infty} sc(n,R) \leq E_\textnormal{sc}^\flat(R).
	\end{equation}
\end{prop}

\subsubsection{Proof of Proposition~\ref{prop:SC-achiev-flat-bound}} \label{subsubsec:proof_SC_flat}
Using the variational representation, Proposition~\ref{prop:representation}, one has
\begin{equation} \label{eq:F=tildeE}
E_\text{sc}^\flat(R) = \min_{\sigma \in \mathcal{S}_{\rho}(XB)} \{D (\sigma_{XB}\|\rho_{XB}) + |H(X|B)_\sigma - R|^+)\}.
\end{equation}
The main idea is to use this variational representation to introduce a ``dummy state'' $\sigma_{XB}$ for which we can apply Theorem~\ref{theo:large_ach}. Then we relate the probability of success for this state to that of the source state $\rho_{XB}$.

To obtain explicit bounds $sc(n,R)$, we will proceed in the $n$-shot setting. Let us define an $n$-shot analog of the right-hand side for $\alpha\in[0,1)$, $\gamma>0$, and $\delta >0$ by
\begin{equation}
F(R,\rho_{XB},n,\delta,\alpha,\gamma) := \min_{\sigma \in \mathcal{S}_{\rho}(XB)} \left\{ \frac{1}{n}D_\text{max}^\delta (\sigma_{XB}^{\otimes n}\|\rho_{XB}^{\otimes n}) + |H^\downarrow_\alpha(X|B)_\sigma - R + \gamma|^+ \right\}
\end{equation}
where $D_\text{max}^\delta(\sigma\|\rho)$ indicates the max-relative entropy smoothed by $\delta$ (defined by \eqref{eq:Dmax_smoothed}) in the distance of optimal purifications, denoted $d_\text{op}$ (defined in \eqref{eq:def_dist_purifications}).  This quantity is upper bounded by $ E_\text{sc}^\flat$ (with error terms) in Lemma~\ref{lem:F-n-shot-bound} below. First, we establish the preliminary result given by Lemma~\ref{lem:dummy-state}. Next we bound the strong converse error exponent by $F(R,n,\delta,\alpha,\gamma)$ in two steps (Lemmas \ref{prop:n-shot-F1} and \ref{prop:n-shot-F2} below), and use Lemma \ref{lem:F-n-shot-bound} to bound with $ E^\flat_\text{sc}$.

The following result allows us to compare the success probability of the same code $\cC$ for  $\rho_{XB}$ and a dummy state $\sigma_{XB}\in S_\rho(XB)$.
\begin{lemm}	\label{lem:dummy-state}
	Let $\sigma_{XB} \in S_\rho(XB)$. For any code $\sC$, and any $a>0$ we have
	
	\begin{equation}
	\Ps(\rho_{XB}, \sC) \geq \e^{- a} \left( \Ps(\sigma_{XB}, \sC)  - \tr[(\sigma_{XB} - \e^{a}\rho_{XB})_+]  \right).
	\end{equation}
	Here, we write $\Ps(\rho_{XB}, \sC)$ and $ \Ps(\sigma_{XB}, \sC)$ to emphasize the dependency on the state, which is taken to be $\rho_{XB}$ when not explicitly indicated.
\end{lemm}

\begin{proof}[Proof of Lemma~\ref{lem:dummy-state}] Let $\sC$ be a coder with encoder $\cE$ and a decoder given by the family of POVMs $\mathcal{D} = \{\cD_w\}_{w\in\cs}$ where $\cD_w = \{\Pi_x^{(w)}\}_{x\in \cX}$.  Let $\Pi_{XB} = \sum_{x\in \cX} \ket{x}\bra{x}\otimes \Pi_x^{(\cE(x))}$. Then
\begin{equation}
\tr[\Pi_{XB}\rho_{XB}] = \Ps(\rho_{XB}, \sC), \qquad \tr[\Pi_{XB}\sigma_{XB}] = \Ps(\sigma_{XB}, \sC).
\end{equation}
For any self-adjoint operator $X$ and $Y$ with $0\leq Y\leq \one$,
\begin{equation} \label{eq:pos_part_ieq}
\tr[X_+] \geq \tr[X_+Y] \geq \tr[X_+Y] - \tr[X_- Y] = \tr[XY].
\end{equation}
Since $0\leq \Pi_{XB} = \sum_{x\in \cX} \ket{x}\bra{x}\otimes \Pi_x^{(\cE(x))} \leq \sum_{x\in \cX} \ket{x}\bra{x}\otimes \one_B = \one_{XB}$, eq.~\eqref{eq:pos_part_ieq} with $X = (\sigma_{XB}-\e^a \sigma_{XB})$ and $Y = \Pi_{XB}$ yields
\begin{equation}
\tr[(\sigma_{XB}-\e^a \sigma_{XB})_+] \geq \tr[ \Pi_{XB}(\sigma_{XB}-\e^a \sigma_{XB})] = \Ps(\sigma_{XB},\sC) - \e^{a}\Ps(\rho_{XB},\sC)
\end{equation}
which yields the result after dividing by $\e^a$.
\end{proof}

In the following, we prove Theorem \ref{prop:SC-achiev-flat-bound} using two lemmas, analogous to Lemmas V.9 and V.10 of \cite{MO17}.
Let us define 
\begin{align}
F_1(R,n,\delta,\alpha,\gamma) &:= \min_{\sigma_{XB}: H^\downarrow_{\alpha}(X|B)_\sigma \leq R - \gamma}\frac{1}{n}D_\text{max}^\delta (\sigma_{XB}^{\otimes n}\|\rho_{XB}^{\otimes n}).\\
F_2(R,n,\delta,\alpha,\gamma) &:= \min_{\sigma_{XB}:H^\downarrow_\alpha(X|B)_\sigma \geq R - \gamma} \frac{1}{n}D_\text{max}^\delta (\sigma_{XB}^{\otimes n}\|\rho_{XB}^{\otimes n}) + H^\downarrow_\alpha(X|B)_\sigma - R + \gamma
\end{align}
where the minima are restricted to $\sigma_{XB} \in \mathcal{S}_{\rho}(XB)$. Then
\begin{equation} \label{eq:F-min-F1-F2}
F(R,n,\delta,\alpha,\gamma) = \min \left\{ F_1(R,n,\delta,\alpha,\gamma),F_2(R,n,\delta,\alpha,\gamma) \right\}.
\end{equation}
First, we will bound the error probability using $F_1$.

\begin{lemm}	 \label{prop:n-shot-F1}
	Let $\rho_{XB}$ be a c-q source state. For any $\delta \in(0,\frac{1}{2})$, $\gamma > 0$, $\alpha\in [0,1)$ and $n\geq 1$,
	\begin{equation} \label{F1-$n$-shot-bound}
	sc(n,R)\leq  F_1(R,n,\delta,\alpha,\gamma) - \frac{1}{n}\log \left( 1 - 4 \e^{-n \gamma (1-\alpha)} -2\delta\right) .
	\end{equation}
				\end{lemm}
\begin{proof}[Proof of Lemma~\ref{prop:n-shot-F1}]
	Let $r> F_1(R,n,\delta,\alpha,\gamma)$. Then there exists $\sigma_{XB}$ such that 
	\begin{equation}
	\frac{1}{n}D_\text{max}^\delta (\sigma_{XB}^{\otimes n}\|\rho_{XB}^{\otimes n}) \leq r \label{eq:F1-r}
	\end{equation}
	and
	\begin{equation} \label{eq:F1-R}
	H^\downarrow_\alpha(X|B)_\sigma \leq R-\gamma.
	\end{equation}

	By Lemma \ref{lem:dummy-state} with $a=nr$, we have
	\begin{equation}
	\Ps(\rho_{XB}^{\otimes n}, \sC_n) \geq \e^{-n r} \left( \Ps(\sigma_{XB}^{\otimes n}, \sC_n)  - \tr[(\sigma_{XB}^{\otimes n} - \e^{nr}\rho_{XB}^{\otimes n})_+]  \right)
	\end{equation}
	for any code $\sC_n$. Then, since $D_\text{max}^\delta (\sigma_{XB}^{\otimes n}\|\rho_{XB}^{\otimes n}) \leq nr$, we have that
	$\tr[(\sigma_{XB}^{\otimes n} - \e^{nr}\rho_{XB}^{\otimes n})_+] \leq 2\delta$ by Lemma~\ref{lem:$n$-shot-stein-exp}. Additionally, since $\alpha\mapsto H^\downarrow_\alpha(X|B)_\sigma$ is monotone decreasing (see Proposition~\ref{prop:H}), we have that $H(X|B)_\sigma \leq H^\downarrow_\alpha(X|B)_\sigma \leq R$. Thus, we can choose a code $\sC_n$ with rate $R$ such that
	\begin{equation}
	\Ps(\sigma_{XB}^{\otimes n},\sC_n) \geq 1- 4 \exp (-n  E_\text{r}^{\downarrow}(\sigma_{XB},R))
	\end{equation}
	by Theorem~\ref{theo:large_ach}. Moreover,
	\begin{equation}
	  E_r^\downarrow(\sigma_{XB},R) = \sup_{s\in[0,1]} - s [H^\downarrow_{1-s}(X|B)_\sigma - R] \geq (1-\alpha)[R - H^\downarrow_\alpha(X|B)_\sigma] \geq \gamma(1-\alpha),
	\end{equation}
	since $\alpha\in [0,1)$. Therefore,
	\begin{equation} \label{eq:F1-good-code}
	1-\Pe(\rho_{XB}^{\otimes n}, \sC_n) =\Ps(\rho_{XB}^{\otimes n},\sC_n) \geq \e^{-nr} \left( 1 - 4  \e^{-n \gamma(1-\alpha)} -2 \delta \right).
	\end{equation}
	Then since $\Pestar(n,R) \leq \Pe(\rho_{XB}^{\otimes n}, \sC_n)$,
	\begin{equation}
	- \frac{1}{n}\log ( 1 - \Pestar(n,R)) \leq r - \frac{1}{n}\log \left( 1 - 4 \e^{-n \gamma(1-\alpha)} -2\delta\right).
	\end{equation}
	Since this is true for any $r > F_1(R,n,\delta,\alpha,\gamma)$, we have \eqref{F1-$n$-shot-bound}. 
\end{proof}

\begin{lemm}	\label{prop:n-shot-F2}
	For all $n\geq 1$, $\delta \in (0,\frac{1}{2})$, $\alpha\in [0,1)$, and $\gamma > 0$, we have
	\begin{equation} \label{eq:F2-bound}
	sc(n,R) \leq F_2(R,n,\delta,\alpha,\gamma) - \frac{1}{n}\log\left[(1-2^{-nR}) \left( 1 - 4  \e^{-n \gamma(1-\alpha)} - 2\delta \right)\right] .
	\end{equation}
				\end{lemm}
\begin{proof}[Proof of Lemma~\ref{prop:n-shot-F2}]
	Let $r > F_2(R,\rho_{XB},n,\delta,\alpha,\gamma)$. Then there exists $\sigma_{XB}$ such that
	\begin{equation} \label{eq:F2-r-satisfies}
	r > \frac{1}{n}D_\text{max}^\delta (\sigma_{XB}^{\otimes n}\|\rho_{XB}^{\otimes n}) + H^\downarrow_\alpha(X|B)_\sigma - R + \gamma
	\end{equation}
	and $R - \gamma\leq H^\downarrow_\alpha(X|B)_\sigma$. Let 
	\begin{gather}
	R_1 := H^\downarrow_\alpha(X|B)_\sigma +  \gamma \\
	r_1 := r +R -H^\downarrow_\alpha(X|B)_\sigma -  \gamma
	\end{gather}
	which has $r_1 \geq \frac{1}{n}D_\text{max}^\delta (\sigma_{XB}^{\otimes n}\|\rho_{XB}^{\otimes n})$ by \eqref{eq:F2-r-satisfies}.
	We see then that equation~\eqref{eq:F1-r} holds for the state $\sigma_{XB}$ when  $r$ is replaced by $r_1$, and \eqref{eq:F1-R} holds for $\sigma_{XB}$ when $R$ is replaced by $R_1$. Therefore, \eqref{eq:F1-good-code} yields that there is a code $\sC_n$ with rate $R_1$ such that
	\begin{equation} \label{eq:long-code-F2}
	\Ps(\rho_{XB}^{\otimes n},\sC_n) \geq \e^{-nr_1} \left( 1 - 4  \e^{-n \gamma(1-\alpha)} - 2\delta \right).
	\end{equation} 
	
	We have $R \leq R_1$ by construction. We aim to construct a code $\tilde \cC_n$ with rate $R$ and a high probability of success by pruning the low probability elements of $\cW$, as follows. If $R= R_1$, take $\tilde \sC_n = \sC_n$. Then \eqref{eq:F2-bound-success-tilde-C} holds, and we may continue the proof from there. Otherwise, $R<R_1$. Let $\sC_n$ have encoder $\cE_n:\cX^n \to \cW_n$ and decoding POVMs $\cD_w = \{\Pi^{(w)}_{\vec x}\}_{\vec x\in \cX^n}$, for $w\in \cW_n$, where $|\cW_n|=2^{n R_1}$. Then given a sequence $\vec x\in \cX^n$, the probability of correctly decoding the sequence is given by
	\begin{equation}
	\tr[ \rho_B^{\vec x} \Pi_{\vec x}^{(\cE_n(\vec x))}]
	\end{equation}
	where $\rho_B^{\vec x} = \rho_B^{x_1}\otimes \rho_B^{x_2}\otimes \dotsm \otimes \rho_B^{x_n}$. Similarly, writing $p(\vec x) = p(x_1)\dotsm p(x_n)$, the quantity
	\begin{equation}
	P_w := \sum_{\vec x\in \cX^n: \,\cE_n(\vec x) = w} p(\vec x) \tr[ \rho_B^{\vec x} \Pi_{\vec x}^{(w)}]
	\end{equation}
	is the probability of success for sequences which are encoded as $w$. Then
	\begin{equation}
	\Ps(\rho_{XB}^{\otimes n}, \sC_n) = \sum_{w \in \cW_n} P_w.
	\end{equation}
	Therefore, $\vec P := \left( \frac{P_w}{ \Ps(\rho_{XB}^{\otimes n}, \sC_n)} \right)_{w\in \cW_n}$ is a probability vector of length $2^{n R_1}$. Let $k = 2^{n R_1} - 2^{n R}$. In order to make a code of rate $R$ from $\sC_n$, we need to remove $k$ elements from $\cW_n$. We will choose the elements of $\cW_n$ corresponding to the $k$ smallest entries of the vector $\vec P$, to keep as much probability of success as possible. Without loss of generality, assume $\vec P$ is in decreasing order: $\vec P_1\geq \vec P_2 \geq \ldots$. 
	
	Let $\cR_n := \{ j  : j \geq 2^{nR}\}$ be the set of indices to remove (note $|\cR_n| = k$). Define $\tilde \cW_n = \cW_n \backslash \cR_n$. Choose $w_0\in W$ to be the element of $\cW$ with index $2^{n R} -1$, that is, with index corresponding to the smallest element of $\vec P$ which has not been removed. Define the encoder
	\begin{equation}
	\tilde \cE_n ( x ) = \begin{cases}
	w_0 & \text{ if } \cE_n(x)\in \cR_n \\
	\cE_n(x) & \text{else}.
	\end{cases}
	\end{equation}
	Then $\tilde \cE_n$ and the decoding POVMs $\{\cD_w\}_{w \in \tilde \cW_n}$ forms a code $\tilde \sC_n$ of rate $R$, since $|\tilde \cW_n| = 2^{n R_1} - k = 2^{n R}$.
	
	Let us briefly introduce the majorization pre-order of vectors. Given $\vec x,\vec y\in \mathbb{R}^n$, we say $\vec y$ majorizes $\vec x$, written $\vec x \prec \vec y$, if $\sum_{i=1}^k \vec y_k^\downarrow \geq \sum_{i=1}^k \vec x_k^\downarrow$ for each $k=1,\dotsc,n$, with equality for the case $k=n$, where $\vec y^\downarrow$ is the rearrangement of $\vec y$ in decreasing order: $\vec y^\downarrow_1 \geq \dotsm \geq \vec y_n$.

	Setting $\vec u = \left( \frac{1}{2^{n R_1}} \right)_{w \in \cW_n}$ as the uniform distribution, we have
	\begin{equation}
	\vec u \prec  \vec P
	\end{equation}
	since the uniform distribution is majorized by every other probability vector. In particular,
	\begin{equation}
	\sum_{j=1}^{2^{n R_1}-k-1} \vec P^\downarrow_j \geq \frac{2^{n R_1}-k-1}{2^{nR_1}}=\frac{2^{n R}-1}{2^{nR_1}}.
	\end{equation}
	The probability of success for any element which $\tilde\cE_n$ does not map to $\omega_0$ is the same as it was under the code $\sC_n$. Therefore,
	\begin{align}
	\Ps(\rho_{XB}^{\otimes n}, \tilde \sC_n) &= \sum_{w\in \tilde \cW_n \setminus \{\omega_0\}} P_w + \sum_{x\in \cX^n:\, \tilde \cE_n(x) = \omega_0}\tr[\rho_B^x \Pi^{\omega_0}_x]\\
	&\geq \sum_{w\in \tilde \cW_n \setminus \{\omega_0\}} P_w =\Ps(\rho_{XB}^{\otimes n}, \sC_n) \sum_{j=1}^{2^{n R_1}-k-1} \vec P^\downarrow_j \\
	&\geq \frac{2^{nR}-1}{2^{n R_1}}\Ps(\rho_{XB}^{\otimes n}, \sC_n)\\
	&= \frac{2^{nR}(1-2^{-nR})}{2^{n R_1}}\Ps(\rho_{XB}^{\otimes n}, \sC_n)
	\end{align}
	By \eqref{eq:long-code-F2}, we have therefore
	\begin{equation} \label{eq:F2-bound-success-tilde-C}
	\Ps(\rho_{XB}^{\otimes n}, \tilde \sC_n) \geq \frac{2^{nR}(1-2^{-nR})}{2^{n R_1}} \e^{-nr_1} \left( 1 - 4  \e^{-n \gamma(1-\alpha)} - 2\delta \right).
	\end{equation}
	Since $1 - \Pestar(n,R)\leq   \Ps(\rho_{XB}^{\otimes n}, \tilde \sC_n)$, by substituting the definition of $sc(n,R)$ from \eqref{eq:def_sc} we find
	\begin{align}
	sc(n,R)&\leq- \frac{1}{n}\log\left(\frac{2^{nR}(1-2^{-nR})}{2^{n R_1}} \right) + r_1 - \frac{1}{n}\log \left( 1 - 4  \e^{-n \gamma(1-\alpha)} - 2\delta \right)\\
		&=r_1+  R_1 - R - \frac{1}{n}\log\left[(1-2^{-nR}) \left( 1 - 4  \e^{-n \gamma(1-\alpha)} -2 \delta \right)\right]\\
	 &\leq r - \frac{1}{n}\log\left[(1-2^{-nR}) \left( 1 - 4  \e^{-n \gamma(1-\alpha)} - 2\delta \right)\right] 
	\end{align}
	using $r_1+  R_1 - R = r$.
	Since this inequality holds for all $r > F_2(R,n,\delta,\alpha,\gamma)$ we have \eqref{eq:F2-bound}.
\end{proof}

We can see as $n\to\infty$, $\alpha\to 1$, and $\gamma\to 0$,
\begin{equation}
F(R,n,\delta,\alpha,\gamma) \to  E_\text{sc}^\flat(\rho,R)
\end{equation}
using the variational characterization of $ E_\text{sc}^\flat(R)$ given in \eqref{eq:F=tildeE}. However, to find an $n$-shot bound on the error exponent in terms of $ E_\text{sc}^\flat$, we need to find the error terms hidden in this limit. \begin{lemm} \label{lem:F-n-shot-bound}
	If $R>0$, $\delta\in(0,\frac{1}{2})$, $\gamma > 0$,  $n\in \N$ and $\alpha \in (\frac{1}{2},1)$, then
	\begin{equation} \label{eq:F-n-shot-bound}
	F(R,n,\delta,\alpha,\gamma) \leq E_\text{sc}^\flat(R) + \gamma + \frac{1}{\sqrt {n}}e_1(\delta)  + (1-\alpha) e_2
	\end{equation}
	where $e_1(\delta)$ and $e_2$ are defined in \eqref{eq:def_e1} and \eqref{eq:def_e2}.
\end{lemm}
\begin{proof}[Proof of Lemma~\ref{lem:F-n-shot-bound}]
	Consider the variational representation
	\begin{equation} \label{eq:def_F-infinity}
	E_\text{sc}^\flat(R) = \min_{\sigma \in \mathcal{S}_{\rho}(XB)}  \left\{ D(\sigma_{XB}\|\rho_{XB}) + |H(X|B)_\sigma - R|^+ \right\}
	\end{equation}
	given by Proposition~\ref{prop:representation}.
	Let $\sigma^*_{XB}$ achieve the minimum in \eqref{eq:def_F-infinity}. Note $\sigma^*_{XB}$ depends on $\rho_{XB}$ and $R$. Then
	\begin{align}
	F(R,\rho_{XB},n,\delta,\alpha,\gamma) &= \min_{\sigma_{XB}\in \mathcal{S}_\rho(XB)} \left\{ \frac{1}{n}D_\text{max}^\delta (\sigma_{XB}^{\otimes n}\|\rho_{XB}^{\otimes n}) + |H^\downarrow_\alpha(X|B)_\sigma - R + \gamma|^+ \right\} \\
	&\leq  \frac{1}{n}D_\text{max}^\delta (\sigma_{XB}^*{}^{\otimes n}\|\rho_{XB}^{\otimes n}) + |H^\downarrow_\alpha(X|B)_{\sigma^*} - R + \gamma|^+.
	\end{align}
	By \eqref{eq:Ddelta-AEP-UB}, we have
	\begin{equation}
	\frac{1}{n}D_\text{max}^\delta (\sigma_{XB}^{*\otimes n}\|\rho_{XB}^{\otimes n}) \leq D(\sigma_{XB}^*\|\rho_{XB}) + \frac{1}{\sqrt {n}} 4 \sqrt{2}(\log \eta^*)\log \frac{1}{1 - \sqrt{1-\delta^2}}
	\end{equation}
	where
	\begin{equation}
	\eta^* = 1 + \tr[ \sigma_{XB}^*{}^{3/2}\rho_{XB}^{-1/2}] + \tr[\sigma_{XB}^*{}^{1/2} \rho_{XB}^{1/2}]
	\end{equation}
	We can remove the dependence on $\sigma^*_{XB}$ by the inequality
	\begin{equation}
	\eta^* \leq 1 + \tr[ P_{\texttt{supp}(\rho)}\rho_{XB}^{-1/2} + \rho_{XB}^{1/2}]
	\end{equation}
	which follows from $\sigma^*_{XB}{}^\alpha \leq P_{\texttt{supp}(\sigma)}$, the orthogonal projection onto the support of $\sigma$, which holds for all $\alpha > 0$, and since $P_{\texttt{supp}(\sigma)}\leq P_{\texttt{supp}(\rho)}$ as $\sigma_{XB}^* \in \mathcal{S}_\rho(XB)$. In particular, this removes the dependence on the rate $R$. Therefore we have
	\begin{equation} \label{eq:Dmax-bound}
	\frac{1}{n}D_\text{max}^\delta (\sigma_{XB}^{*\otimes n}\|\rho_{XB}^{\otimes n}) \leq D(\sigma_{XB}^*\|\rho_{XB}) + \frac{1}{\sqrt{n}} e_1(\delta).
	\end{equation}

	Next,  we use a continuity bound in $\alpha$ for $H_\alpha^\downarrow$: by Lemma 2.3 of \cite{AMV12}, for any $c>0$,
	\begin{equation}
	H^\downarrow_\alpha(X|B)_{\sigma^*} \leq H(X|B)_{\sigma^*} + 4 \cosh c (1-\alpha)(\log \eta')^2
	\end{equation}
	where
	\begin{equation}
	\eta' := 1 + \e^{\frac{1}{2}H^\downarrow_{3/2}(X|B)_{\sigma^*}} + \e^{- \frac{1}{2}H^\downarrow_{1/2} ( X|B)_{\sigma^*}}
	\end{equation}
	which holds when $|1-\alpha| \leq \min \left\{ \frac{1}{2}, \frac{c}{2\log \eta'} \right\}$. Note that for any state $\omega$ and $\alpha \in [0,2]$, we have
	\begin{equation}
	- \log \min(\rank \omega_X,\rank\omega_B) \leq H^\downarrow_\alpha(X|B)_\omega \leq \log \rank \omega_X
	\end{equation}
	by Lemma 5.2 of \cite{Tom16}. Thus,
	\begin{equation} \label{eq:bounds-etap}
	1 + \frac{1}{\sqrt{\min (|\mathcal{X}|, |B|) )}} + \frac{1}{\sqrt{|\mathcal{X}|}} \leq \eta' \leq 1 + \sqrt{|\mathcal{X}|} +\sqrt{\min (|\mathcal{X}|, |B|) )} \leq 1 + 2\sqrt{|\mathcal{X}|}
	\end{equation}
	independently of $\sigma^*_{XB}$. Then $\log \eta' > 0$,  and 
	\begin{equation} 
	(\log \eta')^2 \leq [\log(1+2\sqrt{|\mathcal{X}|})]^2.
	\end{equation}
	Moreover, since $\log \eta' > 0$, we can take $c= \log \eta'$ to recover for $\alpha \in (\frac{1}{2},1)$,
	\begin{align} 
	H^\downarrow_\alpha(X|B)_{\sigma^*} &\leq H(X|B)_{\sigma^*} + 4 \cosh (\eta') (1-\alpha)(\log \eta')^2 \\
	&\leq H(X|B)_{\sigma^*} +(1-\alpha) e_2\label{Hdownalpha-near-1-bound}
	\end{align}
	using \eqref{eq:bounds-etap}.

	Then for any $R>0$, $\delta\in(0,\frac{1}{2})$, $\gamma > 0$, and  $\alpha \in ( \frac{1}{2},1)$ such that  $H^\downarrow_\alpha(X|B)_{\sigma^*} - R + \gamma \geq 0$, we have
	\begin{align}
	F &= F(R,n,\delta,\alpha,\gamma)\\
	&\leq D(\sigma_{XB}^* \|\rho_{XB}) + H^\downarrow_\alpha(X|B)_{\sigma^*} -R + \gamma + \frac{1}{\sqrt {n}}e_1(\delta) \label{eq:using-Dmax-bound} \\
	&\leq D(\sigma_{XB}^* \|\rho_{XB}) + H(X|B)_{\sigma^*}  -R + \gamma + \frac{1}{\sqrt {n}}e_1(\delta)  + (1-\alpha) e_2 \label{eq:using-Hdown-bound}\\
	&\leq D(\sigma_{XB}^* \|\rho_{XB}) + |H(X|B)_{\sigma^*}  -R |_+ + \gamma + \frac{1}{\sqrt {n}}e_1(\delta)  + (1-\alpha) e_2\\
	&=  E_\text{sc}^\flat(R) + \gamma + \frac{1}{\sqrt {n}}e_1(\delta)  + (1-\alpha) e_2
	\end{align}
	where in \eqref{eq:using-Dmax-bound} we use \eqref{eq:Dmax-bound}, and in \eqref{eq:using-Hdown-bound} we use \eqref{Hdownalpha-near-1-bound}.

	On the other hand, if $R>0$, $\delta\in(0,\frac{1}{2})$, $\gamma > 0$,  and $\alpha \in (\frac{1}{2},1)$, and   $H^\downarrow_\alpha(X|B)_{\sigma^*} - R + \gamma < 0$, then
	\begin{align}
	F &= F(R, n,\delta,\alpha,\gamma)\\
	&\leq D(\sigma_{XB}^* \|\rho_{XB}) + \frac{1}{\sqrt {n}}e_1(\delta) \label{eq:using-Dmax-2}\\
	&=  E_\text{sc}^\flat(R) + \frac{1}{\sqrt {n}}e_1(\delta) \label{eq:using_pos_part_zero}
	\end{align}
	where in \eqref{eq:using-Dmax-2} we use \eqref{eq:Dmax-bound}, and in \eqref{eq:using_pos_part_zero} we use that $|H^\downarrow_\alpha(X|B)_{\sigma^*} - R|_+ = 0$. 
	In either case then, we recover \eqref{eq:F-n-shot-bound}.
\end{proof}

With these results, we may prove Proposition~\ref{prop:SC-achiev-flat-bound}.
\begin{proof}[Proof of Proposition~\ref{prop:SC-achiev-flat-bound}]
	By \eqref{eq:F-min-F1-F2} and Lemmas~\ref{prop:n-shot-F1} and \ref{prop:n-shot-F2}, we have the bound
	\begin{equation}
	sc(n,R) \leq F(R,n,\delta,\alpha,\gamma)- \frac{1}{n}\log\left[(1-2^{-nR}) \left( 1 - 4  \e^{-n \gamma(1-\alpha)} - \delta \right)\right] \label{eq:full-F-bound}
	\end{equation}
	for any $R> 0$, $n\geq 1$, $\delta \in (0,\frac{1}{2})$, $\alpha \in (0,1]$, and $\gamma > 0$. 	We choose $\delta\in(0,1/4)$ and $\kappa = 4 \log 2 = \log 16$. Then $4\e^{-\kappa}+\delta \leq \frac{1}{2}$.
	For $x\in [0,\frac{1}{2}]$, we have the bound $-\log(1-x)\leq 2x$, and therefore
	\begin{equation}
	- \log \left( 1-4\e^{-\kappa} -\delta \right) \leq 2 \cdot \frac{1}{2}=1.
	\end{equation}
	Now, choose $\alpha = 1- \frac{\sqrt{\kappa}}{\sqrt{n}}$, and $\gamma =\sqrt{\frac{\kappa}{n}}$. Then for $\delta\in(0,\frac{1}{4})$ by \eqref{eq:full-F-bound}, we have
	\begin{align}
	sc(n,R) \leq F(R, n,\delta,\alpha,\gamma)   + \frac{1}{n} \left[1- \log( 1 - 2^{-nR}) \right].
	\end{align}
	Thus, by Lemma~\ref{lem:F-n-shot-bound}, we have
	\begin{equation}
	sc(n,R) \leq E_\text{sc}^\flat(R) + \frac{\sqrt{\kappa}}{\sqrt{n}}[1+e_2] + \frac{1}{\sqrt {n}}e_1(\delta)  + \frac{1}{n} \left[1- \log( 1 - 2^{-nR}) \right]
	\end{equation}
	where $e_1(\delta)$ and $e_2$ are defined in \eqref{eq:def_e1} and \eqref{eq:def_e2}. Taking $n\to\infty$ recovers the asymptotic bound.
\end{proof}

To prove Proposition~\ref{prop:full-SC-achiev-bound}, we need to introduce the idea of universal symmetric state.
The universal symmetric state is obtained via representation theory (see e.g.~\cite[Appendix A]{MO17}). Briefly, $\cH_B^{\otimes m}$ is decomposed into a direct sum of tensor products of irreducible representations $U_\lambda$ of the symmetric group on $m$ letters and $V_\lambda$ of the special unitary group of dimension $d:=|\cH_B|$,
\[
\cH_B^{\otimes m} \cong \bigoplus_{\lambda \in Y_m^d} U_\lambda \otimes V_\lambda
\]
where the index set $Y_m^d$ consists of Young diagrams of depth up to $d$. The state $\sigma_{\text{u},m}$ is constructed by taking the maximally mixed state on each representation:
\begin{equation}
\sigma_{\text{u},m} = \bigoplus_{\lambda\in Y_m^d} \frac{1}{|Y_{m}^d|} \frac{\one_{U_\lambda}}{\dim U_\lambda} \otimes \frac{\one_{V_\lambda}}{\dim V_\lambda}.
\end{equation}

\begin{lemm}[Lemma II.8 of \cite{MO17}] \label{lem:univ-symm-state}
	There is a symmetric state $\sigma_{\text{u},m}$ on $\cH^{\otimes m}$ such that for any other symmetric state $\omega$ on $\cH^{\otimes m}$, we have that $\sigma_{\text{u},m}$ commutes with $\omega$ and obeys the bound
	\begin{equation}
	\omega\leq v_{m,d} \sigma_{\text{u},m}
	\end{equation}
	where $v_{m,d}$ is a number obeying $v_{m,d}\leq (m+1)^{\frac{(d+2)(d-1)}{2}}$ where $d=\dim \cH$. Additionally, the number of different eigenvalues of $\sigma_{\text{u},m}$ is bounded by $v_{m,d}$.
\end{lemm}

In the following, we employ the pinching argument to relate $E_\text{sc}^\flat$ to $E_\text{sc}^*$.
\begin{lemm} \label{lem:flat-star-bound}
	Let $\cP_m(\rho_{XB}^{\otimes m})$ be the map pinching by $\one_{X^m}\otimes \sigma_{\text{u},m}$, and define $\rho_m = \cP_m(\rho_{XB}^{\otimes m})$.
	Then for $-1 < s< 0$,
	\begin{equation} \label{eq_aux-flat-star-bound}
	E_0^{\flat}(\rho_m, s) \leq   m E_0^{*}(s) - s3\log v_{m,|\cH_B|}.
	\end{equation}
	In particular,
	\begin{equation}\label{eq_sc-flat-star-bound}
	E_\textnormal{sc}^\flat(\rho_m,m R) \leq m E_\textnormal{sc}^*(R) + 3 \log v_{m,|\cH_B|}.
	\end{equation}
	Here, we write $E_0^{\flat}(\rho_m, s)$ and $E_\textnormal{sc}^\flat(\rho_m,m R)$ to emphasize the dependency on $\rho_m$. For $E_0^{*}(s)$ and $E_\textnormal{sc}^*(R)$, the underlying system is given by $\rho_{XB}$.
\end{lemm}

\begin{proof}[Proof of Lemma~\ref{lem:flat-star-bound}]
	Recall for $t=\{*\},\{\flat\}$ the definition
	\begin{equation} \label{eq:def-t-aux}
	E_0^{t}(s) = - s H^{t,\uparrow}_{\frac{1}{1+s}}(X|B)_\rho.
	\end{equation}
	Let us then consider, for $\alpha > 1$, the quantities
	\begin{equation} \label{eq:DefHflat}
	H^{\flat,\uparrow}_\alpha (X^m|B^m)_{\rho_m} = \max_{\sigma_{B^m}} - D^\flat_\alpha(\rho_m\|\one_{X^m}\otimes \sigma_{B^m})
	\end{equation}
	and
	\begin{equation} \label{eq:DefHstar}
	H^{*,\uparrow}_\alpha (X^m|B^m)_{\rho_m} = \max_{\sigma_{B^m}} - D^*_\alpha(\rho_m\|\one_{X^m}\otimes \sigma_{B^m}).
	\end{equation}

	The symmetric group on $m$ letters has a natural unitary representation on $\cH^{\otimes m}$ as follows (see e.g.~\cite{harrow_church_2013}). For $\pi\in \mathcal{S}_m$, we define
	\begin{equation}
	\pi_{\cH} := \sum_{i_1,\dotsc,i_m} \ket{i_{\pi^{-1}(1)} ,\dotsc, i_{\pi^{-1}(m)}}\bra{i_1,\dotsc,i_m}.
	\end{equation}
	which is a representation of $\mathcal{S}_m$ on $\cH^{\otimes m}$.
	Then define a mixed-unitary channel $\Sym_{B^m}$ given by
	\begin{equation}
	\Sym_{B^m}(\sigma_{B^m}) := \frac{1}{m!}\sum_{\pi \in \mathcal{S}_m} \pi_{B^m}^*(\sigma_{B^m})\pi_{B^m}.
	\end{equation}
	Since $D^{(t)}_\alpha$ is convex in its second argument for any $\alpha\in[\frac{1}{2},\infty)$ for both $t=\{*\}$ and $t=\{\flat\}$ by Prop. III.17 of \cite{MO17}, we have
	
	\begin{equation} \label{eq:sym-ieq}
	D^{(t)}_\alpha(\rho_m\|\one_{X^m}\otimes \Sym_{B^m}(\sigma_{B^m})) \leq \frac{1}{m!}\sum_{\pi\in \mathcal{S}_m}D^{(t)}_\alpha(\rho_m\|\one_{X^m}\otimes \pi_{B^m}^*\sigma_{B^m}\pi_{B^m}).
		\end{equation}
	But since $\rho_m$ is pinched by $\one_{X^m}\otimes\sigma_{\text{u},m}$, the state $\rho_m$ commutes with $\one_{X^m}\otimes\sigma_{\text{u},m}$ and therefore with $\one_{X^m} \otimes \pi_{B^m}$ for every $\pi\in \mathcal{S}_m$. Thus,
	\begin{align}
	D^{(t)}_\alpha(\rho_m\|\one_{X^m}\otimes \pi_{B^m}^*\sigma_{B^m}\pi_{B^m}) &= D^{(t)}_\alpha((\one_{X^m}\otimes \pi_{B^m}^*) \rho_m(\one_{X^m}\otimes \pi_{B^m})\|\one_{X^m}\otimes \pi_{B^m}^*\sigma_{B^m}\pi_{B^m}) \nonumber\\
	&=D^{(t)}_\alpha(\rho_m\|\one_{X^m}\otimes \sigma_{B^m})
	\end{align}
	since $D^{(t)}_\alpha$ is invariant under unitary conjugation. Thus, by \eqref{eq:sym-ieq}, we have
	\begin{equation}
	-D^{(t)}_\alpha(\rho_m\|\one_{X^m}\otimes \Sym_{B^m}(\sigma_{B^m})) \geq -D^{(t)}_\alpha(\rho_m\|\one_{X^m}\otimes \sigma_{B^m}).
	\end{equation}
	Since $\Sym_{B^m}(\sigma_{B^m})$ is symmetric, we have that $-D_\alpha^{(t)}$ is maximized on a symmetric state for $t=\{*\},\{\flat\}$, so we may restrict to symmetric states in the maximums in \eqref{eq:DefHflat} and \eqref{eq:DefHstar}.

	Moreover, again since $\rho_m$ is pinched by $\one_{X^m}\otimes\sigma_{\text{u},m}$ we have that $\rho_m$ commutes with $\one_{X^m}\otimes\sigma_{\text{u},m}$, and therefore with $\one_{X^m}\otimes \sigma_{B^m}$, for any symmetric state $\sigma_{B^m}$. By commutativity then, for any symmetric $\sigma_{B^m}$, we have
	\begin{equation}
	D^\flat_\alpha(\rho_m\|\one_{X^m}\otimes \sigma_{B^m}) = D^*_\alpha(\rho_m\|\one_{X^m}\otimes \sigma_{B^m}).
	\end{equation}
	Therefore,
	\begin{align}
	H^{\flat,\uparrow}_\alpha (X^m|B^m)_{\rho_m} & = \max_{\sigma_{B^m} \in S(B^m)} - D^\flat_\alpha(\rho_m\|\one_{X^m}\otimes \sigma_{B^m})\\
	&= \max_{\sigma_{B^m} \in S_\text{sym}(B^m)} - D^\flat_\alpha(\rho_m\|\one_{X^m}\otimes \sigma_{B^m}) \\
	&=\max_{\sigma_{B^m} \in S_\text{sym}(B^m)} - D^*_\alpha(\rho_m\|\one_{X^m}\otimes \sigma_{B^m}) \label{eq:Hflat_as_Dstar} \\
	&= \max_{\sigma_{B^m} \in S(B^m)} - D^*_\alpha(\rho_m\|\one_{X^m}\otimes \sigma_{B^m}) \\
	&= H^{*,\uparrow}_\alpha (X^m|B^m)_{\rho_m}.
	\end{align}
	However, we want a quantity in terms of $\rho_{XB}^{\otimes m}$ instead of $\rho_m$.
	For any symmetric $\sigma_{B^m}$, we have $\sigma_{B^m}\leq  v_{m,|\cH_B|} \sigma_{\text{u},m}$ so by Lemma III.23 of \cite{MO17},
	\begin{align}
	D^*_\alpha(\rho_m\|\one_{X^m}\otimes \sigma_{B^m}) &\geq  D^*_\alpha(\rho_m\|\one_{X^m}\otimes  v_{m,|\cH_B|} \sigma_{\text{u},m})\\
	&= D^*_\alpha(\rho_m\|\one_{X^m}\otimes \sigma_{\text{u},m}) - \log v_{m,|\cH_B|}
	\end{align}
	and then by Lemma 3 of \cite{HT14}, we have
	\begin{equation}
	D^*_\alpha(\rho_m\|\one_{X^m}\otimes \sigma_{\text{u},m}) \geq D^*_\alpha(\rho_{XB}^{\otimes m}\|\one_{X^m}\otimes \sigma_{\text{u},m}) - 2 \log v_{m,|\cH_B|}.
	\end{equation}
	Putting it together, by \eqref{eq:Hflat_as_Dstar},
	\begin{align}
	H^{\flat,\uparrow}_\alpha (X^m|B^m)_{\rho_m} &\leq - D^*_\alpha(\rho_{XB}^{\otimes m}\|\one_{X^m}\otimes \sigma_{\text{u},m}) + 3 \log v_{m,|\cH_B|} \\
	&\leq H^{*,\uparrow}_\alpha (X^m|B^m)_{\rho^{\otimes m}} + 3 \log v_{m,|\cH_B|}.
	\end{align}
	By writing $\alpha = \frac{1}{1+s}$ for $s\in(-1,0)$ and using the definition \eqref{eq:def-t-aux}, we recover \eqref{eq_aux-flat-star-bound}. Then, using $H^{*,\uparrow}_\alpha (X^m|B^m)_{\rho^{\otimes m}} = m H^{*,\uparrow}_\alpha (X|B)_{\rho}$, and the definitions
	\begin{equation}
	E_\text{sc}^\flat(\rho_m,mR) = \sup_{-1<s<0} [ E_0^{\flat,\text{SW}}(s,\rho_m) + smR]
	\end{equation}
	and
	\begin{equation}
	mE_\text{sc}^*(R) = \sup_{-1<s<0} [ mE_0^{* }(s ) + smR]
	\end{equation}
	we obtain \eqref{eq_sc-flat-star-bound}.
\end{proof}

\subsubsection{Proof of Proposition~\ref{prop:full-SC-achiev-bound}} \label{subsubsec:proof_SC_pinching}

Let $n,m\in \N$, $n\geq m>1$, and $\delta\in(0,\frac{1}{4})$. Let $k\in \N$ such that $km\leq n \leq (k+1)m$.  We will consider $\rho_m = \cP_m(\rho_{XB})$, the state $\rho_{XB}^{\otimes m}$ pinched by the universal symmetric state, and apply Proposition~\ref{prop:SC-achiev-flat-bound} to obtain a code for $\rho_m^{\otimes k}$ and bounds on its probability of success. From this, we will construct a code for $\rho_{XB}^{\otimes n}$.

We define
\begin{equation}
T^{(m)}(\delta,k) :=  \sqrt{4\log 2} (e_2^{(m)} + 1) + e_1^{(m)}(\delta)] + \frac{1}{\sqrt{k}}[1- \log( 1 - 2^{-kR})]
\end{equation}
where
\begin{gather}
e_1^{(m)}(\delta) :=   4\sqrt{2} \left(\log \left(1 + \tr[\rho_{m}^{-1/2} + \rho_{m}^{1/2}]\right)\right)\sqrt{-\log(1 - \sqrt{1-\delta^2})} \label{eq:def_e1m},\\
e_2^{(m)} := 4 \cosh (1+2\sqrt{|\mathcal{X}|^m}) [\log(1+2\sqrt{|\mathcal{X}|^m})]^2. \label{eq:def_e2m}
\end{gather}

	By Proposition~\ref{prop:SC-achiev-flat-bound} for any rate $R>0$, we can construct a code $\cC_k^{(m)}$ for $k$ copies of the state $\rho_m$ with rate $Rm$ such that
	\begin{equation}
	-\frac{1}{k} \log \Ps(\rho_m,\cC_k^{(m)})) \leq E_\text{sc}^\flat(\rho_m, mR) + \frac{1}{\sqrt{k}} T^{(m)}(\delta,k) .
	\end{equation}

	We may use Lemma~\ref{lem:flat-star-bound} to upper bound $ E_\text{sc}^\flat(\rho_m,mR)$ to obtain
	\begin{equation} \label{eq:UB-PsCkm}
	-\frac{1}{k} \log \Ps(\rho_m,\cC_k^{(m)})) \leq m E_\text{sc}( R) +  3 \log v_{m,|\cH_B|} + \frac{1}{\sqrt{k}} T^{(m)}(\delta,k) 
	\end{equation}
	
	From this code, we wish to construct a code $\sC_n$ for $n$ copies of $\rho_{XB}$ with rate $R$ that has the same probability of success as $\sC_k^{(m)}$ does for $\rho_m$.
	
	$\sC_k^{(m)}$ consists of an encoder $\cE_k^{(m)}: \cX^{km} \to \cW_{kmR}$, where $|\cW_{kmR}| = 2^{kmR}$,  and a decoder $\cD_k^{(m)}$ that associates to any element  $w\in \cW_{kmR}$ a POVM $\{\Pi_{\vec x}^{(w)}\}_{\vec x\in\cX^{km}}$ on $ \cH_B^{\otimes km}$ with outcomes in $\cX^{km}$.
	Note 
	\begin{equation}
	\rho_m = \sum_{\vec x \in \cX^{km}} p(\vec x)\ket{\vec x}\bra{\vec x} \otimes \cP_m(\rho_B^{\vec x})
	\end{equation}
	for $\rho_B^{\vec x} = \rho_B^{x_1}\otimes \dotsm \otimes \rho_B^{x_{km}}$ and $p(\vec x) =p(x_1)\dotsm p(x_{km})$ when $\vec x = (x_1,\dotsc,x_{km})$. Moreover,
	\begin{align}
	\Ps(\rho_m,\cC_k^{(m)}) &= \sum_{\vec x \in \cX^{km}} p(\vec x) \tr[ \Pi^{\cE_k^{(m)}(\vec x)}_{\vec x}\cP_m^{\otimes k}(\rho_B^{\vec x})].
	\end{align}
	We may replace the code $\sC_k^{(m)}$ by one with a decoder invariant under the pinching operation, 
	\begin{equation} \label{eq:invar-under-pinching}
	\Pi^{\cE_k^{(m)}(\vec x)}_{\vec x} = \cP_m^{\otimes k}( \Pi^{\cE_k^{(m)}(\vec x)}_{\vec x})
	\end{equation} since the success probability does not change:
	\begin{equation}
	\sum_{\vec x \in \cX^{km}} p(\vec x) \tr[ \Pi^{\cE_k^{(m)}(\vec x)}_{\vec x}\cP_m^{\otimes k}(\rho_B^{\vec x})] = \sum_{\vec x \in \cX^{km}} p(\vec x) \tr[ \cP_m^{\otimes k}(\Pi^{\cE_k^{(m)}(\vec x)}_{\vec x})\cP_m^{\otimes k}(\rho_B^{\vec x})]
	\end{equation}
	using that $\cP_m$ consists of a sum of conjugations by orthogonal projections. Moreover, 
	\begin{equation}
	\sum_{\vec x \in \cX^{km}} p(\vec x) \tr[ \Pi^{\cE_k^{(m)}(\vec x)}_{\vec x}\cP_m^{\otimes k}(\rho_B^{\vec x})] = \sum_{\vec x \in \cX^{km}} p(\vec x) \tr[ \cP_m^{\otimes k}(\Pi^{\cE_k^{(m)}(\vec x)}_{\vec x})\rho_B^{\vec x}]
	\end{equation}
	and assuming \eqref{eq:invar-under-pinching}, we find
	\begin{equation} \label{eq:prob-succ-rho-m-Ckm}
	\Ps(\rho_m,\cC_k^{(m)})  
		= \sum_{\vec x \in \cX^{km}} p(\vec x) \tr[ \Pi^{\cE_k^{(m)}(\vec x)}_{\vec x}\rho_B^{\vec x}].
	\end{equation}
	
	Let us define the code $\sC_n$ as follows. Let $\cW_{nR}\supset \cW_{kmR}$ have size $2^{nR}$, and define $\cE_n: \cX^n \to \cW_{nR}$ by
	\begin{equation}
	\cE_n: \quad (x_1,\dotsc,x_{km}, \dotsc, x_n) \mapsto \cE_k^{(m)} ( x_1,\dotsc, x_{km}).
	\end{equation}
	Note then $\sC_n$ has rate $R$.
	Next, define the decoder for $\vec w\in \cW_{kmR} \subset \cW_{nR}$ by
	\begin{equation} \label{eq:new-decoder-def}
	\tilde\Pi_{\vec x}^{(\vec w)} := \Pi_{(x_1,\dotsc,x_{km})}^{(\vec w)} \otimes \one_{\cH_B^{\otimes (n-km)}}
	\end{equation}
	where $ \Pi_{(x_1,\dotsc,x_{km})}^{\vec w}$ is the POVM element defined by the decoder of $\sC_k^{(m)}$. For $\vec w \in \cW_{nR}\setminus \cW_{kmR}$, let $\{\tilde \Pi_{\vec x}^{(\vec w)}\}_{\vec x\in \cX^n}$ be any POVM. Let us evaluate the success probability of this code. By definition
	\begin{equation}
	\Ps(\rho_{XB}, \sC_n) \equiv \sum_{\vec x\in \cX^n} p(\vec x) \tr[\tilde \Pi_{\vec x}^{\cE_n(\vec x)} \rho_B^{\vec x}]. 
	\end{equation}
	We may split the sum as
	\begin{equation}
	\Ps(\rho_{XB}, \sC_n)=\sum_{\vec y\in \cX^{km}} \sum_{\substack{\vec x = (x_1,\dotsc,x_n)\in \cX^n: \\(x_1,\dotsc,x_{km}) = \vec y}} p(\vec x) \tr[\tilde \Pi_{\vec x}^{\cE_n(\vec x)} \rho_B^{\vec x}]. 
	\end{equation}
Using $p(\vec x) = p(\vec y) p(x_{km+1}) \dotsm p(x_n)$ for $\vec x = (y_1,\dotsc, y_{km}, x_{km+1},\dotsc,x_n)$, and $\cE_n(\vec x) = \cE_k^{(m)}(\vec y)$, we have
\begin{equation}
\Ps(\rho_{XB}, \sC_n)= \sum_{\vec y\in \cX^{km}} p(\vec y)\sum_{\substack{\vec x = (x_1,\dotsc,x_n)\in \cX^n: \\(x_1,\dotsc,x_{km}) = \vec y}} p(x_{km+1})\dotsm p(x_n) \tr[\tilde \Pi_{\vec x}^{\cE_k^{(m)}(\vec y)} \rho_B^{\vec x}].
\end{equation}
By \eqref{eq:new-decoder-def}, and evaluating the trace over the last $n-km$ copies of $\cH_B$,
\begin{align}
\Ps(\rho_{XB}, \sC_n)&=\sum_{\vec y\in \cX^{km}} p(\vec y)\sum_{\substack{\vec x = (x_1,\dotsc,x_n)\in \cX^n: \\(x_1,\dotsc,x_{km}) = \vec y}} p(x_{km+1})\dotsm p(x_n) \tr[\tilde\Pi_{(x_1,\dotsc,x_{km})}^{\cE_k^{(m)}(\vec y)}\rho_B^{\vec y}] \\
&= \sum_{\vec y\in \cX^{km}} p(\vec y)  \tr[\tilde \Pi_{(x_1,\dotsc,x_{km})}^{\cE_k^{(m)}(\vec y)}\rho_B^{\vec y}] \sum_{\substack{\vec x = (x_1,\dotsc,x_n)\in \cX^n: \\(x_1,\dotsc,x_{km}) = \vec y}} p(x_{km+1})\dotsm p(x_n).
\end{align}
Then since $\sum_{\substack{\vec x = (x_1,\dotsc,x_n)\in \cX^n: \\(x_1,\dotsc,x_{km}) = \vec y}} p(x_{km+1})\dotsm p(x_n)=\sum_{x_1,\dotsc,x_{km}} p(x_{km+1})\dotsm p(x_n) = 1$, by \eqref{eq:prob-succ-rho-m-Ckm} we have $\Ps(\rho_{XB}, \sC_n)=\Ps(\rho_m, \sC_k^{(m)})$.
												Therefore, using $n\geq mk$ and \eqref{eq:UB-PsCkm}, we have
	\begin{equation}
		-\frac{1}{n} \log \Ps(\rho_{XB}, \cC_n) \leq  -\frac{1}{mk} \log \Ps(\rho_{XB}, \cC_n) \leq  E_\text{sc}^*(  R) +  \frac{3}{m} \log v_{m,|\cH_B|} +  \frac{1}{m\sqrt{k}} T^{(m)}(\delta,k).
	\end{equation}
	By using the upper bound on $v_{m,|\cH_B|}$ given in Lemma~\ref{lem:univ-symm-state} and that $1-\Pestar(n,R) \geq \Ps(\rho_{XB},\sC_n)$, we find
	\begin{equation}
	-\frac{1}{n} \log (1-\Pestar(n,R)) \leq  E_\text{sc}^*(  R) +  \frac{3(|\cH_B|+2)(|\cH_B|-1)}{2}\frac{1}{m} \log (m+1) + \frac{1}{m\sqrt{k}} T^{(m)}(\delta,k).
	\end{equation}
	Finally, we take $k = \floor{\frac{n}{m}}$ and make use of the inequality $k \geq \frac{n}{m}-1$ to obtain the result. \qed

\section{Moderate Deviation Regime \label{sec:mod_dev_reg}}
Recall that a sequence of real numbers $(a_n)_{n\in\mathbb{N}}$ satisfying  \eqref{eq:a_n} is called a moderate sequence.

\subsection{Optimal error when the rate approaches the Slepian-Wolf limit moderately quickly}\label{sec:mod_dev_for_error}

Let us recall Theorem~\ref{theo:mod_large}:

\modlarge*

\begin{proof}[Proof of Theorem~\ref{theo:mod_large}]
	We shorthand $H = H(X|B)_\rho$, $V = V(X|B)_\rho$  for notational convenience.	
	We first show the achievability, i.e.~the ``$ \geq$" in Eq.~\eqref{eq:mod_large0}.
	Let $\{a_n\}_{n\geq 1}$ be any sequence of real numbers satisfying Eq.~\eqref{eq:a_n}. 
 For every $n\in\mathbb{N}$, Theorem~\ref{theo:large_ach} implies that there exists a sequence of $n$-block codes with rates $R_n = H + a_n$ so that
	\begin{align}
	\Pestar(n,R_n) \leq 4 \exp \left\{ -n \left[ \max_{0\leq s \leq 1} \left\{
	{E}_0^{\downarrow}(s ) + s R_n
	\right\} \right] \right\}. \label{eq:sup5}
	\end{align}
	Applying Taylor's theorem to ${E}^{\downarrow}_0(s)$ at $s=0$ together with Proposition \ref{prop:E0_SW+E0down} gives
	\begin{align} \label{eq:sup2}
	 {E}_0^\downarrow\left(s \right) = - sH - \frac{s^2}{2} V + 
	\left. \frac{s^3}{6} \frac{\partial^3 {E}^{\downarrow}_0\left(s\right)}{\partial s^3} \right|_{s = \bar{s}},
	\end{align}
	for some $\bar{s} \in [0,s]$. 
	Now, let $s_n = a_n/V$. Then,  $s_n\leq 1$ for all sufficiently large $n$ by the assumption in Eq.~\eqref{eq:a_n} and $V >0$. For all $s_n\leq 1$,
	\begin{equation} 
	\max_{0\leq s \leq 1} \left\{ {E}^{\downarrow}_0\left(s\right) + sR_n
	\right\} \geq  {E}_0^{\downarrow}\left(s_n \right) + s_n R_n .
	\end{equation}
 Then eq.~\eqref{eq:sup2} yields 
	\begin{equation}
	\max_{0\leq s \leq 1}\left\{ {E}^{\downarrow}_0\left(s\right) + sR_n
	\right\}  = \frac{a_n}{V } \left( -H + R_n \right)
	-\frac{a_n^2}{2 V }  + \frac{a_n^3}{6 V^3}
	\frac{\partial^3  {E}^{\downarrow}_0\left(s\right)}{\partial s^3} 	\Big|_{s = \bar{s}_n} 
	\end{equation}
where $\bar{s}_n\in[0,s_n]$. Then, since {$R_n = H + a_n$},
	\begin{align}
	\max_{0\leq s \leq 1}\left\{ {E}^{\downarrow}_0\left(s\right) + sR_n
	\right\}  &= \frac{a_n^2}{2 V }  + \frac{a_n^3}{6 V^3}
	 \frac{\partial^3 {E}^{\downarrow}_0\left(s\right)}{\partial s^3} \Big|_{s = \bar{s}_n} \label{eq:sup3} \\
	&\geq \frac{a_n^2}{2 V }  - \frac{a_n^3}{6 V ^3}
	\left| \frac{\partial^3 {E}^{\downarrow}_0\left(s \right) }{\partial s^3} 	\Big|_{s = \bar{s}_n} \right| \\	
	&\geq 
	\frac{a_n^2}{2 V }  - \frac{a_n^3}{6 V ^3}\Upsilon, \label{eq:sup4}
	\end{align}
	where
	\begin{equation} \label{eq:Upsilon}
	\Upsilon = \max_{ s \in[0,1]  } 
	\left| \frac{\partial^3  {E}^{\downarrow}_0\left(s\right)}{\partial s^3} \right|.
	\end{equation}
	This quantity is finite due to the compactness of $[0,1]$ and the continuity, item \ref{E0_SW-a} in Proposition \ref{prop:E0_SW+E0down}.
	Therefore, substituting Eq.~\eqref{eq:sup4} into Eq.~\eqref{eq:sup5} gives, for all sufficiently large $n\in\mathbb{N}$,
	\begin{align}
	\frac{1}{na_n^2} \log \left( \frac{1}{\Pestar(n, R_n ) } \right)  \geq - \frac{\log 4}{ na_n^2}  + \frac{1}{2V } \left( 1 - \Upsilon \frac{a_n}{3V ^2} \right).
	\end{align}
	Recall Eq.~\eqref{eq:a_n} and let $n\to+\infty$, which completes the lower bound:
	\begin{align}
	\liminf_{n\to +\infty} \frac{1}{na_n^2} \log   \left( \frac{1}{\Pestar(n, R_n ) } \right)  \geq \frac{1}{2V }. \label{eq:converse_HT4}
	\end{align}
	
	We move on to show the converse, i.e.~the ``$ \leq$" in Eq.~\eqref{eq:mod_large0}.
	Let $N_1\in\mathbb{N}$ be an integer such that $R_n = H + a_n \in (H_1(X|B)_\rho, H_0(X|B)_\rho) $ for all $n\geq N_1$. 
	We denote by $(\alpha_{R_n}^\star, \sigma_{R_n}^\star )$ the unique saddle-point of 
	\begin{align}
	\sup_{\alpha\in (0,1]} \inf_{\sigma_B \in \mathcal{S}(B)}
	\frac{1-\alpha}{\alpha}\left( R_n + D_\alpha( \rho_{XB}\|\mathds{1}_X\otimes \sigma_B
	\right).
	\end{align}
	By invoking the one-shot converse bound, Proposition~\ref{prop:one-shot_converse}, with $M = \exp\{n R_n\}$, we  obtain for all $n\geq N_1$,
\begin{align}
\log \left( \frac{1}{\Pestar(n, R_n ) } \right) 
&\leq \min_{\sigma_B^n \in \mathcal{S}(B^n) } - \log \widehat{\alpha}_{ \frac{M}{|\mathcal{X}^n|} } \left( \rho_{X^n B^n} \| \tau_{X^n} \otimes \sigma_B^n \right) \\
&\leq - \log \widehat{\alpha}_{ \frac{M}{|\mathcal{X}^n|} } \left( \rho_{X^n B^n} \| \tau_{X^n} \otimes (\sigma_{R_n}^\star)^{\otimes n} \right) \\
&= - \log \widehat{\alpha}_{ \frac{M}{|\mathcal{X}|^n} } \left( \rho_{X B}^{\otimes n} \| (\tau_X \otimes \sigma_{R_n}^\star)^{\otimes n} \right). \label{eq:converse_HT3}
\end{align}

Next, we verify that we are able to employ Eq.~\eqref{eq:sharp1} in Proposition~\ref{prop:sharp_Hoeffding} to asymptotically expand Eq.~\eqref{eq:converse_HT3}.
Equation~\eqref{eq:spChc3} in Proposition~\ref{prop:spCh} below shows that $\lim_{n\to+\infty} \alpha_{R_n}^\star = 1$. This together with the closed-form expression of $\sigma_{R_n}^\star$ \cite{SW12}, \cite[Lemma 1]{TBH14}, \cite[Lemma 5.1]{Tom16} shows that
\begin{align}
\lim_{n\to+\infty} \sigma_{R_n}^\star = \lim_{n\to+\infty} \frac{ \left( \Tr_X\left[\rho_{XB}^{\alpha_{R_n}^\star}\right] \right)^{\frac{1}{\alpha_{R_n}^\star}} }{ \Tr\left[\left( \Tr_X\left[\rho_{XB}^{\alpha_{R_n}^\star}\right] \right)^{\frac{1}{\alpha_{R_n}^\star}}\right] } = \rho_B.
\end{align}
Since $V = V(\rho_{XB}\|\mathds{1}_X\otimes \rho_B) > 0$, by the continuity of $V(\cdot\|\cdot)$ (c.f.~\eqref{eq:V}), for every $\kappa\in(0,1)$ there exists $N_2\in\mathbb{N}$ such that for all $n\geq N_2$,
\begin{align}
V\left( \rho_{XB}\| \tau_X \otimes \sigma_{R_n}^\star    \right)
= V\left( \rho_{XB}\| \mathds{1}_X \otimes \sigma_{R_n}^\star    \right) \geq (1-\kappa)V = :\nu > 0.
\end{align}
Hence, we apply Eq.\eqref{eq:sharp1} in  Proposition~\ref{prop:sharp_Hoeffding} with $r_n = \log|\mathcal{X}| - R_n$, $\rho = \rho_{XB}$ and $\sigma = \tau_X\otimes \sigma_{R_n}^\star$ to obtain for all $n\geq \max\{N_1, N_2\}$, 
	\begin{align}
	-\log \widehat{\alpha}_{\exp\left\{-n r_n\right\}} \left( \rho^{\otimes n} \| \sigma^{\otimes n} \right) 
	&\leq n \sup_{\alpha\in(0,1]} \frac{1-\alpha}{\alpha} \left( D_\alpha\left(\rho\|\sigma\right) - r_n + \gamma_n \right) + 
	\log \left( s_n^\star\sqrt{n} \right) + K, \\
	&= n E_\text{sp} ( H + a_n + \gamma_n  ) + 
\log \left( s_n^\star\sqrt{n} \right) + K, 	
	\label{eq:converse_HT1}
	\end{align}
	for some constant $K>0$, and $s_n^\star := ( 1-\alpha_{R_n}^\star )/ \alpha_{R_n}^\star$.
	Now, let $\delta_n := a_n + \gamma_n$, and notice that $\gamma_n = O(\frac{\log n}{n} ) = o(a_n)$.
	We invoke Proposition \ref{prop:spCh} below  to have
	\begin{align}
	\limsup_{n\to+\infty} \frac{ E_\text{sp} ( H(X|B)_\rho + \delta_n ) }{a_n^2} =
	\limsup_{n\to+\infty} \frac{ E_\text{sp} ( H(X|B)_\rho + \delta_n ) }{\delta_n^2} \leq \frac{1}{2V}. \label{eq:converse_HT6}
	\end{align}
	Moreover, Eq.~\eqref{eq:spChc3} in Proposition \ref{prop:spCh} gives that $\lim_{n\to+\infty}\frac{s_n^\star}{\delta_n} = 1/V$.
	Combining Eqs.~\eqref{eq:a_n}, \eqref{eq:converse_HT3}, \eqref{eq:converse_HT1} and \eqref{eq:converse_HT6} to conclude our claim
	\begin{align}
	\limsup_{n\to+\infty} \frac{1}{n a_n^2}\log \left( \frac{1}{\Pestar(n, R_n ) } \right)  &\leq 
	\limsup_{n\to+\infty} - \frac{\log \widehat{\alpha}_{\exp\left\{-n r_n\right\}} \left( \rho^{\otimes n} \| \sigma^{\otimes n} \right) }{ n a_n^2 }  \\
	&\leq \frac{1}{2V} + 
	\limsup_{n\to+\infty}  \frac{\log \left( s_n^\star\sqrt{n} \right)}{n a_n^2} \\
	&= \frac{1}{2V} + 
	\limsup_{n\to+\infty}  \frac{\log \left( s_n^\star\sqrt{n} \right)}{n \delta_n^2} \\
	&= \frac{1}{2V} + \limsup_{n\to+\infty}  \frac{ \frac12 \log \left( {n}\delta_n^2  \right) - \log V}{ n \delta_n^2 }  \\
	&= \frac{1}{2V}, \label{eq:converse_HT2}
	\end{align}	
	where the last line follows from $\lim_{n\to+\infty} n \delta_n^2 = +\infty$.
	Hence, Eq~\eqref{eq:converse_HT4} together with Eq.~\eqref{eq:converse_HT2}  completes the proof.

\begin{prop}[Error Exponent around Conditional Entropy] \label{prop:spCh}
	Let  $(\delta_n)_{n\in\mathbb{N}}$ be a sequence of positive numbers with $\lim_{n\to+\infty} \delta_n = 0$.
					The following hold:
	\begin{align}
	\limsup_{n\to+\infty} \frac{ E_\textnormal{sp} \left( H(X|B)_\rho + \delta_n \right) }{ \delta_n^2 } &\leq   \frac{ 1 }{2 V(X|B)_\rho }; \label{eq:spChc1} \\
	\limsup_{n\to+\infty} \frac{s_n^\star}{\delta_n} &= \frac{1}{V(X|B)_\rho}, \label{eq:spChc3}
	\end{align}
	where 
	\begin{align}
	s_n^\star := \argmax_{s\geq 0} \left\{ s\left( H(X|B)_\rho + \delta_n \right) - s H_{\frac{1}{1+s}}^\uparrow(X|B)_\rho
	\right\}.
	\end{align}
\end{prop}
\noindent The proof of Proposition~\ref{prop:spCh} is provided in Appendix \ref{app:spCh}.
\end{proof}

\subsection{Optimal rate when the error approaches zero moderately quickly} \label{sec:mod_rate}

Theorem 9 of \cite{TH13} provides bounds on $R^*(1,\eps)$. By applying these bounds to $\rho^{\otimes n}_{XB}$ and slightly reformulating them, we find that for any $n\in \N$ and any  $\alpha \in (0,1)$, we have
\begin{equation} \label{n-shot-bounds-rate}
-\frac{1}{n}D_\text{H}^\eps(\rho_{XB}^{\otimes n}\|(\one_X\otimes \rho_B)^{\otimes n})  \leq\frac{1}{n}  R^*(n, \eps) \leq - \frac{1}{n}D_\text{H}^{\alpha \eps}(\rho_{XB}^{\otimes n}\|(\one_X\otimes \rho_B)^{\otimes n}) +\frac{1}{n} \log \frac{8}{(1-\alpha)^2 \eps} .
\end{equation}
By combining this result with a moderate-deviations expansion of the hypothesis testing relative entropy developed in \cite{CTT2017}, we obtain a moderate deviations result for $R^*(n,\eps)$.

\modrate*

This result relies heavily on the following expansion of the hypothesis testing relative entropy.
\begin{prop}[Theorem 1 of \cite{CTT2017}] \label{prop:moderate-expansion-entropy}
	For any moderate sequence $a_n$ and $\eps_n := \e^{-na_n^2 }$, and quantum states $\rho$ and $\sigma$ with $\rho \ll \sigma$, we have
	\begin{equation} \label{eq:moderate-expansion-entropy}
	\frac{1}{n}D^{\eps_n}_\textnormal{H} (\rho^{\otimes n}\|\sigma^{\otimes n}) = D(\rho\|\sigma) - \sqrt{2 V(\rho\|\sigma)} a_n + o(a_n).
	\end{equation}
\end{prop}

\begin{proof}[Proof of Proposition~\ref{theo:moderate_rate}]
We may extend Proposition~\ref{prop:moderate-expansion-entropy} to unnormalized $\sigma\geq 0$ simply by factoring out the trace of $\sigma$ from the second slot of the hypothesis testing relative entropy using that 
\begin{equation}
D_\text{H}^\eps(\rho\|\lambda\sigma) = D_\text{H}^\eps(\rho\|\sigma) - \log \lambda, \qquad D(\rho\|\lambda\sigma) = D(\rho\|\sigma) - \log \lambda,
\end{equation}
and
\begin{equation}
V(\rho\|\lambda \sigma) = V(\rho\|\sigma)
\end{equation}
for $\lambda > 0$.
Therefore,
\begin{align*}	
\frac{1}{n}D^{\eps_n}_\text{H} (\rho^{\otimes n}\| (\lambda\sigma)^{\otimes n}) &=
\frac{1}{n}D^{\eps_n}_\text{H} (\rho^{\otimes n}\|\sigma^{\otimes n}) -\log \lambda= D(\rho\|\sigma) - \sqrt{2 V(\rho\|\sigma)} a_n + o(a_n) - \log \lambda.\\
&=D(\rho\|\lambda\sigma) - \sqrt{2 V(\rho\|\lambda\sigma)} a_n + o(a_n)
\end{align*}
and thus the relation \eqref{eq:moderate-expansion-entropy} holds for unnormalized $\sigma\geq 0$.

Next, we consider \eqref{n-shot-bounds-rate} with $\eps = \eps_n := \e^{-na_n^2}$, yielding
\begin{equation} \label{n-shot-bounds-rate-moderate}
-\frac{1}{n}D_\text{H}^{\eps_n}(\rho_{XB}^{\otimes n}\|(\one_X\otimes \rho_B)^{\otimes n})  \leq\frac{1}{n} \log R^* (n, \eps) \leq - \frac{1}{n}D_\text{H}^{\alpha \eps_n}(\rho_{XB}^{\otimes n}\|(\one_X\otimes \rho_B)^{\otimes n}) +\frac{1}{n} \log \frac{8}{(1-\alpha)^2 \eps_n} .
\end{equation}
 We next need to apply Proposition~\ref{prop:moderate-expansion-entropy}  to the hypothesis testing relative entropy on each side. While this application on the left-hand side is immediate, for the right-hand side, we need to check that $(b_n)$ satisfying 
 \begin{equation}
  \alpha \eps_n  = \e^{- nb_n^2}
  \end{equation} is a moderate sequence. We define $b_n := \sqrt{ a_n^2 +\frac{1}{n} \log \frac{1}{\alpha}}$.
Since for any $x,y>0$ we have
\begin{equation}
\sqrt{x+y} \leq \sqrt{x + y  + 2\sqrt{xy} } = \sqrt{(\sqrt{x} + \sqrt{y})^2} = \sqrt{x} + \sqrt{y}
\label{eq:sqrt-bound}
\end{equation}
we therefore obtain
\begin{equation} \label{eq:b_n-a_n_ieq}
0 \leq b_n \leq a_n + \sqrt{ \frac{1}{n}\log \frac{1}{\alpha}}  \xrightarrow{n\to \infty} 0 
\end{equation}
 by taking $x = a_n^2$ and $y= \frac{1}{n}\log \frac{1}{\alpha}$ in \eqref{eq:sqrt-bound}. Since $nb_n^2 = n a_n^2 + \log \frac{1}{\alpha} \xrightarrow{n\to\infty} \infty$, the sequence $(b_n)_n$ is indeed moderate. Thus, Proposition~\ref{prop:moderate-expansion-entropy} yields
\begin{equation} \label{eq:moderate-UB-expansion-1}
\frac{1}{n}D_\text{H}^{\alpha \eps_n}(\rho_{XB}^{\otimes n}\|(\one_X\otimes \rho_B)^{\otimes n}) = D(\rho_{XB}\|\one_X\otimes \rho_B) - \sqrt{2V(\rho_{XB}\|\one_X\otimes \rho_B)}  b_n + o(b_n).
\end{equation}
Now, we have that
\begin{equation}
b_n - a_n \leq \frac{1}{\sqrt{n}} \sqrt{\log \frac{1}{\alpha}} = o(a_n)
\end{equation}
since  for any $\delta > 0$,
\begin{equation}
\frac{1}{\sqrt{n}} \sqrt{\log \frac{1}{\alpha}} \leq \delta a_n \iff \sqrt{\log \frac{1}{\alpha}}\leq \sqrt{n}a_n
\end{equation}
which occurs for all $n$ sufficiently large because $\sqrt{n}a_n\to\infty$.

Moreover, if $f_n = o(b_n)$, since $b_n-a_n = o(a_n)$, we have $f_n = o(a_n)$.
Therefore, the error terms $f_n$ hidden in the $o(b_n)$ of \eqref{eq:moderate-UB-expansion-1} are in fact $o(a_n)$. Moreover, we may write
\begin{equation}
\sqrt{2V} b_n = \sqrt{2V}  a_n + \sqrt{2V}(b_n-a_n) = \sqrt{2V}  a_n + o(a_n)
\end{equation}
with $V:= V(\rho_{XB}\|\one_X\otimes\rho_B)$. Thus, \eqref{eq:moderate-UB-expansion-1} yields
\begin{equation}
\frac{1}{n}D_\text{H}^{\alpha \eps_n}(\rho_{XB}^{\otimes n}\|(\one_X\otimes \rho_B)^{\otimes n}) = D(\rho_{XB}\|\one_X\otimes \rho_B) - \sqrt{2V}  a_n + o(a_n)
\end{equation}
for any $\alpha\in(0,1)$.

The second term of the right-hand side of \eqref{n-shot-bounds-rate-moderate} is
\begin{equation}
\frac{1}{n}\log \frac{8}{(1-\alpha)^2 \eps_n} = \frac{1}{n} \log \frac{8}{(1-\alpha)^2} - \frac{1}{n}\log \eps = \frac{1}{n} \log \frac{8}{(1-\alpha)^2} + a_n^2.
\end{equation}
Since both $\frac{1}{n}=o(a_n^2)$ and $a_n^2 = o(a_n)$, the second term on the right-hand side of \eqref{n-shot-bounds-rate-moderate} is $o(a_n)$. 
Thus, we may conclude
\begin{equation}
\frac1n\log R^*(n, \eps_n) \leq- D(\rho_{XB}\|\one_X\otimes \rho_B) + \sqrt{2V(\rho_{XB}\|\one_X\otimes \rho_B)}  a_n + o(a_n).
\end{equation}
This precisely matches the bound obtained by applying Proposition~\ref{prop:moderate-expansion-entropy} to the left-hand side of \eqref{n-shot-bounds-rate-moderate}, and therefore we obtain \eqref{eq:moderate_exp-rate}.
\end{proof}

\section{Discussion} \label{sec:conclusions}

In this paper, we study the CQSW protocol, which is the task of classical data compression with quantum side information associated to a c-q state $\rho_{XB}$, for which the asymptotic data compression limit was shown to be $H(X|B)_\rho$ \cite{DW03}. We focus primarily on the non-asymptotic (i.e.~finite $n$) scenario, and obtain results for both the large and moderate deviation regimes. In the large deviation regime, the compression rate $R$ is fixed. We derive lower and upper bounds on the error exponent function for the range $R > H(X|B)_\rho$ (Theorems \ref{theo:large_ach} and \ref{theo:sp_SW}),  and corresponding bounds for the strong converse exponent for the range $R<H(X|B)_\rho$ (Theorems \ref{thm:SC-converse-bound} and \ref{thm:SC-achiev-bound}). Comparing the finite blocklength lower bound on the strong converse exponent given in Theorem~\ref{thm:SC-converse-bound} with the bounds given given by Equation~(6.19) in Theorem 6.2 of \cite{leditzky_strong_2016} and Equation~(8.6) of Section 8.1.3 of \cite{tom-thesis} remains an open question.

In addition, we analyze two scenarios in the moderate deviation regime. In the first, the rate depends on $n$ and slowly decays to the $H(X|B)_\rho$ from above, and we characterize the speed of convergence of the optimal error probability to zero in terms of the conditional entropy variance (Theorem \ref{theo:mod_large}). In the second, we obtain an expansion for the minimal rate possible to accomplish the CQSW protocol when the error probability is less than a threshold value which decays slowly with $n$ (Theorem \ref{theo:moderate_rate}).

There seems to be an interesting duality between our results on the error exponents for the CQSW protocol and those for classical-quantum channel coding \cite{Hol00, Dal13, CH16, CHT17}. In the former, the entropic error exponent functions arising in our bounds involve the difference between the compression rate and a conditional R\'enyi entropy, while in the latter they involve the difference between the R\'enyi capacity and the transmission rate. The above duality mimics the connections found between the tasks of Slepian Wolf coding and classical channel coding \cite{Ahl80, AD82, CHJ+09}.
We summarize such connections in Table~\ref{table:exponent} below.

Besides investigating this duality in detail \cite{duality}, other open problems include extending variable-rate Slepian-Wolf coding \cite{Gal76, CK80, CK81, WM15, chen_reliability_2017}, and the universal coding scenario \cite{Gal76, OH94, Csi82, Hay09} to the CQSW setting.

\begin{table}[th!]
	\centering
	\resizebox{1\columnwidth}{!}{
		\begin{tabular}{|c|c|c|} 			\toprule
			
			Bounds\textbackslash Settings & Slepian-Wolf Coding with Quantum Side Information & Classical-Quantum Channel Coding \\	
			
			\midrule
			\midrule

						Achievability & \multirow{2}{*}{$ \displaystyle  {E}_\text{r}^\downarrow (R)
				:= \sup_{ \frac12 \leq \alpha\leq 1} \frac{1-\alpha}{\alpha} \left( R - H_{2-\frac{1}{\alpha}}^\downarrow(X|B)_\rho \right)  $}
			& \multirow{2}{*}{$\displaystyle {E}_\text{r}^\downarrow(R) := \sup_{ \frac12 \leq \alpha\leq 1} \frac{1-\alpha}{\alpha} \left( \max_{P\in\mathscr{P}(\mathcal{X})} I_{2-\frac{1}{\alpha}}^\downarrow(p,\mathscr{W}) -   R \right)
				$} \\
			($R<C_\mathscr{W}$ or $R>H(X|B)_{\rho}$ )  & &
			\\			
			\midrule
			
			Optimality & $ \displaystyle {E}_\text{sp} (R)
			:= \sup_{s\geq 0}\left\{ E_0 (s ) + sR\right\}$
			& $\displaystyle {E}_\text{sp}(R)
			:= \sup_{s\geq 0}\left\{ \max_{P\in\mathscr{P}(\mathcal{X})}E_0(s,P) - sR\right\}$ \\
			($R<C_\mathscr{W}$ or $R>H(X|B)_{\rho}$ )  & 
			$\displaystyle \qquad\qquad\qquad\quad\; = \sup_{0\leq \alpha\leq1 } \left\{\frac{1-\alpha}{\alpha}\left( R - H_\alpha^{\uparrow}(X|Y)_\rho\right)  \right\} $	
			& $\displaystyle \qquad\quad = \sup_{0\leq \alpha\leq1 } \left\{\frac{1-\alpha}{\alpha}\left( C_{\alpha,\mathscr{W}} - R\right)  \right\} $\\			
			
			\midrule
			
			Strong Converse 
			& $ \displaystyle {E}_\text{sc}^*(R)
			:= \sup_{-1< s< 0}\left\{ {E}_0^\textsf{* }(s ) + sR\right\}$ & $\displaystyle {E}_\text{sc}^*(R)
			:= \sup_{-1< s<0}\left\{ \max_{P\in\mathscr{P}(\mathcal{X})} {E}_0^*(s,P) - sR\right\}$ \\
			($R>C_\mathscr{W}$ or $R<H(X|B)_{\rho}$ )  
			& $\displaystyle \qquad\qquad\;\; = \sup_{ \alpha>1 } \left\{\frac{1-\alpha}{\alpha}\left( R - H_\alpha^{*,\uparrow}(X|Y)_\rho\right)  \right\} $	
			& $\displaystyle \;\; = \sup_{ \alpha>1 } \left\{\frac{1-\alpha}{\alpha}\left( C_{\alpha,\mathscr{W}}^* - R\right)  \right\} $ \\		
			
			\midrule
			
			Auxiliary Function & $\displaystyle E_0 (s ) := -\log \Tr_B \left[
			\left( \Tr_X (\rho_{XB})^{1/(1+s)}\right)^{1+s}
			\right]$  & $\displaystyle E_0(s,P) := -\log \Tr \left[
			\left( \sum_{x\in\mathcal{X}} P(x) \cdot W_x^{1/(1+s)}\right)^{1+s}
			\right]$\\
			
			\bottomrule				
		\end{tabular}}
		\caption{The comparison of the error exponent analysis for Slepian-Wolf coding with quantum side information and classical-quantum channel coding. 
			The classical-quantum channel is denoted by $\mathscr{W}:\mathcal{X}\to \mathcal{S}(B)$, i.e. $x\mapsto W_x \in \mathcal{S}(B)$ for some Hilbert space $\mathcal{H}$.	
			The achievability of classical-quantum channel coding was proved in \cite{Hay07}, and the quantity $I_{\alpha}^\downarrow(p,\mathscr{W})$ is defined by  $I_{\alpha}^\downarrow(p,\mathscr{W}):= D_{\alpha}\left( \rho_{XB} \| \rho_X \otimes \rho_B \right)$ for $\rho_{XB} := \sum_x p(x) |x\rangle\langle x|\otimes W_x$.
		} \label{table:exponent}
	\end{table}

\section*{Acknowledgements}
M.-H.~Hsieh was supported by an ARC Future Fellowship under Grant FT140100574 and by US Army Research Office for Basic Scientific Research Grant W911NF-17-1-0401.
H.-C.~Cheng was supported by Ministry of Science and Technology Overseas Project for Post Graduate Research (Taiwan) under Grant 105-2917-I-002-028 and 104-2221-E-002-072. E.~P.~Hanson was partly supported by the Cantab Capital Institute for the Mathematics of Information (CCIMI) and would like to thank Institut Henri Poincar\'e, where part of this research was carried out, for its support and hospitality.

\appendix

\section{Proof of Proposition~\ref{prop:H}} \label{app:H}
\begin{prop2}[Properties of $\alpha$-R\'enyi Conditional Entropy] 	Given any classical-quantum state $\rho_{XB} \in \mathcal{S}(XB)$, the following holds:
	\begin{enumerate}[(a)]
		\item\label{H-aa} The map $\alpha \mapsto H_\alpha^\uparrow (X|B)_\rho $ is continuous and monotonically decreasing on $[0,1]$.
		
		\item\label{H-bb} The map $\alpha \mapsto \frac{1-\alpha}{\alpha} H_\alpha^\uparrow(X|B)_\rho$ is strictly concave on $(0,1)$.
	\end{enumerate}
\end{prop2}

\begin{proof}[Proof of Proposition \ref{prop:H}]
	~\\
	\begin{itemize}
		\item[(\ref{prop:H})-\ref{H-aa}]		
		Fix an arbitrary sequence $(\alpha_k)_{k\in\mathbb{N}}$ such that $\alpha_k\in[0,1]$ and $\lim_{k\to+\infty} \alpha_k = \alpha_\infty \in [0,1]$.
		Let 
		\begin{align}
		\sigma_k^\star \in \argmin_{\sigma\in\mathcal{S(H)}} D_{\alpha_k} \left(\rho_{XB}\|\mathds{1}_X\otimes \sigma \right), \quad \forall k\in\mathbb{N}\cup\{+\infty\}.
		\end{align}
		The definition in Eq.~\eqref{eq:cond_ent} implies that
		\begin{align}
		\limsup_{k \to +\infty}
		H_{\alpha_k}^\uparrow (X|B)_\rho &= - \liminf_{k \to +\infty} D_{\alpha_k} \left(\rho_{XB}\| \mathds{1}_X \otimes \sigma_k^\star \right) \\
		&\leq -   D_{\alpha_\infty} \left( \rho_{XB} \left\| \mathds{1}_X\otimes \left(\lim_{k \to +\infty} \sigma_k^\star\right)\right.  \right) \label{eq:II1}\\
		&\leq - \min_{\sigma\in\mathcal{S(H)}} D_{\alpha_\infty} \left( \rho_{XB} \| \mathds{1}_X\otimes \sigma \right) \\
		&= H_{\alpha_\infty}^\uparrow (X|B)_\rho, \label{eq:II2}
		\end{align}
		where, in order to establish \eqref{eq:II1}, we used the lower semi-continuity of the map $\sigma\mapsto D_{\alpha_k}(\rho_{XB} \| \mathds{1}_X\otimes \sigma)$ (see Eq.~\eqref{eq:chaotic4} in Lemma \ref{lemma:chaotic})  and the continuity of $\alpha \mapsto D_\alpha\left( \rho_{XB} \| \mathds{1}_X\otimes \sigma_k^\star \right)$ (Eq.~\eqref{eq:chaotic5} in Lemma~\ref{lemma:chaotic}).
		
		Next, we let 
		\begin{align}
		\sigma_k &:= 	\left( 1 - \eps_k \right)\sigma_\infty^\star +  \eps_k \frac{\mathds{1}}{d},\quad \forall k\in\mathbb{N},				\end{align}
		where $(\eps_k)_{k\in\mathbb{N}}$ is an arbitrary positive sequence that converges to zero.
		Then, it follows that
		\begin{align}
		\liminf_{k \to +\infty}
		H_{\alpha_k}^\uparrow(X|B)_\rho
		&\geq - \limsup_{k \to +\infty} \left\{ D_{\alpha_k}\left( \rho_{XB} \| \mathds{1}_X\otimes \sigma_k  \right) \right\} \\
		&= - D_{\alpha_\infty} \left( \rho_{XB} \| \mathds{1}_X\otimes \sigma_\infty^\star  \right) \label{eq:II4} \\
		&= H_{\alpha_\infty}^\uparrow (X|B)_\rho . \label{eq:II5}
		\end{align}
		Here, equality~\eqref{eq:II4} holds because $ \mathds{1}_X\otimes \sigma_k \gg \rho_{XB}$ for all $k\in\mathbb{N}\cup \{+\infty\}$. Thus, the map $(\alpha_k, \sigma_k) \mapsto D_{\alpha_k}(\rho_{XB}\| \mathds{1}_X\otimes \sigma_k)$ is continuous for $k\in\mathbb{N}\cup\{+\infty\}$.
		Hence, we prove the continuity.
		
		Now, we show the monotonicity.
		For all $\sigma_B \in \mathcal{S}(B)$, Eq.~\eqref{eq:chaotic4} in Lemma~\ref{lemma:chaotic} implies that $-D_\alpha(\rho_{XB} \| \mathds{1}\otimes \sigma_B)$ is monotonically decreasing in $\alpha\geq 0$.
		Since $H_\alpha^\uparrow(X|B)_\rho$ is the pointwise supremum of the above function, we conclude that $H_\alpha^\uparrow(X|B)_\rho$ is monotonically decreasing in $\alpha\geq 0$.
		Hence, item \ref{H-aa} is proven.
	
	\item[(\ref{prop:H})-\ref{H-bb}]	
	This proof follows directly from item \ref{E0_SW-c} in Proposition \ref{prop:E0_SW+E0down}, Eq.~\eqref{eq:E0SW1}, and the substitution $\alpha = 1/(1+s)$.
\end{itemize}
\end{proof}

\section{Properties of Auxiliary Function and Error Exponent Function} \label{app:prop}
Let us recall Prop.~\ref{prop:E0_SW+E0down}.
\propaux*

\begin{proof}[Proof of Proposition \ref{prop:E0_SW+E0down}]
	~\\
	\begin{itemize}
		\item[(\ref{prop:E0_SW+E0down})-\ref{E0_SW-a}] (Continuity)
				Since $E_0 (s )$ admits a closed-form
		\begin{align}
		-\log \Tr \left[ \left( \Tr_{X} \rho_{XB}^{\frac{1}{1+s}} \right)^{1+s} \right], \quad \forall s> -1.
		\end{align}
		It is clearly continuous for all $s>-1$.

		Likewise, since $ {E}_0^{\downarrow}(s ) = - \log \Tr\left[ \rho_{XB}^{1-s} \left( \mathds{1}_X\otimes \rho_B\right)^s \right]$,
		it is continuous for all $s\geq 0$.
				
		\item[(\ref{prop:E0_SW+E0down})-\ref{E0_SW-b}] (Negativity)
				The negativity of $E_0(s)$ and $E_0^\downarrow(s)$ directly follows from the non-negativity of the conditional R\'enyi entropy and the definition, Eq.~\eqref{eq:E0SW1}.

		\item[(\ref{prop:E0_SW+E0down})-\ref{E0_SW-c}] (Concavity) 
		For $s\mapsto E^\downarrow_0(s)$, the claim follows from the concavity of the map $s\mapsto sD_{1-s}(\,\cdot\,\|\,\cdot\,)$, Eq.~\eqref{eq:chaotic1} in Lemma~\ref{lemma:chaotic}.

		Let us now consider $s\mapsto E_0(s)$.
				The concavity for $s\geq 0$ can be proved with the geometric matrix means in \cite{CH16}. Here, we present another proof by the following matrix inequality.
				
		\begin{lemm}{\cite[Corollary 3.6]{BL16}} \label{lemm:BL16}
			Let $A_i$ be $m\times m$ positive semi-definite matrix and $Z_i$ be $n \times m$ matrix for  $i = 1,\ldots, k$. Then, for all unitarily invariant norms $\|\cdot\|$ and $\gamma>0$, the map
			\begin{align}
			(p,t) \mapsto \left\| \left( \sum_{i=1}^k Z_i^* A_i^{t/p} Z_i \right)^{\gamma p} \right\|
			\end{align}
			is jointly log-convex on $(0,+\infty)\times (-\infty,+\infty)$.
		\end{lemm}
		
		Let $\rho_{XB} = \sum_{x\in\mathcal{X}} p(x)|x\rangle\langle x|\otimes W_x$,
		$t = \gamma = 1$, $i = x$, $k = |\mathcal{X}|$, $A_i = p(x) W_x$, and $Z_i = I_{n,m}$.
		We obtain the log-convexity of the map by applying Lemma~\ref{lemm:BL16}:
		\begin{align}
		p \mapsto \Tr \left( \sum_{x\in\mathcal{X}} (p(x) W_x)^{\frac{1}{p}} \right)^{p}, \quad \forall p>0,
		\end{align}
		which is exactly the concavity of the map $s\mapsto E_0 (s )$ for all $s>-1$.

		\item[(\ref{prop:E0_SW+E0down})-\ref{E0_SW-d}] (First-order derivative)
				By the definition of $E_0 (s )$,
		\begin{align}
		\left.\frac{ \partial E_0 (s )}{\partial s}\right|_{s=0}
		&= \left.- H^{\uparrow}_{\frac{1}{1+s}}(X|B)_\rho - s \frac{ \partial H^{\uparrow}_{\frac{1}{1+s}}(X|B)_\rho}{ \partial s }\right|_{s=0}
		= - H(X|B)_{\rho}.
		\end{align}
				Likewise, one can verify that 
		\begin{align}
		\left.\frac{\partial {E}_0^{\downarrow}(s,\rho_{XB}) }{\partial s}\right|_{s=0} &=  \left.D_{1-s} \left( \rho_{XB} \| \mathds{1}_X \otimes \rho_B \right) - s  D_{1-s}'\left( \rho_{XB}  \| \mathds{1}_X \otimes \rho_B \right)\right|_{s=0} \\
		&= \left. D_{1-s}\left( \rho_{XB}  \| \mathds{1}_X \otimes \rho_B \right) \right|_{s=0} \\
		&= D( \rho_{XB} \|\mathds{1}_X \otimes \rho_B) \\
		&= -H(X|B)_\rho. \label{eq:tE_I}
		\end{align}
		
		\item[(\ref{prop:E0_SW+E0down})-\ref{E0_SW-e}] (Second-order derivative)	We first consider $E_0(s)$.
				Similar to Item \ref{E0_SW-d}, it follows that 
		\begin{align}
		\left.\frac{ \partial^2 E_0 (s )}{\partial s^2}\right|_{s=0}
		&= \left.- 2\frac{ \partial H^{\uparrow}_{\frac{1}{1+s}}(X|B)_\rho}{ \partial s } - s \frac{ \partial^2 H^{\uparrow}_{\frac{1}{1+s}}(X|B)_\rho}{ \partial s^2 }\right|_{s=0}. \label{eq:temp2}
		\end{align}
		The above equation indicates that we need to evaluate the first-order derivative of $H^{\uparrow}_{\frac{1}{1+s}}(X|B)_\rho$ at $0$.
		In the following, we directly deal with the closed-form expression, Eq.~\eqref{eq:E0SW2}.
		
		To ease the burden of derivations, we denote some notation:
		\begin{align}
		f(s) &:= 
				\Tr_X  \rho_{XB}^{1/(1+s)}  , \label{eq:f0} \\
		g(s) &:= f(s)^{(1+s)}, \label{eq:g0} \\
		F(s) &:= \Tr\left[ g(s) \right] , \label{eq:F0}
		\end{align}	
		Then,
		\begin{align}
		\frac{ \partial E_0 (s )}{\partial s} &=  - \frac{F'(s)}{F(s)} \label{eq:E01}\\
		\frac{ \partial^2 E_0 (s )}{\partial s^2} &= -\frac{F''(s)}{F(s)} - \left( \frac{ \partial E_0 (s )}{\partial s} \right)^2.  \label{eq:E02} \end{align}		
		Direct calculation shows that
		\begin{align}
		f'(s) &=  -\frac1{(1+s)^2} \Tr_X \rho_{XB}^{1/(1+s)} {\log} \rho_{XB}, \label{eq:f1}\\
		f''(s) &=  \frac1{(1+s)^3} \Tr_X \rho_{XB} {\log} \rho_{XB} \cdot \left[ 2 + \frac{{\log}  \rho_{XB}}{(1+s)}\right]. \label{eq:f2}
				\end{align}
		Note that\footnote{Here, let's assume $\rho_{XB}$ has full support on $\mathcal{S}(XB)$ for brevity. The general case should hold with more technical derivations. } $g(s) =  \mathrm{e}^{(1+s) {\log} f(s) }$.
		By applying the chain rule of the Fr\'echet derivatives, one can show 
		\begin{align}
		g'(s) 
		&= 
		\mathsf{D}\exp\left[{\log} g(s)\right] \left( (1+s) \mathsf{D} {\log}\left[f(s)\right] \left( f'(s) \right)  + {\log} f(s) \right). \label{eq:g1}
														\end{align}
		Further, we employ Lemma~\ref{lemm:trace_Petz} and Eqs.\eqref{eq:F0}, ~\eqref{eq:g1}, to obtain
		\begin{align}
		F'(s) &= \Tr\left[ g'(s) \left( (1+s) \mathsf{D}\log[f(s)](f'(s)) + \log f(s) 	\right) 	\right], \label{eq:F1} \\
		F''(s)|_{s=0} &= \left.\Tr\left[ g'(s) \left( (1+s) \mathsf{D} {\log}\left[f(s)\right] \left( f'(s) \right)  + {\log} f(s) \right) \right]\right|_{s=0} \notag\\
		&\quad + \left. \Tr\left[ g(s) \left( 2 \mathsf{D} {\log}\left[f(s)\right] \left( f'(s) \right) + (1+s) 
		\left\{  \mathsf{D} {\log}\left[f(s)\right] \left( f''(s) \right)  \right. \right. \right. \right. \notag\\&\quad+ \left. \left.\left. \left. \mathsf{D}^2 {\log}\left[f(s)\right] \left( f'(s) \right)
		\right\} \right)\right]\right|_{s=0}.  \label{eq:F2_0}
		\end{align}		
		Before evaluating $F''(s)$ at $s=0$, note that Eqs.~\eqref{eq:f0}, \eqref{eq:g0}, \eqref{eq:f1}, \eqref{eq:f2}, and \eqref{eq:g1} yield
		\begin{align}
		f(0) &= g(0) = \rho_B, \\
		f'(0) &= - \Tr_X \rho_{XB} \log \rho_{XB}, \label{eq:f1_0}\\
		f''(0) &= 2 \Tr_X \rho_{XB} \log \rho_{XB} + \Tr_X \rho_{XB} \log^2 \rho_{XB}, \label{eq:f2_0} \\
		g'(0) &= \mathsf{D}\exp\left[{\log} g(0)\right] \left( (1+0) \mathsf{D} {\log}\left[f(0)\right] \left( f'(0) \right)  + {\log} f(0) \right) \\
		&= \mathsf{D}\exp\left[{\log} f(0)\right] \left( \mathsf{D} {\log}\left[f(0)\right] \left( f'(0) \right)  + {\log} f(0) \right) \\
		&= f'(0) + f(0) \log f(0) \label{eq:g1_0}\\
		&= - \Tr_X \rho_{XB} \log \rho_{XB} + \rho_B \log \rho_B.
		\end{align}		
		From Eqs.~\eqref{eq:g1_0}, \eqref{eq:F1}, the first term in Eq.~\eqref{eq:F2_0} leads to
		\begin{align}
		&\Tr\left[ g'(0) \left( (1+0)\mathsf{D}{\log}\left[ f(0) \right] \left( f'(0) \right) + {\log} f(0) \right) \right] \\
		&= \Tr\left[ f'(0) \mathsf{D}{\log}\left[ f(0) \right] \left( f'(0) \right)
		+ 2f'(0) {\log} f(0) + f(0) {\log}^2 f(0) 
		\right] \\
		&= \Tr\left[ f'(0) \mathsf{D}{\log}\left[ f(0) \right] \left( f'(0) \right)
		- 2 \Tr_X \rho_{XB} \log \rho_{XB} \cdot {\log} \rho_B + \rho_B {\log}^2 \rho_B 
		\right]
		\label{eq:F2_01}
		\end{align}
		Further, from Eqs.~\eqref{eq:f2}, \eqref{eq:f1_0}, and \eqref{eq:f2_0}, the second term in Eq.~\eqref{eq:F2_0} leads to
		\begin{align}
		&\Tr\left[ f(0) \left( 2 \mathsf{D} {\log}\left[f(0)\right] \left( f'(0) \right) +  
		\left\{  \mathsf{D} {\log}\left[f(0)\right] \left( f''(0) \right)  \right. \right. \right.  \notag\\&\quad+  \left.\left.\left. \mathsf{D}^2 {\log}\left[f(0)\right] \left( f'(0) \right)
		\right\} \right)\right]\\
		&=\Tr\left[ 2f'(0) + f''(0) - f'(0) \mathsf{D}{\log}\left[ f(0) \right] \left( f'(0) \right)    \right] \\
		&= \Tr\left[  
		\Tr_X \rho_{XB} \log^2 \rho_{XB} 
		-f'(0) \mathsf{D}{\log}\left[ f(0) \right] \left( f'(0) \right)  \right]. \label{eq:F2_02}
		\end{align}
		Combining  Eqs.~\eqref{eq:F2_0}, \eqref{eq:F2_01}, \eqref{eq:F2_02} gives
		\begin{align}
		F''(0) = \Tr\left[  \rho_{XB}\left( \log \rho_{XB} - \log \mathds{1}_X \otimes \rho_B \right)^2
		\right]. \label{eq:F2_03}
		\end{align}
		Finally, Eqs.~\eqref{eq:E02} and \eqref{eq:F2_03} conclude our result:
		\begin{align}
		\left.\frac{\partial E_0 (s )}{\partial s}\right|_{s=0} = - V(\rho_{XB}\|\mathds{1}_X\otimes \sigma_B) =
		-V(X|Y)_{\rho}.
		\end{align}	
		Moreover, Eq.~\eqref{eq:temp2} gives
		\begin{align}
		\left.\frac{ \partial H^{\uparrow}_{\alpha}(X|B)_\rho}{ \partial \alpha }\right|_{\alpha = 0} = \frac12 V(X|B)_\rho.
		\end{align}

		For $E_0^\downarrow$, continuing from item \ref{E0_SW-d}, one obtains
		\begin{align}
	\left.\frac{\partial^2  {E}_0^{\downarrow}(s ) }{\partial s^2}\right|_{s=0} &= \left. -2 D_{1-s}' \left( \rho_{XB}  \| \mathds{1}_X \otimes \rho_B \right) + s D_{1-s}''\left( \rho_{XB}  \| \mathds{1}_X \otimes \rho_B \right)\right|_{s=0} \\
	&= \left. -2 D_{1-s}' \left( \rho_{XB}  \| \mathds{1}_X \otimes \rho_B \right)\right|_{s=0} \\
	&= - V( \rho_{XB} \| \mathds{1}_X \otimes \rho_B ) \label{eq:tE_V2} \\
	&= V(X|B)_\rho, 
	\end{align}
	where in equality \eqref{eq:tE_V2} we use the fact $D_{1/1+s}'(\cdot\|\cdot)|_{s=0} = V (\cdot\|\cdot)/2$ \cite[Theorem 2]{LT15}.

		\begin{lemm}[{\cite[Theorem 3.23]{HP14}}] \label{lemm:trace_Petz}
			Let $ \bm{A}, \bm{X}$ be $d\times d $ Hermitian matrices, and $t\in\mathbb{R}$. Assume $f:I\to \bbR$ is a continuously differentiable function.
						Then
			\[
			\left.\frac{\mathrm{d} }{ \mathrm{d}t} \Tr f(\bm{A}+t \bm{X})\right|_{t=t_0} = \Tr [  \bm{X} f' ( \bm{A} + t_0 \bm{X}) ].
			\]
		\end{lemm}	
	\end{itemize}
\end{proof}

\begin{prop5}[Properties of the Exponent Function] 	Let $\rho_{XB}$ be a classical-quantum state with $H(X|B)_{\rho}>0$, the following holds.
	\begin{enumerate}[(a)]
		\item\label{E_SW-aa}	$E_\textnormal{sp} (\cdot)$ is convex, differentiable, and monotonically increasing on $[0,+\infty]$. Further,
		\begin{align}
		E_\textnormal{sp} (R) = \begin{cases}
		0, & R\leq H_1^{\uparrow}(X|Y)_{\rho} \\
		E_\textnormal{r} (R) , & H_{1}^{\uparrow}(X|Y)_{\rho} \leq R \leq  H_{1/2}^{\uparrow}(X|Y)_{\rho} \\
		+\infty, & R >  H_{0}^{\uparrow}(X|Y)_{\rho}
		\end{cases}.
		\end{align}
		
		\item\label{E_SW-bb}	Define
		\begin{align} \label{eq:FF}
		F_R(\alpha,\sigma_B) := \begin{dcases}\frac{1-\alpha}{\alpha} \left( R + D_\alpha\left(\rho_{XB}\| \mathds{1}_X \otimes \sigma_B  \right) \right), &\alpha\in(0,1), \\
		0, &\alpha = 1,
		\end{dcases}
		\end{align}
		on $(0,1]\times \mathcal{S}(B)$.
		For $R\in(H_1^{\uparrow}(X|Y)_\rho, H_0^{\uparrow}(X|Y)_\rho)$, there exists a unique saddle-point $(\alpha^\star,\sigma^\star) \in (0,1)\times\mathcal{S}(B)$ of $F_R(\cdot,\cdot)$ such that
		\begin{align}
		F_R(\alpha^\star,\sigma^\star) = \sup_{\alpha\in [0,1] } \inf_{\sigma_B \in \mathcal{S}(B)} F_R(\alpha,\sigma_B) = \inf_{\sigma_B \in \mathcal{S}(B)}\sup_{\alpha\in [0,1] }   F_R(\alpha,\sigma_B) = E_\textnormal{sp} (R).
		\end{align}
		
		\item\label{E_SW-cc} Any saddle-point $(\alpha^\star, \sigma^\star)$ of $F_{R}(\cdot,\cdot)$ satisfies 		\begin{align}
		\mathds{1}_X\otimes \sigma^\star 
		\gg \rho_{XB}.
		\end{align}

													\end{enumerate}	
\end{prop5}

\begin{proof}[Proof of Proposition \ref{prop:E_SW}]
	~\\
\begin{itemize}
	\item[(\ref{prop:E_SW})-\ref{E_SW-aa}]
	Item \ref{H-a} in Proposition \ref{prop:H} shows that the map $\alpha \mapsto H_\alpha^\uparrow(X|B)_\rho$ is monotonically decreasing on $[0,1]$. Hence, from the definition:
	\begin{align}
	E_\text{sp} (R) := \sup_{\alpha\in(0,1]} \frac{1-\alpha}{\alpha} \left( R - H_\alpha^\uparrow (X|B)_\rho\right) ,
	\end{align}
	it is not hard to verify that $E_\text{sp} (R) = +\infty$ for all $ R > H_0^\uparrow(H|B)_\rho$; finite for all $  R< H_0^\uparrow(H|B)_\rho$; and $E_\text{sp}^\textsf{SW}(R) = 0$, for all $ R\geq  H_1^\uparrow(H|B)_\rho$.	
	Moreover, $E_\text{sp} (R) = E_\text{r} (R)$ for $R\in[H_{1}^{\uparrow}(X|Y)_{\rho},  H_{1/2}^{\uparrow}(X|Y)_{\rho} ]$ by the definition in Eq.~\eqref{eq:gallager_r2}.
	
	For every $\alpha\in(0,1]$, the function $\frac{1-\alpha}{\alpha} ( R - H_\alpha^\uparrow(X|B)_\rho )$ is an non-decreasing, convex, and continuous function in $R\in\mathbb{R}_{>0}$. Since $E_\text{sp} (R)$ is the pointwise supremum of the above function, $E_\text{sp} (R)$ is non-decreasing, convex, and lower semi-continuous function for all $R\geq 0$. 
	Furthermore, since a convex function is continuous on the interior of the interval if it is finite \cite[Corollary 6.3.3]{Dud02}, thus $E_\text{sp} (R)$ is continuous for all $R < H_0^\uparrow(X|B)_\rho$, and continuous from the left at $R = H_0^\uparrow(X|B)_\rho$.

	\item[{(\ref{prop:E_SW})-\ref{E_SW-bb}}]
		Let 
		\begin{align}
		\mathcal{S}_\rho(B) := \left\{ \sigma_B \in \mathcal{S}(B): \rho_{XB} \not\perp  \mathds{1}_X \otimes \sigma_B   \right\}.
		\end{align}
		Fix an arbitrary $R \in (H_1^\uparrow (X|B)_\rho, H_0^\uparrow (X|B)_\rho)$.
		In the following, we first prove the existence of a saddle-point of $F_{R}(\cdot,\cdot)$ on $(0,1	] \times \mathcal{S}_{\rho}(B)$. Ref.~\cite[Lemma 36.2]{Roc64} states that $(\alpha^\star,\sigma^\star)$ is a saddle point of $F_{R}(\cdot,\cdot)$ if and only if the supremum in 
		\begin{align}
		\sup_{\alpha\in (0,1]} \inf_{\sigma \in \mathcal{S}_{\rho}(B)} F_{R} (\alpha,\sigma) \label{eq:saddle15}
		\end{align}
		is attained at $\alpha^\star\in(0,1]$, the infimum in 
		\begin{align}
		\inf_{\sigma \in \mathcal{S}_{\rho}(B)} \sup_{\alpha\in (0,1]}  F_{R} (\alpha,\sigma) \label{eq:saddle16}
		\end{align}
		is attained at $\sigma^\star\in \mathcal{S}_{\rho}(B)$, and the two extrema in Eqs.~\eqref{eq:saddle15}, \eqref{eq:saddle16} are equal and finite.
		We first claim that, $\forall \alpha\in(0,1],$
		\begin{align}
		\inf_{\sigma \in \mathcal{S}_{\rho}(B)} F_{R}(\alpha,\sigma)
		= \inf_{\sigma \in \mathcal{S}(B)} F_{R}(\alpha,\sigma). \label{eq:saddle21}
		\end{align}
		To see this, observe that for any $\alpha \in (0,1)$, Eqs.~\eqref{eq:Petz} yield
		\begin{align}
		\forall \sigma \in \mathcal{S}(B)\backslash \mathcal{S}_{\rho}(B), \quad 
				D_{\alpha}\left( \rho_{XB}\|\mathds{1}_X \otimes \sigma \right) = +\infty, 
		\label{eq:saddle23}
		\end{align}
		which, in turn, implies
		\begin{align} \label{eq:saddle22}
		\forall \sigma \in \mathcal{S}(B)\backslash \mathcal{S}_{\rho}(B), \quad 
		F_{R}(\alpha,\sigma ) = +\infty.
		\end{align}
				Further, Eq.~\eqref{eq:saddle21} holds trivially when $\alpha=1$.
		Hence, Eq.~\eqref{eq:saddle21} yields
		\begin{align}
		\begin{split}
		\sup_{\alpha\in (0,1]} \inf_{\sigma \in \mathcal{S}_{\rho}(B)} F_{R} (\alpha,\sigma)
		&= \sup_{\alpha\in (0,1]} \inf_{\sigma \in \mathcal{S}(B)} F_{R} (\alpha,\sigma) 
				\end{split}
		\end{align}
						Owing to the fact $R< H_0^\uparrow(X|B)_\rho$ and Eq.~\eqref{eq:gallager_sp2}, we have
		\begin{align}
		E_\text{sp}  (R) = \sup_{\alpha\in (0,1]} \inf_{\sigma \in \mathcal{S}(B)} F_{R} (\alpha,\sigma) < +\infty, \label{eq:saddle28}
		\end{align}
		which guarantees the supremum in the right-hand side of Eq.~\eqref{eq:saddle28} is attained at some $\alpha\in(0,1]$. Namely, there exists some $\bar{\alpha}_{R} \in (0,1]$ such that
		\begin{align} \label{eq:fact1}
		\sup_{\alpha\in (0,1]} \inf_{\sigma \in \mathcal{S}_{\rho}(B)} F_{R} (\alpha,\sigma)
		= \max_{\alpha\in [\bar{\alpha}_{R},1]} \inf_{\sigma \in \mathcal{S}(B)} F_{R} (\alpha,\sigma)  < +\infty.
		\end{align}
		Thus, we complete our claim in Eq.~\eqref{eq:saddle15}. It remains to show that the infimum in Eq.\eqref{eq:saddle16} is attained at some $\sigma^\star \in \mathcal{S}_{\rho}(B)$ and the supremum and infimum are exchangeable. To achieve this, we will show that $\left( [\bar{\alpha}_{R},1], \mathcal{S}_{\rho}(B), F_{R} \right)$ is a closed saddle-element (see Definition \ref{defn:saddle} below) and employ the boundedness of $[\bar{\alpha}_{R},1]\times \mathcal{S}_{\rho}(B)$ to conclude our claim.

		\begin{defn} [Closed Saddle-Element {\cite{Roc64}}] \label{defn:saddle}
			We denote by $\texttt{ri}$ and $\texttt{cl}$ the relative interior and the closure of a set, respectively.	Let $\mathcal{A},\mathcal{B}$ be subsets of a real vector space, and $F:\mathcal{A}\times \mathcal{B}\to\mathbb{R}\cup\{\pm \infty  \}$.	
			The triple $\left(\mathcal{A},\mathcal{B}, F\right)$ is called a closed saddle-element if for any $x\in \texttt{ri}\left(\mathcal{A}\right)$ (resp.~$y\in \texttt{ri}\left(\mathcal{B}\right)$), 
			\begin{itemize}
				\item[(i)] $\mathcal{B}$ (resp.~$\mathcal{A}$) is convex.
				\item[(ii)] $F(x,\cdot)$ (resp.~$F(\cdot, y)$) is convex (resp.~concave) and lower (resp.~upper) semi-continuous.
				\item[(iii)] Any accumulation point of $\mathcal{B}$ (resp.~$\mathcal{A}$) that does not belong to $\mathcal{B}$ (resp.~$\mathcal{A}$), say $y_o$ (resp.~$x_o$) satisfies $\lim_{y\to y_o} F(x,y) = +\infty$ (resp.~$\lim_{x\to x_o} F(x, y) = -\infty$).
			\end{itemize}
		\end{defn}
		
		Fix an arbitrary $\alpha \in \texttt{ri}\left( [\bar{\alpha}_{R},1] \right) = (\bar{\alpha}_{R},1)$.
		We check that $\left( \mathcal{S}_{\rho}(B), F_{R}(\alpha, \cdot)\right)$ fulfills the three items in Definition \ref{defn:saddle}.
		(i) The set $\mathcal{S}_{\rho}(B)$ is clearly convex.
		(ii) Eq.~\eqref{eq:chaotic4} in Lemma \ref{lemma:chaotic} implies that $\sigma\mapsto D_{\alpha}(W_x\|\sigma)$ is convex and lower semi-continuous. Since convex combination preservers the convexity and the  lower semi-continuity, Eq.~\eqref{eq:FF} yields that  $\sigma \mapsto F_{R}(\alpha,\sigma)$ is convex and lower semi-continuous on $\mathcal{S}_{\rho}(B)$.
		(iii) Due to the compactness of $\mathcal{S}(B)$, any accumulation point of $\mathcal{S}_{\rho}(B)$ that does not belong to $\mathcal{S}_{\rho}(B)$, say $\sigma_o$, satisfies $\sigma_o \in \mathcal{S}(B) \backslash \mathcal{S}_{\rho}(B)$. Eqs.~\eqref{eq:saddle23} and \eqref{eq:saddle22} then show that $F_{R}(\alpha, \sigma_o) = +\infty$.
		
		Next, fix an arbitrary $\sigma \in \texttt{ri}\left(  \mathcal{S}_{\rho}(B) \right)$. Owing to the convexity of $\mathcal{S}_{\rho}(B)$, it follows that $\texttt{ri}\left(  \mathcal{S}_{\rho}(B) \right) $ $= \texttt{ri}\left(\texttt{cl}\left( \mathcal{S}_{\rho}(B)\right)\right)$ (see e.g.~\cite[Theorem 6.3]{Roc70}).
		We first claim $\texttt{cl}\left( \mathcal{S}_{\rho}(B)\right) = \mathcal{S}(B)$. To see this, observe that $\mathcal{S}_{>0}(B) \subseteq  \mathcal{S}_{\rho}(B)$ since a full-rank operator  is not orthogonal with $\rho_{XB}$.
		Hence, 
		\begin{align}
		\mathcal{S}(B)=
		\texttt{cl}\left( \mathcal{S}_{>0}(B) \right) 
		\subseteq \texttt{cl} \left(  \mathcal{S}_{\rho}(B) \right). \label{eq:saddle25}
		\end{align}
		On the other hand, the fact $\mathcal{S}_{\rho}(B) \subseteq  \mathcal{S}(B)$ leads to
		\begin{align}
		\texttt{cl}\left( \mathcal{S}_{\rho}(B) \right) \subseteq  
		\texttt{cl}\left(\mathcal{S}(B)\right) 
		= \mathcal{S}(B). \label{eq:saddle26}
		\end{align}
		By Eqs.~\eqref{eq:saddle25} and \eqref{eq:saddle26}, we deduce that 
		\begin{align}
		\texttt{ri}\left(  \mathcal{S}_{\rho}(B) \right)
		= \texttt{ri}\left( \texttt{cl}\left( \mathcal{S}_{\rho}(B) \right) \right)
		= \texttt{ri}\left(  \mathcal{S}(B) \right)
		=  \mathcal{S}_{>0}( {B}), \label{eq:saddle24}
		\end{align}
		where the last equality in Eq.~\eqref{eq:saddle24} follows from \cite[Proposition 2.9]{Wei11}.
		Hence, we obtain
		\begin{align} \label{eq:saddle19}
		\forall \sigma \in \texttt{ri}\left(  \mathcal{S}_{\rho}(B) \right) \quad \text{and} \quad \mathds{1}_X\otimes \sigma \gg \rho_{XB}.
		\end{align}
		Now we verify that $\left( [\bar{\alpha}_{R},1], F_{R}(\cdot,\sigma)\right)$ satisfies the three items in Definition \ref{defn:saddle}.
		Fix an arbitrary $\sigma \in \texttt{ri}\left(  \mathcal{S}_{\rho}(B) \right)$.
		(i) The set $(0,1]$ is obviously convex.
		(ii) From Eq.~\eqref{eq:chaotic5} in Lemma \ref{lemma:chaotic}, the map $\alpha \mapsto F_{R}(\alpha,\sigma)$ is continuous on $(0,1)$. Further, it is not hard to verify that $F_{R}(1,\sigma) = 0 = \lim_{\alpha\uparrow 1} F_{R}(\alpha,\sigma)$ from Eqs.~\eqref{eq:saddle19}, \eqref{eq:FF}, and \eqref{eq:Petz}.
		Item \ref{H-b} in Proposition~\ref{prop:H} implies that $\alpha \mapsto F_{R}(\alpha,\sigma)$ on $[\bar{\alpha}_R,1)$ is concave.
		Moreover, the continuity of $\alpha \mapsto F_{R}(\alpha,\sigma)$ on $[\bar{\alpha}_{R},1)$ guarantees the concavity of $\alpha \mapsto F_{R}(\alpha,\sigma)$ on $[\bar{\alpha}_{R},1]$.
		(iii) Since $[\bar{\alpha}_{R},1]$ is closed, there is no accumulation point of $[\bar{\alpha}_{R},1]$ that does not belong to $[\bar{\alpha}_{R},1]$.
		
		We are at the position to prove the saddle-point property.
		The closed saddle-element, along with the boundedness of $\mathcal{S}_{\rho}(B)$ and Rockafellar's saddle-point result \cite[Theorem 8]{Roc64}, \cite[Theorem 37.3]{Roc70} imply that
		\begin{align}
		- \infty < \sup_{ \alpha\in [\bar{\alpha}_{R},1] } \inf_{\sigma \in \mathcal{S}_{\rho}(B)} F_{R} (s,\sigma)
		= \min_{\sigma \in \mathcal{S}_{\rho}(B)} \sup_{ \alpha\in[\bar{\alpha}_{R},1] }  F_{R} (s,\sigma). \label{eq:fact2}
		\end{align}
		Then Eqs.~\eqref{eq:fact1} and \eqref{eq:fact2} lead to the existence of a saddle-point of $F_{R}(\cdot,\cdot)$ on $(0,1]\times \mathcal{S}_{\rho}(B)$.

		Next, we prove the uniqueness.
		The rate $R$ and item \ref{E_SW-aa} in Proposition \ref{prop:E_SW} shows that
		\begin{align} \label{eq:saddle1}
		\sup_{0<\alpha\leq 1} \min_{\sigma \in \mathcal{S}(B)}  F_{R} (\alpha, \sigma) \in \mathbb{R}_{>0}.
		\end{align}
		Note that $\alpha^\star = 1$ will not be  a saddle point of $F_{R,P}(\cdot,\sigma)$ because $F_{R} (1, \sigma) = 0$, $\forall\sigma \in \mathcal{S}(B)$,  contradicting Eq.~\eqref{eq:saddle1}.
		
		Now,  fix $\alpha^\star\in(0,1)$ to be a saddle-point of $F_{R}(\cdot,\cdot)$.
		Eq.~\eqref{eq:chaotic4} in Lemma \ref{lemma:chaotic} implies that the map $\sigma\mapsto D_{\alpha^\star}(\rho_{XB}\|\mathds{1}_X\otimes \sigma )$ is strictly convex, and thus the minimizer of Eq.~\eqref{eq:saddle1} is unique.	
						Next, let $\sigma^\star \in \mathcal{S}_{\rho}(B)$ be a saddle-point of $F_{R}(\cdot,\cdot)$. Then,
		\begin{align}
		F_{R}(\alpha,\sigma^\star ) &= \frac{1-\alpha}{\alpha}\left( R - H_\alpha^\uparrow(X|B)_\rho \right). \label{eq:saddle7} 
		\end{align}
		Item \ref{H-b} in Proposition \ref{prop:H} then shows that $\frac{1-\alpha}{\alpha} H_\alpha^\uparrow(X|B)_\rho$ is strictly concave on $(0,1)$, which in turn implies that $F_{R}(\cdot, \sigma^\star)$ is also strictly concave on $(0,1)$.	
		Hence, the maximizer of Eq.~\eqref{eq:saddle1} is unique, which completes item \ref{E_SW-bb} of Proposition \ref{prop:E_SW}.
	
	\item[(\ref{prop:E_SW})-\ref{E_SW-cc}]
		As shown in the proof of item \ref{E_SW-bb}, $\alpha^\star = 1$ is not a saddle point of $F_{R}(\cdot,\cdot)$ for any $R< H_0^\uparrow(X|B)_\rho$.
		We assume $(\alpha^\star, \sigma^\star)$ is a saddle-point of $F_{R}(\cdot,\cdot)$ with $\alpha^\star \in (0,1)$, it holds that
		\begin{align} \label{eq:saddle3}
		F_{R} (\alpha^\star, \sigma^\star) = \min_{\sigma\in\mathcal{S}B} F_{R} (\alpha^\star, \sigma)
		= \frac{1-\alpha^\star}{\alpha^\star} R + \frac{1-\alpha^\star}{\alpha^\star}  \min_{\sigma\in\mathcal{S}(B)}  D_{\alpha^\star} (\rho_{XB}\|\mathds{1}_X \otimes \sigma).
		\end{align}		
		It is known \cite{SW12}, \cite[Lemma 1]{TBH14}, \cite[Lemma 5.1]{Tom16} that the minimizer of Eq.~\eqref{eq:saddle3} is
		\begin{align} \label{eq:saddle6}
		\sigma^\star =  \frac{ \left( \Tr_X\left[\rho_{XB}^{\alpha^\star}\right] \right)^{\frac{1}{\alpha^\star}} }{ \Tr\left[\left( \Tr_X\left[\rho_{XB}^{\alpha^\star}\right] \right)^{\frac{1}{\alpha^\star}}\right] }.
		\end{align}
		From this expression, it is clear that
		$\mathds{1}_X\otimes \sigma^\star \gg \rho_{XB}$,
		and thus item \ref{E_SW-cc} is proved.

\end{itemize}
\end{proof}

\begin{proof}[Proof of Proposition~\ref{prop:representation}]
We only provide the proof for Eq~\eqref{eq:representation3} since Eqs.~\eqref{eq:representation1} and \eqref{eq:representation2} follow similarly.

Starting with the left-hand side,
\begin{align}
 \sup_{-1 < s<0} \left\{ E_0^\flat(s,\rho_{XB}) + sR \right\} &=  \sup_{-1 < s<0} \left\{ -s H^{\flat,\uparrow}_{\frac{1}{1+s}}(X|B)_\rho + sR \right\}\\
 &=\sup_{-1 < s<0} \left\{ -s \max_{\tau_B} -D^{\flat}_{\frac{1}{1+s}}(\rho_{XB}\|\one_X\otimes\tau_B) + sR \right\}\\
 &=\sup_{-1 < s<0} \max_{\tau_B}\left\{ s D^{\flat}_{\frac{1}{1+s}}(\rho_{XB}\|\one_X\otimes\tau_B) + sR \right\}.
\end{align}

Next, Lemma~\ref{lemm:MO17} gives the variational representation 
\begin{equation}
D^{\flat}_{\frac{1}{1+s}}(\rho_{XB}\|\one_X\otimes\tau_B) = \max_{\sigma_{XB}} \left[ D(\sigma_{XB}\| \one_X\otimes \tau_B) + \frac{1}{s} D(\sigma_{XB}\| \rho_{XB}) \right]
\end{equation}
and shows that the optimizer is a state $\sigma_{XB}^*\in \mathcal{S}_\rho(XB)$. Therefore,  we can write
\begin{equation}
D^{\flat}_{\frac{1}{1+s}}(\rho_{XB}\|\one_X\otimes\tau_B) =\max_{\sigma_{XB}\in \mathcal{S}_\rho(XB)} \left[ D(\sigma_{XB}\| \one_X\otimes \tau_B) + \frac{1}{s} D(\sigma_{XB}\| \rho_{XB}) \right]
\end{equation}

Substituting this variational representation into our expansion of $E^\flat_\text{sc}(R)$, we have
\begin{align*}
 E^\flat_\text{sc}(R) &=\sup_{-1 < s<0} \max_{\tau_B} \min_{\sigma_{XB}\in \mathcal{S}_\rho(XB)} \left\{ s \left[ D(\sigma_{XB}\| \one_X\otimes \tau_B) + \frac{1}{s} D(\sigma_{XB}\| \rho_{XB})  \right]+ sR \right\}\\
  &=\sup_{-1 < s<0}s  \min_{\tau_B}\max_{\sigma_{XB}\in \mathcal{S}_\rho(XB)} \left\{  D(\sigma_{XB}\| \one_X\otimes \tau_B) + \frac{1}{s}D(\sigma_{XB}\| \rho_{XB}) + R \right\}.
\end{align*}

 Consider the function
\begin{equation}
G(\tau_B, \sigma_{XB} ) = D(\sigma_{XB}\| \one_X\otimes \tau_B) + \frac{1}{s}D(\sigma_{XB}\| \rho_{XB}).
\end{equation}

\begin{lemm}[Properties of $G$] Let $s\in (-1,0)$. As a function of states $\tau_B$ in $\mathcal{S}(B)$ and $\sigma_{XB}$ in $\mathcal{S}(XB)$,
\begin{enumerate}
  \item $G(\tau_B,\sigma_{XB} )$ is convex and lowersemicontinuous in $\tau_B$,
  \item $G(\tau_B, \sigma_{XB} )$ is concave and continuous in $\sigma_{XB}$.
\end{enumerate}
\end{lemm}
\begin{proof}  
Since the only dependence in $\tau_B$ is in the second argument of the relative entropy, $G$ is convex and lowersemicontinuous in $\tau_B$. To prove the second part, we expand $G$ as
\begin{align}
G(\tau_B, \sigma_{XB} ) &= - H(\sigma_{XB}) - \tr[\sigma_{XB} \log (\one_X\otimes \tau_B)] - \frac{1}{s}H(\sigma_{XB}) - \tr[\sigma_{XB} \log \rho_{XB}]\\
&= - \left( 1 + \frac{1}{s} \right) H(\sigma_{XB}) - \tr[\sigma_{XB} ( \log(\one_X\otimes \tau_B) + \log \rho_{XB})].
\end{align}
Since $s\in (-1,0)$, the coefficient of $H(\sigma_{XB})$ is positive. The second term is linear in $\sigma_{XB}$, so $G$ is concave and continuous in $\sigma_{XB}$.
\end{proof}

By these properties of $G$, the compactness and convexity of $\mathcal{S}(B)$ and $\mathcal{S}_\rho(XB)$, we may apply the min-max theorem given by Theorem II.7 of~\cite{MO17} to find
\begin{equation}
\min_{\tau_B}\max_{\sigma_{XB}\in \mathcal{S}_\rho(XB) } G(\tau_B,\sigma_{XB}) =\max_{\sigma_{XB}\in \mathcal{S}_\rho(XB) } \min_{\tau_B}G(\tau_B,\sigma_{XB}).
\end{equation}
Therefore,
\begin{align}
 E^\flat_\text{sc}(R)  &= \sup_{-1 < s<0}s  \max_{\sigma_{XB}\in \mathcal{S}_\rho(XB)} \min_{\tau_B} \left\{  D(\sigma_{XB}\| \one_X\otimes \tau_B) + \frac{1}{s}D(\sigma_{XB}\| \rho_{XB}) + R \right\}\\
  &=\sup_{-1 < s<0} \min_{\sigma_{XB}\in \mathcal{S}_\rho(XB)} \left\{ -s \left[ \max_{\tau_B} -D(\sigma_{XB}\| \one_X\otimes \tau_B) - \frac{1}{s} D(\sigma_{XB}\| \rho_{XB})  \right]+ sR \right\}\\
  &=\sup_{-1 < s<0} \min_{\sigma_{XB}\in \mathcal{S}_\rho(XB)} \left\{ -s \left[ H(X|B)_\sigma - \frac{1}{s} D(\sigma_{XB}\| \rho_{XB})  \right]+ sR \right\}\\
  &=\sup_{0<s<1} \min_{\sigma_{XB}\in \mathcal{S}_\rho(XB)} \left\{  D(\sigma_{XB}\| \rho_{XB}) +s (H(X|B)_\sigma - R )\right\}\\
  &\leq \min_{\sigma_{XB}\in \mathcal{S}_\rho(XB)} \sup_{0<s<1}\left\{  D(\sigma_{XB}\| \rho_{XB}) +s (H(X|B)_\sigma - R )\right\} \label{eq:step-ieq}\\
  &= \min_{\sigma_{XB}\in \mathcal{S}_\rho(XB)} \left\{  D(\sigma_{XB}\| \rho_{XB}) +\sup_{0<s<1} s (H(X|B)_\sigma - R )\right\}\\
  &= \min_{\sigma_{XB}\in \mathcal{S}_\rho(XB)} \left\{  D(\sigma_{XB}\| \rho_{XB}) +|H(X|B)_\sigma - R |_+\right\}.
  \end{align}
  However, we can always achieve equality by taking $s\to 0$ or $s\to 1$, and therefore the inequality in \eqref{eq:step-ieq} is an equality.

\begin{lemm}[{\cite[Theorem III.5]{MO17}}] \label{lemm:MO17}
	Let $\rho,\tau \in\mathcal{S(H)}$ with $\rho\ll \tau$.
	For all $s>-1$, it follows that
	\begin{align}
	\min_{\sigma \in \mathcal{S(H)}} D(\sigma\|\rho) + sD(\sigma\|\tau)
	= s D^\flat_{\frac{1}{1+s}}(\rho\|\tau).
	\end{align}
	Additionally, the minimum is achieved uniquely by 
	\begin{equation}
	\sigma^* := P \e^{ \alpha P (\log \rho) P + (1-\alpha) P (\log \sigma) P} / Q_\alpha^\flat(\rho\|\sigma)
	\end{equation}
	where $\alpha = \frac{1}{1+s}$, and $P$ is the projection onto the support of $\rho$. 
\end{lemm}
\end{proof}

\section{Proof of Proposition~\ref{prop:spCh}} \label{app:spCh}
\begin{prop6}[Error Exponent around Conditional Entropy]
	Let  $(\delta_n)_{n\in\mathbb{N}}$ be a sequence of positive numbers with $\lim_{n\to+\infty} \delta_n = 0$.
					The following hold:
	\begin{align}
	\limsup_{n\to+\infty} \frac{ E_\textnormal{sp} \left( H(X|B)_\rho + \delta_n \right) }{ \delta_n^2 } &\leq   \frac{ 1 }{2 V(X|B)_\rho }; \label{eq:spChc11} \\
	\limsup_{n\to+\infty} \frac{s_n^\star}{\delta_n} &= \frac{1}{V(X|B)_\rho}, \label{eq:spChc33}
	\end{align}
	where 
	\begin{align}
	s_n^\star := \argmax_{s\geq 0} \left\{ s\left( H(X|B)_\rho + \delta_n \right) - s H_{\frac{1}{1+s}}^\uparrow(X|B)_\rho
	\right\}.
	\end{align}
\end{prop6}

\begin{proof}[Proof of Proposition~\ref{prop:spCh}]
For notational convenience, we denote by $H := H(X|B)_\rho$, $V := V(X|B)_\rho$.. Thus,
\begin{align}\label{eq_D1}
E_\textnormal{sp}  \left( R \right) &= \sup_{s\geq 0} \left\{  s R + E_0 (s)\right\}, \\
\end{align}
Let a \emph{critical rate} to be
\begin{align} \label{eq:critical}
r_\textnormal{cr} := \left.\frac{\partial E_0(s) }{\partial s}\right|_{s=1}.
\end{align}
Let $N_0$ be the smallest integer such that $H(X|B)_\rho + \delta_n < r_\text{cr}$, for all $ n\geq N_0$. 
Since the map $r \mapsto E_\text{sp}  (r) $ is non-increasing by item \ref{E_SW-a} in Proposition~\ref{prop:E_SW}, the maximization over $s$ in Eq.~\eqref{eq_D1} can be restricted to the set $[0,1]$ for any rate  below $r_\text{cr}$, i.e.,
\begin{align} \label{eq:sp2}
E_\text{sp}  \left( H +\delta_n \right)
=  \max_{0\leq s\leq 1} \left\{  s\left(H +\delta_n\right) + E_0(s)        \right\}.
\end{align}
For every $n\in\mathbb{N}$, let $s_n^\star$ attain the maxima in Eq.~\eqref{eq:sp2} at a rate of $H  + \delta_n$.
It is not hard to observe that $s_n^\star>0$ for all $n\geq N_0$ since $s_n^\star = 0 $ if and only if $H + \delta_n < H $, which violates the assumption of $\delta_n>0$ for finite $n$.
Now, we will show Eq.~\eqref{eq:spChc33} and
\begin{align}
\lim_{n\to +\infty} s_n^\star = 0.
\end{align}
Let $(s_{n_k}^\star)_{k\in\mathbb{N} }$ be arbitrary subsequences.
Since $[0,1]$ are compact, we may assume that
\begin{align} 
\lim_{k\to\infty} s_{n_k}^\star = s_o, \label{eq:s_n_0}
\end{align}
for some $s_o \in [0,1]$.

Since $s\mapsto E_0(s)$ is strictly concave from item \ref{E0_SW-c} in Proposition~\ref{prop:E0_SW+E0down}, the maximizer $s_n^\star$ must satisfy
\begin{align} \label{eq:s_n3}
\left.\frac{\partial E_0(s)}{\partial s}\right|_{s=s_{n_k}^\star} = -( H + \delta_{n_k}),
\end{align}
which together with item \ref{E0_SW-a} in Proposition~\ref{prop:E0_SW+E0down} implies
\begin{align}
\lim_{k\to+\infty} \left.\frac{\partial E_0(s)}{\partial s}\right|_{s=s_{n_k}^\star}
= \left.\frac{\partial E_0(s)}{\partial s}\right|_{s=s_{o}} = - H.
\end{align}
On the other hand, item \ref{E0_SW-d} in Proposition~\ref{prop:E0_SW+E0down} gives
\begin{align} \label{eq:s_n7}
\left.\frac{ \partial E_0(s)}{\partial s}\right|_{s=0} = -H .
\end{align}
Since item~\ref{E0_SW-e} in Proposition \ref{prop:E0_SW+E0down} guarantees
\begin{align} \label{eq:s_n6}
\left.\frac{\partial^2 E_0\left(s \right)}{\partial s^2}\right|_{s=0} = -V < 0,
\end{align}
which implies that the first-order derivative $\partial E_0\left(s \right)/{\partial s}$ is strictly decreasing around $s=0$. Hence, we conclude $s_o = 0$.
Because the subsequence is arbitrary, Eq.~\eqref{eq:s_n_0} is shown.

Next, from Eqs.~\eqref{eq:s_n3} and Eqs.~\eqref{eq:s_n7}, the mean value theorem states that there exists a number $\hat{s}_{n_k} \in \left(0, s_{n_k}^\star\right)$, for each $k\in \mathbb{N}$,  such that
\begin{align}
-\left.\frac{ \partial^2 E_0\left(s\right)}{\partial s^2}\right|_{s=\hat{s}_{n_k}}
&= \frac{ -H + (H + \delta_{n_k})}{ s_{n_k}^\star} = \frac{\delta_{n_k}}{s_{n_k}^\star}. \label{eq:s_n4}
\end{align}
When $k$ approaches infinity, items \ref{E0_SW-a} and \ref{E0_SW-e} in Proposition \ref{prop:E0_SW+E0down}  give
\begin{align} \label{eq:s_n5}
\lim_{k\to+\infty} \left.\frac{\partial^2 E_0\left(s\right)}{\partial s^2}\right|_{s = \hat{s}_{n_k} }
= \left.\frac{\partial^2 E_0\left(s \right)}{\partial s^2}\right|_{s = 0 }
=  -V.
\end{align}
Combining Eqs.~\eqref{eq:s_n4} and \eqref{eq:s_n5} leads to
\begin{align} \label{eq:s_nc}
\lim_{k\to+\infty} \frac{ s_{n_k}^\star }{\delta_{n_k}} = \frac{1}{V }.
\end{align}
Since the subsequence was arbitrary, the above result establishes Eq.~\eqref{eq:spChc33}.

Finally, denote by 
\begin{align} \label{eq:Upsilon2}
\Upsilon = \max_{ s\in[0,1] } 
\left| \frac{\partial^3 {E}_0  \left(s \right)}{\partial s^3} \right| < +\infty.
\end{align}
For every sufficiently large $n\geq N_0$, we apply Taylor's theorem to the map $s_n^\star \mapsto E_0 \left( s_n^\star \right)$ at the original point to obtain
\begin{align}
E_\textnormal{sp}  \left(H+\delta_n\right)
&= s_n^\star\left( H + \delta_n \right) + E_0\left( s_n^\star \right) \\
&= s_n^\star  \delta_n - \frac{(s_n^\star)^2}{2} V  + \frac{(s_n^\star)^3}{6} \left.\frac{ \partial^3 E_0(s,P_n)}{ \partial s^3}\right|_{s=\bar{s}_n}  \\
&\leq s_n^\star \left( H + \delta_n - H   \right) - \frac{(s_n^\star)^2}{2} V  + \frac{(s_n^\star)^3 \Upsilon}{6} \label{eq_D16} \\
&\leq \sup_{s\geq 0} \left\{ s \delta_n - \frac{s^2}{2} V  \right\} +  \frac{(s_n^\star)^3 \Upsilon}{6}  \\
&= \frac{\delta_n^2}{2V} + \frac{(s_n^\star)^3 \Upsilon}{6}, \label{eq_D17}
\end{align}
where $\bar{s}_n$ is some number in $\left[0,s_n^\star\right]$.
Then, Eqs.~\eqref{eq:spChc33}, \eqref{eq:s_n_0}, \eqref{eq_D17}, and the assumption $\lim_{n\to+\infty}\delta_n=0$ imply that the desired inequality
\begin{align}
\limsup_{n\to+\infty} \frac{ E_\textnormal{sp} \left( H + \delta_n \right) }{\delta_n^2}
&\leq  \frac{1}{2 V }. 
\end{align}
\end{proof}



\end{document}